\documentclass[a4paper,oneside,fleqn,12pt]{article}

\usepackage{amsmath,amssymb,amsfonts,amsthm,type1cm,bm,color,mathrsfs}
\usepackage{graphicx,psfrag,epsf}
\usepackage{grffile}
\usepackage{natbib}%
\usepackage{rotating}
\usepackage{authblk}
\RequirePackage[colorlinks,citecolor=blue,urlcolor=blue]{hyperref}
\usepackage{subcaption}
\usepackage{setspace}
\usepackage{tikz}
\usetikzlibrary{positioning,fit,arrows.meta}
\usepackage[T1]{fontenc}
\usepackage{multibib}
\newcites{suppl}{References for Supplementary Material}
\usepackage{here}

\newcommand*\diff{\mathop{}\!\mathrm{d}}
\DeclareMathOperator*\argmin{\arg\min}
\DeclareMathOperator*\plim{plim}

\DeclareMathOperator\supp{supp}

\DeclareMathOperator\E{\mathbb{E}}
\DeclareMathOperator\Var{Var}

\DeclareMathOperator\Pro{\mathbb{P}}

\DeclareMathOperator\sgn{sgn}

\DeclareMathOperator\FDR{FDR}

\DeclareMathOperator\dFDR{dFDR}
\DeclareMathOperator\dFDP{dFDP}
\DeclareMathOperator\Po{Power}
\DeclareMathOperator\dPower{dPower}
\DeclareMathOperator\FWER{FWER}
\DeclareMathOperator\dFWER{dFWER}

\def\F{\mathrm{F}}

\def\bA{\mathbf{A}}
\def\bB{\mathbf{B}}

\def\bI{\mathbf{I}}
\def\bJ{\mathbf{J}}

\def\bM{\mathbf{M}}

\def\bR{\mathbf{R}}
\def\bS{\mathbf{S}}

\def\bU{\mathbf{U}}

\def\bW{\mathbf{W}}
\def\bX{\mathbf{X}}
\def\bY{\mathbf{Y}}
\def\bZ{\mathbf{Z}}
\def\ba{\mathbf{a}}
\def\bb{\mathbf{b}}

\def\be{\mathbf{e}}

\def\br{\mathbf{r}}

\def\bu{\mathbf{u}}
\def\bv{\mathbf{v}}

\def\bx{\mathbf{x}}
\def\by{\mathbf{y}}
\def\bz{\mathbf{z}}

\def\cD{\mathcal{D}}
\def\cE{\mathcal{E}}
\def\cF{\mathcal{F}}

\def\cH{\mathcal{H}}
\def\cI{\mathcal{I}}

\def\cS{\mathcal{S}}

\def\cZ{\mathcal{Z}}

\def\bbR{\mathbb{R}}

\def\bbQ{\mathbb{Q}}
\def\bbZ{\mathbb{Z}}
\def\bbN{\mathbb{N}}

\def\bGamma{\boldsymbol{\Gamma}}

\def\bDelta{\boldsymbol{\Delta}}
\def\bdelta{\boldsymbol{\delta}}

\def\bTheta{\boldsymbol{\Theta}}
\def\bOmega{\boldsymbol{\Omega}}
\def\bomega{\boldsymbol{\omega}}
\def\bPhi{\boldsymbol{\Phi}}
\def\bphi{\boldsymbol{\phi}}
\def\bPsi{\boldsymbol{\Psi}}
\def\bpsi{\boldsymbol{\psi}}

\def\bSigma{\boldsymbol{\Sigma}}

\def\ep{\varepsilon}
\def\bve{\boldsymbol{\varepsilon}}

\def\bzero{\mathbf{0}}

\def\what{\widehat}

\def\tT{\texttt{T}}
\def\tt{\texttt{t}}

\newtheorem{thm}{Theorem}
\newtheorem{lem}{Lemma}
\newtheorem{con}{Condition}
\newtheorem{rem}{Remark}
\newtheorem{prop}{Proposition}
\newtheorem{proc}{Procedure}

\makeatletter
\def\section{\@startsection {section}{1}{\z@}{-3.5ex plus -1ex minus-.2ex}{2.3ex plus .2ex}{\large\bf}}
\makeatother

\makeatletter
\def\subsection{\@startsection {subsection}{1}{\z@}{-3.5ex plus -1ex minus-.2ex}{2.3ex plus .2ex}{\normalsize\bf}}
\makeatother

\title{\textbf{\Large
		Discovering the Network Granger Causality \\
		in Large Vector Autoregressive Models
}}

\author[$\,\!$]{\textsc{Yoshimasa Uematsu}\thanks{Correspondence: Yoshimasa Uematsu, Department of Social Data Science, Hitotsubashi University, 2-1 Naka, Kunitachi, Tokyo 186-8601, Japan (E-mail: yoshimasa.uematsu@r.hit-u.ac.jp). 
}}
\author[$\,\!$]{\textsc{Takashi Yamagata}$^\dagger$}
\affil[$*$]{\textit{Department of Social Data Science, Hitotsubashi University}}
\affil[$\dagger$]{\textit{Department of Economics and Related Studies, University of York}}
\affil[$\dagger$]{\textit{Institute of Social Economic Research, Osaka University}}

\begin{document}

\maketitle

\begin{abstract}
This paper proposes novel inferential procedures for discovering the network Granger causality in high-dimensional vector autoregressive models. In particular, we mainly offer two multiple testing procedures designed to control the false discovery rate (FDR). The first procedure is based on the limiting normal distribution of the $t$-statistics with the debiased lasso estimator. The second procedure is its bootstrap version. We also provide a robustification of the first procedure against any cross-sectional dependence using asymptotic e-variables. Their theoretical properties, including FDR control and power guarantee, are investigated. The finite sample evidence suggests that both procedures can successfully control the FDR while maintaining high power. Finally, the proposed methods are applied to discovering the network Granger causality in a large number of macroeconomic variables and regional house prices in the UK.  
\end{abstract}

\section{Introduction}
Revealing a dynamic interrelationship among variables is a critical challenge in economics, finance, neuroscience, genomics, and so forth. 
In particular, identifying the ability of a time series variable to predict the future values of other time series variables has been of great interest in various fields.
Following the vast literature after \cite{Granger1969}, when the past values of a time series $x:=\{x_t\}$ can predict the future values of another time series $y$, this will be expressed as ``$x$ is \textit{Granger-causal} for $y$'' in this article. 
This Granger causality has conventionally been discussed in a bivariate relationship, typically with (bivariate) vector autoregressive (VAR) models; see \cite{Sims1972} and \cite{Hosoya1977}, among many others.

\subsection{Network Granger causality}\label{subsec:NGC}

Suppose that an $N$-dimensional time series $\by$ follows the stationary VAR($K$) model:
\begin{align*}
\by_t=\bPhi_1 \by_{t-1}+ \cdots +\bPhi_K \by_{t-K} + \bu_t,
\end{align*}
where $\bPhi_k=(\phi_{ij,k})$ is the $k$th coefficient matrix and $\bu_t$ is the error vector. The predictability relationship among the $N$ time series manifests as the sparsity pattern of the coefficient matrices in the VAR($K$) model. The \textit{Granger-causal network} is defined as the graph $G=(V,E)$ with vertex set $V=[N]:=\{1,\dots,N\}$ and edge set $E$ such that for distinct $i,j\in V$, $i\to j \in E$ if and only if $\phi_{ji,k}\not=0$ for some $k\in[K]$. 
The same definition has been adopted by \cite{BasuShojaieMicha2015} and \cite{Eichler2007,Eichler2012}, for instance. See also \cite{Shojaie2021} for recent advances in the analysis of network Granger causality.
As an illustration, we consider the 4-dimensional VAR(2) model with the sparsity pattern of the coefficient matrices shown in Figure \ref{figure:gcn}(a), where the non-zero and zero elements are indicated as the gray and white cells, respectively. From this sparse structure, we readily obtain the associated Granger-causal network in Figure \ref{figure:gcn}(b).

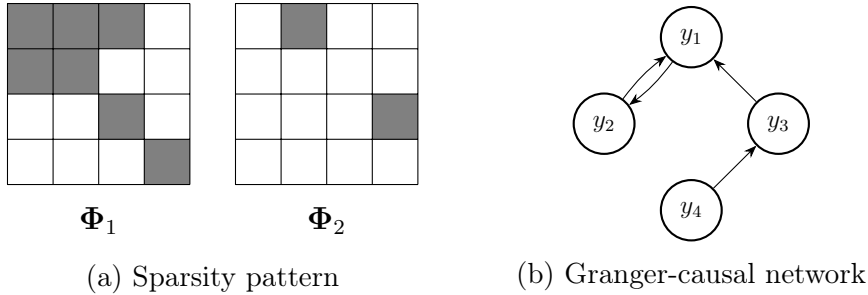
\begin{figure}[h!]
	\centering
	\begin{minipage}{0.4\textwidth}
		\centering
		\begin{tikzpicture}[scale=0.6]
		\fill[gray] (0,4) -- (1,4) -- (1,3) -- (0,3) -- cycle;
		\fill[gray] (1,3) -- (2,3) -- (2,2) -- (1,2) -- cycle;
		\fill[gray] (2,2) -- (3,2) -- (3,1) -- (2,1) -- cycle;
		\fill[gray] (3,0) -- (3,1) -- (4,1) -- (4,0) -- cycle;
		\fill[gray] (1,4) -- (2,4) -- (2,3) -- (1,3) -- cycle;
		\fill[gray] (2,4) -- (3,4) -- (3,3) -- (2,3) -- cycle;
		\fill[gray] (0,3) -- (1,3) -- (1,2) -- (0,2) -- cycle;
		\draw (0,0) grid (4,4);
		\draw (2, 0)+(0,-0.8cm) node {$\bPhi_1$};
		\fill[gray] (6,4) -- (7,4) -- (7,3) -- (6,3) -- cycle;
		\fill[gray] (8,2) -- (9,2) -- (9,1) -- (8,1) -- cycle;
		\draw (5,0) grid (9,4);
		\draw (7, 0)+(0,-0.8cm) node {$\bPhi_2$};
		\end{tikzpicture}
		\subcaption{Sparsity pattern}
	\end{minipage}
	\begin{minipage}{0.4\textwidth}
		\centering
		\begin{tikzpicture}[
		state/.style={draw,circle,inner xsep=0cm,minimum size=1.0cm,scale=0.8,thick}
		]
		\node[state] (1t) {$y_{1}$};
		\node[state] (2t) [below left=0.8cm of 1t] {$y_{2}$};
		\node[state] (3t) [below right=0.8cm of 1t] {$y_{3}$};
		\node[state] (4t) [below right=0.8cm of 2t] {$y_{4}$};
		\draw[-Stealth]
		(1t) edge[bend left=8] node {}(2t);
		\draw[-Stealth]
		(2t) edge[bend left=8] node {}(1t);
		\draw[-Stealth]
		(3t) -- (1t);
		\draw[-Stealth]
		(4t) -- (3t);
		\end{tikzpicture}
		\subcaption{Granger-causal network}
	\end{minipage}
	\caption{Sparse VAR$(2)$ coefficients and associated Granger-causal network}
	\label{figure:gcn}
\end{figure}


Although the network Granger causality has been simply defined as the \textit{direct} causal network as above, we have to remark the existence of further \textit{indirect} causalities over the lagged variables. For the aforementioned VAR(2) model, Figure \ref{figure:gcn2} depicts the causal chain that takes all the lags into account \citep[Sec.\ 3]{Eichler2012book}, where the chain continues to the infinite past and future due to the stationarity. In this figure, we can identify many indirect causalities; for instance, $y_{3,t-2}$ is indirectly causal to $y_{2t}$ via $y_{1,t-1}$. 
To refine the concept of such direct and indirect causalities, \cite{DufourRenault1998} introduced short- and long-run causality called the $h$-step ($h\geq 1$) causality. In this context, the $h$-step non-causality is characterized by zero restrictions to a part of the first $N$ rows of the first to $h$th powers of the companion coefficient matrices \citep[p.50]{Lutkephol2005}. Indeed, in Figure \ref{figure:gcn2}, it can be shown that the one-step (i.e., direct) causality is expressed by the thick arrows while the indirect effects will be included in $h$-step ($h\geq 2$) causalities. Apparently, this one-step causality corresponds to the network Granger causality we have defined. 




\graphicspath{ {./images/} }
\begin{figure}[h!]
	\centering
	\begin{tikzpicture}[every node/.style={draw,circle,inner xsep=0cm,minimum size=1.0cm,scale=0.8}]
	\clip (1.3,-4.0) rectangle (7.8,1.0);
	\node (1t) {$y_{1}^{t+2}$};
	\node[below=0.3cm of 1t] (2t) {$y_{2}^{t+2}$};
	\node[below=0.3cm of 2t] (3t) {$y_{3}^{t+2}$};
	\node[below=0.3cm of 3t] (4t) {$y_{4}^{t+2}$};
	\node[right=0.7cm of 1t] (1t-1) {$y_{1}^{t+1}$};
	\node[right=0.7cm of 2t] (2t-1) {$y_{2}^{t+1}$};
	\node[right=0.7cm of 3t] (3t-1) {$y_3^{t+1}$};
	\node[right=0.7cm of 4t] (4t-1) {$y_4^{t+1}$};
	\node[right=0.7cm of 1t-1] (1t-2) {$y_1^{t}$};
	\node[right=0.7cm of 2t-1] (2t-2) {$y_2^{t}$};
	\node[right=0.7cm of 3t-1] (3t-2) {$y_3^{t}$};
	\node[right=0.7cm of 4t-1] (4t-2) {$y_4^{t}$};
	\node[right=0.7cm of 1t-2] (1t-3) {$y_1^{t-1}$};
	\node[right=0.7cm of 2t-2] (2t-3) {$y_2^{t-1}$};
	\node[right=0.7cm of 3t-2] (3t-3) {$y_3^{t-1}$};
	\node[right=0.7cm of 4t-2] (4t-3) {$y_4^{t-1}$};
	\node[right=0.7cm of 1t-3] (1t-4) {$y_1^{t-2}$};
	\node[right=0.7cm of 2t-3] (2t-4) {$y_2^{t-2}$};
	\node[right=0.7cm of 3t-3] (3t-4) {$y_3^{t-2}$};
	\node[right=0.7cm of 4t-3] (4t-4) {$y_4^{t-2}$};
	\node[right=0.7cm of 1t-4] (1t-5) {$y_1^{t-3}$};
	\node[right=0.7cm of 2t-4] (2t-5) {$y_2^{t-3}$};
	\node[right=0.7cm of 3t-4] (3t-5) {$y_3^{t-3}$};
	\node[right=0.7cm of 4t-4] (4t-5) {$y_4^{t-3}$};
	\node[right=0.7cm of 1t-5] (1t-6) {$y_1^{t-4}$};
	\node[right=0.7cm of 2t-5] (2t-6) {$y_2^{t-4}$};
	\node[right=0.7cm of 3t-5] (3t-6) {$y_3^{t-4}$};
	\node[right=0.7cm of 4t-5] (4t-6) {$y_4^{t-4}$};
	
	\draw[-Stealth, very thick](1t-3) -- (1t-2);
	\draw[-Stealth, very thick]	(1t-3) -- (2t-2);
	\draw[-Stealth, very thick]	(2t-3) -- (1t-2);
	\draw[-Stealth, very thick]	(2t-3) -- (2t-2);
	\draw[-Stealth, very thick]	(3t-3) -- (1t-2);
	\draw[-Stealth, very thick]	(3t-3) -- (3t-2);
	\draw[-Stealth, very thick]	(4t-3) -- (4t-2);
	\draw[-Stealth, very thick]	(2t-4) -- (1t-2);
	\draw[-Stealth, very thick]	(4t-4) -- (3t-2);
%
\draw[-Stealth]
(1t-1) edge (1t)
(1t-1) edge (2t)
(2t-1) edge (1t)
(2t-1) edge (2t)
(3t-1) edge (1t)
(3t-1) edge (3t)
(4t-1) edge (4t)
(2t-2) edge (1t)
(4t-2) edge (3t)

(1t-2) edge (1t-1)
(1t-2) edge (2t-1)
(2t-2) edge (1t-1)
(2t-2) edge (2t-1)
(3t-2) edge (1t-1)
(3t-2) edge (3t-1)
(4t-2) edge (4t-1)
(2t-3) edge (1t-1)
(4t-3) edge (3t-1)

(1t-3) edge (1t-2)
(1t-3) edge (2t-2)
(2t-3) edge (1t-2)
(2t-3) edge (2t-2)
(3t-3) edge (1t-2)
(3t-3) edge (3t-2)
(4t-3) edge (4t-2)
(2t-4) edge (1t-2)
(4t-4) edge (3t-2)

(1t-4) edge (1t-3)
(1t-4) edge (2t-3)
(2t-4) edge (1t-3)
(2t-4) edge (2t-3)
(3t-4) edge (1t-3)
(3t-4) edge (3t-3)
(4t-4) edge (4t-3)
(2t-5) edge (1t-3)
(4t-5) edge (3t-3)

(1t-5) edge (1t-4)
(1t-5) edge (2t-4)
(2t-5) edge (1t-4)
(2t-5) edge (2t-4)
(3t-5) edge (1t-4)
(3t-5) edge (3t-4)
(4t-5) edge (4t-4)
(2t-6) edge (1t-4)
(4t-6) edge (3t-4)

(1t-6) edge (1t-5)
(1t-6) edge (2t-5)
(2t-6) edge (1t-5)
(2t-6) edge (2t-5)
(3t-6) edge (1t-5)
(3t-6) edge (3t-5)
(4t-6) edge (4t-5);
	\end{tikzpicture}
	\caption{Causal chain of the VAR(2)}
\label{figure:gcn2}
\end{figure}
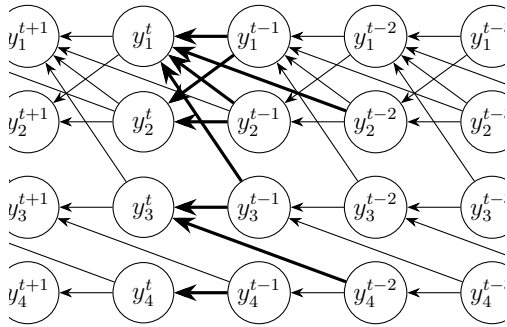

The Granger-causal network can be discovered just by estimating the sparse $\bPhi_k$'s in an appropriate manner. In early empirical studies, including \cite{FujitaEtAl2007} and \cite{LozanoEtal2009} among others, such estimation methods were already employed and a network was estimated. Recently the theoretical properties of these methods have begun to be investigated. \cite{BasuMichailidis2015}, \cite{KockCallot2015}, and \cite{BasuEtAl2019} study properties of the lasso \citep{Tibshirani1996} and its variants in high-dimensional VAR models. \cite{HanLuLiu2015} extends the Dantzig selector of \cite{CandesTao2007} to be applicable to (weakly) sparse VAR models. 
\cite{KockCallot2015} and \cite{BarigozziBrownlees2019} propose the adaptive lasso estimator. \cite{DavisEtAl2016} consider an alternative two-step estimator, which uses estimates of the partial spectral coherence.

\subsection{Classical Granger causality test}

Estimation-based network detection, such as the lasso selection, is appealing because of its simplicity. However, along with the difficulty of evaluating the type I error of the selection result, this may cause instability and a lack of reproducibility. 
Thus, we discusses the inference-based network detection. Indeed, statistical inference of Granger causality was the central issue, as Sir Clive Granger himself put it: ``the problem is how to devise a definition of causality and feedback and to test for their existence'' \citep[p.428]{Granger1969}.


Regardless of the model dimensionality, it is possible to test the null of Granger non-causality if an asymptotically normal estimator is available. For illustration, consider an $N$-dimensional VAR($1$) model. Given $\cH\subset[N]\times[N]$, we may test $H_0: \phi_{ij}=0$ for all $(i,j)\in\cH$ versus $H_1: \phi_{ij}\not=0$ for some $(i,j)\in\cH$ using the asymptotic Wald statistic via the \textit{debiased lasso} estimator  \citep{ZhengRaskutti2019,ZhuLiu2020,BabiiEtAl2021}. This is effective when a very small $\cH$ is of our interest. Unfortunately, however, this approach is not suitable for discovering the high-dimensional networks in $\cH$. This is because, when $\cH$ is large, the rejection of this null hypothesis tells us only that there may exist Granger causality between \textit{some} of the $N$ variables. This is not informative for our purpose.

\subsection{Discovering the network with directional FDR control}

The network discovering is characterized as the \textit{multiple testing} for the sequence of hypotheses, $H_0^{(i,j)}: \phi_{ij}=0$ versus $H_1^{(i,j)}: \phi_{ij}\not=0$ for each $(i,j)\in\cH$. Let $\cS=\{(i,j):H_1^{(i,j)}\text{ is true}\}$ and $\hat{\cS}=\{(i,j):H_0^{(i,j)}\text{ is rejected}\}$. The \textit{false discovery rate} (FDR), considered as a measure of type I error in multiple tests, is defined as $\FDR=\E[{|\hat\cS \cap\cS^c |}/{(|\hat{\cS}|\vee 1)}]$ \citep{benjamini1995controlling}. The FDR-controlled multiple testing tends to exhibit higher power ($\Po = \E[{|\hat{\cS}\cap \cS|}/{|\cS|}]$) than that controlling the \textit{family-wise error rate} ($\FWER=\Pro (|\hat{\cS}\cap\cS^c|\geq 1)$), typically by the so-called Bonferroni correction \citep{bonferroni1935calcolo, holm1979simple}, because the FDR is smaller than or equal to the FWER. FDR control for time series is essentially difficult, and thus there are very few studies. Only exceptionally, \cite{CMChi2021} develop a quite general knockoff framework for time series by extending \cite{candes2016panning}, but we do not pursue the direction.

In this paper, we propose two novel methods for discovering the network Granger causality in high-dimensional VAR models, based on asymptotic and bootstrap $t$-statistics, respectively. We will show that they can control the \textit{directional} FDR (dFDR), which takes into account not only whether each parameter is zero but also the differences in their directions (signs),  with a high power guarantee. This will lead to a more accurate and stable discovery. 
In case a strong cross-sectional dependency in the model is anticipated, we also provide a \textit{robustification} of Procedure 1 that gives a valid FDR control under any dependence structure. 
To the best of our knowledge, this is the first paper to propose such inference-based procedures that focus on discovering network Granger causalities with theoretical guarantees. 
The finite sample evidence suggests that the procedures can successfully control the FDR while maintaining the high power. The proposed methods are applied to a large dataset of macroeconomic time series and regional house prices in the UK.  


There are some studies that consider Granger causality between groups of variables, typically using group lasso; see \cite{BasuShojaieMicha2015}, \cite{LinMichailidis2017}, and \cite{BasuEtAl2019}. This approach requires the researcher to know the group members. \cite{Gudmund2021} develop methods to statistically identify groups among the variables.

\subsection{Organization and notation}

The paper is organized as follows. 
Section \ref{sec:model} defines the VAR model. 
Section \ref{sec:inf} proposes two methods and a robustification technique for discovering the networks based on a multiple test. 
Section \ref{sec:theory} explores the statistical theory for the FDR control and power guarantee of our methods. 
Section \ref{sec:MC} confirms the finite sample validity via Monte Carlo experiments. 
Section \ref{sec:ee} applies our methods to large datasets. 
Section \ref{sec:concl} concludes. 
All the proofs and additional analyses are collected in Supplementary Material.

\textbf{Notation.}
For any matrix $\bM=(m_{ti})\in\mathbb{R}^{T\times N}$, denote by 
$\|\bM\|_{\F}$, $\|\bM\|_2$, $\|\bM\|_1$, $\|\bM\|_{\max}$, and $\|\bM\|_\infty$ the Frobenius norm, induced $\ell_2$ (spectral) norm, entrywise $\ell_1$-norm, 
entrywise $\ell_\infty$-norm, and induced $\ell_\infty$-norm, respectively. Specifically, they are defined by 
$\|\bM\|_{\F}=(\sum_{t,i}m_{ti}^2)^{1/2}$, $\|\bM\|_2=\lambda_{1}^{1/2}(\bM'\bM)$, $\|\bM\|_1=\sum_{t,i}|m_{ti}|$,  $\|\bM\|_{\max}=\max_{t,i}|m_{ti}|$, and $\|\bM\|_\infty=\max_t\sum_{i}|m_{ti}|$, where $\lambda_{i}(\bS)$ refers to the $i$th largest eigenvalue of a symmetric matrix $\bS$. 
Denote by $\bI_N$ and $\bzero_{T\times N}$ the $N\times N$ identity matrix and $T\times N$ matrix with all the entries being zero, respectively. We use $\lesssim$ ($\gtrsim$) to represent $\leq$ ($\geq$) up to a positive constant factor. For any positive sequence $a_n$ and $b_n$,  we write $a_n \asymp b_n$ if $a_n \lesssim b_n$ and $a_n \gtrsim b_n$. 
For any positive values $a$ and $b$, $a\vee b$ and $a\wedge b$ stand for $\max(a,b)$ and $\min(a,b)$, respectively. The indicator function is denoted by $1\{\cdot\}$. 
The sign function is defined as $\sgn(x)=x/|x|$ for $x\not=0$ and $\sgn(0)=0$.
The ceiling function is denoted as $\lceil x \rceil = \min\{n\in\bbZ: n\geq x\}$.

\section{Model}\label{sec:model}

Suppose that the $N$-dimensional vector of stationary time series, $\by_t=(y_{1t},\dots,y_{Nt})'$, is generated from the VAR($K$) model:
\begin{align}
\by_t &= \bPhi_1 \by_{t-1} + \dots + \bPhi_K \by_{t-K} + \bu_t
= \bPhi \bx_t + \bu_t,\label{VAR(K)}
\end{align}
where $\bPhi_k \in \bbR^{N\times N}$ are the \textit{sparse} coefficient matrices with $\bPhi_k=(\phi_{ij,k})$ and $\bPhi = \left(\bPhi_1, \dots, \bPhi_K \right)$, $\bx_t = (\by_{t-1}', \dots, \by_{t-K}')'\in\bbR^{KN}$ is the vector of lagged variables, and $\bu_t=(u_{1t},\dots,u_{Nt})'$ is an error vector with mean zero and finite positive definite covariance matrix $\bSigma_u=(\sigma_{ij})=\E\bu_t\bu_t'$. Let $\bSigma_x=\E \bx_t\bx_t'$ 
and $\bGamma_y(k-1)=\E\by_{t-1}\by_{t-k}'$ with $\bGamma_y(k-1)=\bGamma_y(-k+1)'$ for $k\in[K]$. Then $\bSigma_x$ is composed of submatrices $\bGamma_y(k-1)$. Moreover, define the precision matrix of $\bx_t$ as $\bOmega=(\omega_{ij})=\bSigma_x^{-1}$. 
Denote by $\bomega_{j}$ the $j$th column vector of $\bOmega$ for $j\in [KN]$. We follow the convention, $\sigma_i^2=\sigma_{ii}$ and $\omega_{i}^2=\omega_{ii}$. 
Throughout the paper, we suppose that $N\wedge T\to \infty$ with allowing $N$ to be as large as or possibly larger than $T$ while $K=o(N\wedge T)$. We also assume $K$ to be known for the sake of a concise presentation; see e.g. \cite{NicholsonEtAl2020} for a review and new lag selection method.

Stacking the $T$ observations in columns such that $\bY=(\by_1,\dots,\by_T)\in\mathbb{R}^{N\times T}$, $\bX=(\bx_1,\dots,\bx_{T})\in\mathbb{R}^{KN\times T}$, and $\bU=(\bu_1,\dots,\bu_T)\in\mathbb{R}^{N\times T}$, we have the matrix form, $\bY = \bPhi \bX + \bU$. 
Denote by $\ba_{\cdot i}$ the $i$th row vector of matrix $\bA$. Then, by extracting the $i$th row of the matrix form, the model is also written as $\by_{i\cdot} = \bphi_{i\cdot}\bX  + \bu_{i\cdot}$. 

To describe the sparsity pattern of $\bPhi$, define $\cS=\{(i,j)\in [N]\times[KN]: \phi_{ij}\not=0\}$ with cardinality $s=|\cS|$. Similarly, the sparsity pattern of  $\bphi_{i\cdot}$ is described as 
$\cS_i=\{j\in[KN]: \phi_{ij}\not=0\}$ with $s_i=|\cS_i|$ and $\bar{s}=\max_{i\in[N]}s_i$. There is a one-to-one correspondence between $\cS$ and the Granger-causal network as mentioned in Section \ref{subsec:NGC}. Our goal is to discover the network wtih controlling the FDR. 


\section{Inferential Methodology}\label{sec:inf}

We first propose two multiple testing procedures that can control the FDR of discovering the Granger-causal networks. They are based on asymptotic and bootstrap $t$-statistics, respectively. Then, we introduce robustified version of the first procedure that is valid under arbitrary cross-sectional dependence structure in the model. The associated statistical theory is developed in Section \ref{sec:theory}.

\subsection{Debiased lasso estimator}\label{ssec:dlasso}

We start with constructing the row-wise lasso estimator,  $\hat\bphi_{i\cdot}^{\normalfont{\textsf{L}}}$, defined as 
\begin{align}
\hat\bphi_{i\cdot}^{\normalfont{\textsf{L}}} = \argmin_{\bphi_{i\cdot}\in\bbR^{1\times KN}}~(2T)^{-1}\| \by_{i\cdot} - \bphi_{i\cdot} \bX \|_{2}^2 + \lambda\|\bphi_{i\cdot}\|_1,   \label{syslasso} 
\end{align}
where $ \lambda >0$ is a regularization parameter. It is well-known that the lasso estimator has a bias caused by this regularization. Following \cite{JavanmardMontanari2014}, we remove the bias, which leads to the asymptotic normality of each element; see also \cite{van2014asymptotically} and \cite{ZhangZhang14} for a related discussion. Define the debiased lasso estimator as
\begin{align}
\hat{\bphi}_{i\cdot}
= \hat\bphi_{i\cdot}^{\normalfont{\textsf{L}}} + ( \by_{i\cdot} - \hat\bphi_{i\cdot}^{\normalfont{\textsf{L}}} \bX ) \bX'\hat\bOmega/T, \label{dlasso}
\end{align}
where $\hat{\bOmega}$ is a consistent precision matrix estimator. By a simple calculation, \eqref{dlasso} is written as $\sqrt{T} (\hat\bphi_{i\cdot} - \bphi_{i\cdot} ) 
= \bz_{i\cdot} + \br_{i\cdot}$ for some $\bz_{i\cdot}$ and $\br_{i\cdot}$ such that 
for each $i\in[N]$, each element of $\bz_{i\cdot}$ can be asymptotically normal while $\|\br_{i\cdot}\|_{\max}$ becomes negligible under regularity conditions. This point is formally investigated in Section \ref{subsec:deb}.

\subsection{Multiple test}\label{ssec:ginference}
For any $\cH\subset [N]\times [KN]$, consider discovering the Granger-causal network in $\cH$. This problem is understood as the multiple test for the sequence of hypotheses: 
\begin{align}
H_0^{(i,j)}: \phi_{ij}=0 ~~~ \mbox{versus} ~~~ H_1^{(i,j)}: \phi_{ij}\not=0 
~~\text{ for each }~~ (i,j)\in\cH. \label{hypo}
\end{align}
As observed in Section \ref{ssec:dlasso}, under regularity conditions, the debiased lasso estimator will be expressed as $\sqrt{T}\hat{\phi}_{ij}=z_{ij}+o_p(1)$ under $H_0^{(i,j)}$, where $z_{ij}$ is expected to be asymptotically normal with the variance, 
$\Var(z_{ij})=\sigma_i^2\bomega_{j}'\bSigma_x{\bomega}_{j}=\sigma_i^2\omega_{j}^2$. 
Thus, the $t$-test for each pair of hypotheses in \eqref{hypo} is performed with the $t$-statistic, either
\begin{align}\label{Tstat}
\tT_{ij} = \frac{ \sqrt{T} \hat{\phi}_{ij}}{\hat{\sigma}_i\sqrt{\hat{\bomega}'_{j} \hat{\bSigma}_x\hat{\bomega}_{j}}}\text{~~~or~~~}\frac{ \sqrt{T} \hat{\phi}_{ij}}{\hat{\sigma}_i\hat{\omega}_j}, 
\end{align}
where $ \hat{\bSigma}_{x} = \bX\bX'/ T $. Hereafter, denote by $\hat{m}_{ij}$ either $\hat{\sigma}_i\sqrt{\hat{\bomega}'_{j} \hat{\bSigma}_x\hat{\bomega}_{j}}$ or $\hat{\sigma}_i\hat{\omega}_j$. 

Repeating the $t$-test over $(i,j)\in\cH$ with a critical value $\tt$ leads to a set of discoveries, $\hat{\cS}(\tt):=\{(i,j)\in\cH:|\tT_{ij}|\geq \tt\}$. Here, the choice of $\tt$ is critical; we propose two procedures to determine $\tt$ that achieves controlling the  \textit{directional FDR} of $\hat\cS(\tt)$, 
\begin{align*}
\dFDR = \E\left[ \dFDP \right],~~~
\dFDP = \frac{|\{(i,j)\in\hat{\cS}(\tt):\sgn(\hat{\phi}_{ij})\not=\sgn(\phi_{ij})\}|}{|\hat{\cS}(\tt)|\vee 1},
\end{align*}
below a target level. The $\dFWER$ and $\dPower$ are also defined in a similar manner. 

There are several ways to construct consistent estimators of the nuisance parameters, $\sigma_i^2$ and $\bomega_{j}$. In this paper, we use $\hat{\sigma}_i^2 = \sum_{t=1}\hat{u}_{it}^2/(T-d_i)$ 
with $\hat{u}_{it}$ the $(i,t)$th element of the lasso residual matrix $\hat\bU$ and the CLIME estimator \citep{CaiEtAl2011} $\hat{\bOmega}=(\hat{\bomega}_{1},\dots,\hat{\bomega}_{KN})$, respectively. Here, $d_i=o(T)$ is a positive number for degrees of freedom adjustment. A typical choice is $d_i=\hat{s}_i$, where $\hat{s}_i=|\hat{\cS}_i^{\normalfont{\textsf{L}}}|$
with $\hat{\cS}_i^{\normalfont{\textsf{L}}}=\supp(\hat\bphi_{i\cdot}^{\normalfont{\textsf{L}}})$. 

\subsubsection{First procedure: Limiting normal distribution}

We construct the set of discoveries as follows. 
\begin{proc}\label{proc:asy}\normalfont
\begin{enumerate}
	\item Set $\bar{\tt}=\sqrt{2\log(|\cH|)- a \log\log(|\cH|)}$ for given $\cH\subset[N]\times[KN]$, where $a>3$ is an arbitrary fixed constant. 
	\item For any level $q\in[0,1]$, compute
	\begin{align}
	\tt_0 = \inf\left\{\tt\in[0,\bar{\tt}]: \frac{2|\cH|Q(\tt)}{|\hat{\cS}(\tt)|\vee 1} \leq q \right\}, \label{proc:t0}
	\end{align}
	where $|\hat{\cS}(\tt)|=\sum_{(i,j)\in\cH}1\{|\tT_{ij}|\geq \tt\}$ is the total number of discoveries in $\cH$ with critical value $\tt$, and $Q(\tt)=\Pro(\cZ>\tt)$ with $\cZ$ a standard normal random variable. 
	If \eqref{proc:t0} does not exist, set 
	\begin{align}
	\tt_0=\sqrt{2\log(|\cH|)}. \label{proc:t01}
	\end{align}
	\item Obtain the set of discoveries, $\hat{\cS}(\tt_0)=\left\{(i,j)\in\cH: |\tT_{ij}|\geq \tt_0 \right\}$.
\end{enumerate}
\end{proc}
This procedure is designed to asymptotically control the dFDR of rejected nulls, $\hat{\cS}(\tt_0)$, to be less than or equal to $q$. A similar procedure can be found in \cite{Liu2013}, \cite{JJ2019}, and \cite{UY2021}, for example. The $\dFDR$ control always implies the $\FDR$ control since $\FDR\leq \dFDR$. Threshold \eqref{proc:t01} will even control the dFWER (and hence dFDR), which occurs if it is large enough not to exist in $[0,\bar{\tt}]$. 
A theoretical justification for the $\dFDR$ control and $\dPower$ guarantee is given in Section \ref{subsec:multes}. 

\begin{rem}\normalfont
\begin{enumerate}
\item[(a)] The constant $a>3$ is required for technical reasons. In practice, we can choose an arbitrary value that is slightly larger than three, like $a=3.001$. This is not sensitive to the selection result as long as it is sufficiently small.
\item[(b)] A choice of $d_i$ effects the FDR control to some extent, but $d_i=\hat{s}_i$ seems to bring good results by simulation studies in Section \ref{sec:MC}. For other choices of $d_i$ and constructions of consistent estimator of $\sigma_i^2$, see \cite{Reid2016}.
\end{enumerate}
\end{rem}

\subsubsection{Second procedure: Bootstrapped distribution}

Even in a low-dimensional setting, \cite{Bruggemann2016} point out that the finite sample properties of asymptotic VAR inference can be rather poor, and the use of bootstrap methods is often advocated. See also \cite{Kilian1999}, \cite{GoncalvesKilian2004}, \cite{HafnerHerwartz2009}. 
To improve the performance, we propose a second method based on the \textit{fixed-design wild bootstrap} (FWB). This is different from a conventional bootstrapped $t$-test for a single hypothesis. Roughly speaking, it attempts to replicate a bootstrapped distribution under the null by the FWB, which is then substituted for the limit normal distribution $Q$ in Procedure \ref{proc:asy}.

Let $\tilde{\cS}\subset [N]\times[KN]$ denote an index set of discoveries that is expected to satisfy $\Pro(\tilde{\cS}\supset \cS)\to 1$. A typical choice of  $\tilde\cS$ is the lasso-selected variables, $\hat{\cS}_{\normalfont{\textsf{L}}}$. We may also use $\hat{\cS}(\tt_0)$ obtained by Procedure \ref{proc:asy}. 

\begin{proc}\label{proc:boot}\normalfont
\begin{enumerate} 
\item Obtain a bootstrap version of the $t$-statistics, $\{\tT_{ij}^{*(b)}:(i,j)\in\tilde{\cS}^c\cap \cH\}_{b=1}^B$,  by repeating (a)--(f) $B$ times: 
\begin{enumerate}
\item Generate $\{\zeta_t\}$, a sequence of i.i.d.\ random variables with mean zero and variance one.
\item Obtain $\{\bu_{t}^*\}_{t=1}^T$ by the FWB; i.e., $\bu_{t}^*=\hat{\bu}_{t} \zeta_t$, where $\hat{\bu}_{t}=\by_{t}-\hat{\bPhi}^{\normalfont{\textsf{L}}}\bx_t$. 
\item Generate $\{\by_{t}^{*}\}_{t=1}^T$ by $\by_{t}^{*}=\hat{\bPhi}^{\normalfont{\textsf{L}}}\bx_t+\bu_{t}^*$, and set $\bY^*=(\by_{1}^*,\dots,\by_{T}^*)$. 
\item Compute a bootstrap version of the lasso estimate $\hat{\bPhi}^{\normalfont{\textsf{L}}*}$ by $\bY^*$ and $\bX$. 
\item Construct a bootstrap version of the debiased lasso estimate, 
$\hat{\bPhi}^*=\hat{\bPhi}^{\normalfont{\textsf{L}}*} + ( \bY^* - \hat\bPhi^{\normalfont{\textsf{L}}*} \bX )\bX'\hat\bOmega/T$.
\item Construct a bootstrap version of the $t$-statistics, either
\begin{align}\label{Tstat*}
\tT_{ij}^* 
= \frac{ \sqrt{T} \hat{\phi}_{ij}^*}{\hat{\sigma}^{*}_{i}\sqrt{\hat{\bomega}'_{j} \hat{\bSigma}_x\hat{\bomega}_{j}}} \text{~~or~~} 
\frac{ \sqrt{T} \hat{\phi}_{ij}^*}{\hat{\sigma}^{*}_{i}\hat{\omega}_{j} }~~~\text{for}~~~ (i,j)\in\tilde{\cS}^c\cap \cH, 
\end{align}
where $\hat{\sigma}^{*2}_{i}$ is given by 
$\hat{\sigma}_i^{*2} = \sum_{t=1}^T\hat{u}_{it}^{*2}/(T-\hat{s}_i)$
with $\hat{\bU}^*=\bU^*-(\hat{\bPhi}^{\normalfont{\textsf{L}}*}-\hat{\bPhi}^{\normalfont{\textsf{L}}})\bX$ and $\bU^*=(\bu_{1}^*,\dots,\bu_{T}^*)$. 
\end{enumerate}
\item Compute the empirical distribution,
\begin{align*}
\bbQ_B^{*}(\tt) &= \frac{1}{|\tilde\cS^c\cap\cH|}\sum_{(i,j)\in \tilde{\cS}^c\cap \cH} \left[ \frac{1}{B}\sum_{b=1}^B1\left\{\tT_{ij}^{*(b)}>\tt\right\} \right].
\end{align*}
\item Run Procedure \ref{proc:asy} with \eqref{proc2:t0} replacing \eqref{proc:t0}:
\begin{align}
\tt_0 = \inf\left\{\tt\in[0,\bar{\tt}]: \frac{|\cH|\left\{\bbQ_B^{*}(\tt)+1-\bbQ_B^{*}(-\tt)\right\}}{|\hat{\cS}(\tt)|\vee 1} \leq q \right\}. \label{proc2:t0}
\end{align}
\end{enumerate}
\end{proc}

\begin{rem}\normalfont
\begin{enumerate}
\item[(a)] We may adopt several distributions for $\zeta_t$. For example, \cite{Mammen1993} suggests using $\zeta_t$ that takes values $\mp(\sqrt{5}\mp1)/2$ with probability $(\sqrt{5}\pm1)/(2\sqrt{5})$, respectively, while \cite{DavidsonFlachaire2008} propose using the Rademacher random variable, $\zeta_t = \pm1$ with probability $1/2$. 
\item[(b)] Unlike the construction of $\hat{\sigma}_i^2$ in Procedure \ref{proc:asy}, the degree of freedom adjustment in $\hat{\sigma}_i^{*2}$ is not sensitive to the results of FDR control. 
\end{enumerate}
\end{rem}


\subsubsection{Robustifying transformation}\label{ssec:robust}

The (directional) FDR control by the proposed two procedures requires a condition to restrict the correlation structure of $z_{ij}$'s; see Condition \ref{ass:corr} in the next section. Such a condition that is somewhat restrictive and difficult to verify if actually met in real data often tends to be avoided. As a remedy, we propose \textit{robustification} of Procedure \ref{proc:asy}, which allows any cross-sectional dependency in $z_{ij}$'s. As a trade-off for the robustness, it exhibits conservative selection in general. Thus, the procedure is recommended when a cautious analysis is needed. 

The key is to transform the $t$-statistics to an \textit{$e$-variables}, the definition of which is a non-negative random variable $E$ with $\E [E] \leq 1$ under the null \citep{VovkWang2021AoS}. The significant factor is that $1/E$ becomes a $p$-variable since $\Pro(1/E\leq q)\leq q$ for any level $q\geq 0$ by the Markov inequality, and is valid under any dependence structure. Then a robust FDR-controlled selection is achieved via the standard BH procedure of \cite{benjamini1995controlling} with $1/E$, which is called the e-BH \citep{WangRamdas2022JRSSB}. 

To introduce the e-BH procedure, we convert the double index $(i,j)\in\cH$ to an ordered single index $h\in[|\cH|]$ using some converter $h(i,j)$. For example, if $\cH=[N]\times[KN]$, we set $h=h(i,j)$ with $h(i,j)=i+N(j-1)$ for given $(i,j)\in\cH$; conversely, $(i,j)$ is recovered by setting $j=\lceil h/N \rceil$ and $i=h-N(j-1)$ for given $h\in[|\cH|]$. Suppose that $|\cH|$ $e$-variables $E_h = E_{h(i,j)}$ corresponding to $H_0^{(i,j)}$ are available. 
\begin{proc}[e-BH]\label{proc:eBH}\normalfont
\begin{enumerate}
\item Make order statistics, $E_{(1)}\geq \dots \geq E_{(|\cH|)}$. 
\item Compute $h^* = \max \left\{h\in[|\cH|]:1/E_{(h)}\leq q h/|\cH| \right\}$.
\item Obtain the set of discoveries, $\hat{\cS}_R(h^*)=\{(i,j)\in\cH: h(i,j)=h,j=\lceil h/N \rceil,E_{h}\geq E_{(h^*)}\}$. 
\end{enumerate}
\end{proc}
A critical aspect is how to construct a ``good'' $e$-variable as there are countless ways to make it. Since little is known about the optimal construction, we focus on a simple formulation of e-variables transformed from the $t$-statistics \eqref{Tstat}, which we call the \textit{robustifying transformation}. Precisely, we adopt the following: 
\begin{align}\label{e-var}
E_{h(i,j)} = \frac{f(\tT_{ij})}{\E [f(\cZ)]},~~~~~\cZ\sim N(0,1), 
\end{align}
where $f:\mathbb{R}\to \mathbb{R}_+$ is a continuous strictly increasing function. 
We may compute $\E [f(\cZ)]$ once $f$ is specified. This construction yields asymptotic $e$-variables since $\E [E_{h(i,j)}]\to 1$ is expected under $H_{0}^{(i,j)}$; whereas, under $H_{1}^{(i,j)}$, we can obtain $\E [E_{h(i,j)}] \asymp f(\sqrt{T})\to \infty$  thanks to the fact of $\tT_{ij}\approx \sqrt{T}$ and increasing $f$. In practice, we may use $f:x\mapsto |x|^p$ for some $p>0$ with $\E|\cZ|^p=\sqrt{2^p/\pi}\Gamma((p+1)/2)$ or $f:x\mapsto \exp(c|x|^\alpha)$ for some $c,\alpha>0$ with $\E [\exp(c|\cZ|)]=2\exp(c^2/2)\Phi(c)$ if \
$\alpha=1$. 
A theory is provided in Section \ref{subsubsec:robust} and Section \ref{ssec:eval} of Supplementary Material.




\section{Statistical Theory}\label{sec:theory}

We develop a formal statistical theory for the inferential methodology proposed in Section \ref{sec:inf}. Throughout this section, we set $\cH= [N]\times [KN]$, $d_i=0$, and $\tilde\cS=\hat{\cS}_{\normalfont{\textsf{L}}}$ to alleviate unnecessary technical complications. We suppose that $N,T\to \infty$, $N=O(T^d)$ for some $d>0$, and $K=o(N\vee T)$.

\begin{con}\label{ass:subG}\normalfont
	The error term, $\{\bu_t\}$, is a sequence of i.i.d.\ sub-Gaussian random vectors with mean zero and covariance matrix $\bSigma_u$; for every $N>0$, there exists some constant $c_u>0$ such that for all $x>0$,
	$\max_{i\in[N]}\Pro\left( |u_{it}| > x \right) \leq 2\exp(-x^2/c_u)$.
\end{con}
 
\begin{con}\label{ass:stab}\normalfont
All the eigenvalues of the companion matrix of $(\bPhi_1,\dots,\bPhi_K)$ are strictly less than one in modulus uniformly in $K$ and $N$. 
\end{con}

\begin{con}\label{ass:mineig}\normalfont
For all $K$ and $N$, there exists some constant $\gamma>0$ such that 
$\gamma \leq \lambda_{\min}(\bSigma_x)\leq \lambda_{\max}(\bSigma_x)\leq 1/\gamma$.
\end{con}

Conditions \ref{ass:subG}--\ref{ass:mineig} are commonly used in the literature.  
Condition \ref{ass:stab} guarantees that the VAR($K$) model is stable and is inverted to the VMA($\infty$) model:
\begin{align}\label{cond:sum0}
\by_t = \sum_{\ell=0}^\infty \bB_\ell \bu_{t-\ell}, ~~~~~b:=\sum_{\ell=0}^\infty \|\bB_\ell\|_{\infty} < \infty,
\end{align}
where $\bB_0=\bI_N$ and $\bB_\ell=\bJ'\bA^\ell\bJ$ with $\bJ'=(\bI_N, \bzero_{N\times(KN-N)})$ for $\ell=1,2,\dots$ and $\bA$ the companion matrix of $(\bPhi_1,\dots,\bPhi_K)$; see Lemma \ref{lem:VMA} in Supplementary Material. 
Throughout this section, we set the lasso regularization parameter in \eqref{syslasso} to be
\begin{align*}
\lambda=8bc_{uu}\sqrt{2(\nu+7)^3T^{-1}\log^3(N\vee T)},
\end{align*}
where $\nu>0$ is a fixed constant.

\subsection{Theory for the debiased lasso estimator}\label{subsec:deb}

The first proposition derives the nonasymptotic error bounds of the lasso estimator. 

\begin{prop}[Nonasymptotic error bounds for the lasso]\label{thm:errbound}
If Conditions \ref{ass:subG}--\ref{ass:mineig} hold, then the lasso estimator defined in \eqref{syslasso}  satisfies the following inequalities with probability at least $1-O((N\vee T)^{-\nu})$: 
\begin{align*}
(a)~&~\left\| \hat{\bphi}_{i\cdot}^{\normalfont{\textsf{L}}}-\bphi_{i\cdot}\right\|_{2} \leq \frac{12\sqrt{s_i} \lambda}{\gamma-8bs_i\lambda},\\
(b)~&~\left\| \hat{\bphi}_{i\cdot}^{\normalfont{\textsf{L}}}-\bphi_{i\cdot}\right\|_{1} \leq \frac{48s_i \lambda}{\gamma-8bs_i\lambda},\\
(c)~&~\frac{1}{T}\left\|(\hat{\bphi}_{i\cdot}^{\normalfont{\textsf{L}}}-\bphi_{i\cdot})\bX  \right\|_{2}^2 \leq \frac{144s_i \lambda^2}{\gamma-8bs_i\lambda}
\end{align*}
for all $i\in[N]$ such that $8bs_i\lambda<\gamma$.
\end{prop}

Using Proposition \ref{thm:errbound}, we next show the asymptotic linearity of the debiased lasso estimator, $\hat{\bPhi}$. This requires some conditions on the precision matrix, $\bOmega$.

\begin{con}\normalfont\label{ass:invest}
There exist some positive numbers $s_\omega$ and $M_\omega$ and some constant $r\in[0,1)$ such that 
$\max_{i\in[KN]}\sum_{j=1}^{KN}|\omega_{ij}|^r\leq s_{\omega}$ and  $\max_{j\in[KN]}\|{\bomega}_{j}\|_{1}\leq M_\omega$, where $s_\omega$ and $M_\omega$ can diverge as $N,T\to \infty$. 
\end{con}
Condition \ref{ass:invest} has frequently been used to derive the rate of convergence of the CLIME estimator $\hat{\bOmega}$; see e.g., \cite{CaiEtAl2011} and \cite{ShuNan2019}. 
This condition requires that $\bOmega$ is approximately sparse. Especially when $r=0$, the condition reduces to the exact sparsity assumption. 
Denote by $\{\be_j\}$ the standard basis of $\bbR^{KN}$.

\begin{prop}[Asymptotic linearity of the debiased lasso]\label{thm:asynormal}
The debiased lasso estimator defined in \eqref{dlasso} has the representation $\sqrt{T}( \hat\bPhi - \bPhi) = \bZ + \bR$, where the $(i,j)$th elements of $\bZ$ and $\bR$ are respectively given by
\begin{align*}
z_{ij} = \frac{1}{\sqrt{T}}\sum_{t=1}^Tu_{it}\bx_t'{\bomega}_{j}, ~~~
r_{ij} = \frac{1}{\sqrt{T}}\bu_{i\cdot}\bX^{\prime}(\hat{\bomega}_{j}-{\bomega}_{j})
- \sqrt{T}( \hat\bphi_{i\cdot}^{\normalfont{\textsf{L}}} - \bphi_{i\cdot} )(\hat{\bSigma}_{x}\hat{\bomega}_{j}-\be_{j}).
\end{align*}
Furthermore, if Conditions \ref{ass:subG}--\ref{ass:invest} are assumed, then the following inequality holds with probability at least $1 - O ((N\vee T)^{-\nu})$:
\begin{align}
\max_{j\in[KN]}|r_{ij}|
\lesssim \left(s_\omega M_\omega^{2-2r}\lambda^{1-r}+M_\omega s_i\lambda\right)\sqrt{\log^3(N\vee T)}=:\bar{r}_i \label{ribar}
\end{align}
for all $i\in[N]$ such that $s_i\lambda=o(1)$. 
\end{prop}
To achieve the asymptotic normality of $\sqrt{T}(\hat{\phi}_{ij}-\phi_{ij})$, we need $\bar{r}_i=o(1)$. Then $z_{ij}$ becomes the dominating term, which can converge in distribution to a Gaussian random variable. Note that $\bar{r}_i\asymp (s_\omega^3+s_\omega s_i)\ell_{NT}/\sqrt{T}$, where $\ell_{NT}$ is some power of $\log(N\vee T)$, if $r=0$ in Condition \ref{ass:invest}.

\subsection{Theory for the multiple testing}\label{subsec:multes}

We are now ready to develop a statistical theory of the multiple testing of \eqref{hypo} by Procedures \ref{proc:asy}--\ref{proc:eBH} in Section \ref{ssec:ginference}. We first establish the asymptotic normality of the $t$-statistics defined in \eqref{Tstat}. 

\begin{con}\label{ass:mineig2}\normalfont
For all $K$ and $N$, there exists some constant $\gamma>0$ such that 
$\gamma\leq \min_{i\in[N]}\sigma_i^2\leq \max_{i\in[N]}\sigma_i^2\leq 1/\gamma$ and 
$\gamma\leq \min_{j\in[KN]}\omega_{j}^2\leq \max_{j\in[KN]}\omega_{j}^2\leq 1/\gamma$.
\end{con}
Condition \ref{ass:mineig2} complements Condition \ref{ass:mineig}, and is required to deal with the standard errors; see Lemma \ref{lem:omegahat1} in Supplementary Material. 


\begin{thm}[Asymptotic normality of $t$-statistic]\label{thm:t}
If Conditions \ref{ass:subG}--\ref{ass:mineig2} hold, then either $t$-statistics $\tT_{ij}$ defined in \eqref{Tstat} satisfies that  $\tT_{ij}-\sqrt{T}\phi_{ij}/\hat{m}_{ij}$ converges in distribution to $N(0,1)$ for all $(i,j)\in\cH$ such that $\bar{v}_i:=\bar{r}_i+M_\omega^2 \lambda = o(1)$.  
\end{thm}

In the theorem, $\bar{v}_i= o(1)$ makes both $\bar{r}_i$ in \eqref{ribar} and the estimation error of the standard error, $M_\omega^2 \lambda$, asymptotically negligible. If $r=0$ in Condition \ref{ass:invest}, it reduces to $(s_\omega^6+s_\omega^2 s_i^2)\ell_{NT}\ll T$.

\subsubsection{Theory for the FDR control}

For $\left((i,j),(k,\ell)\right)\in\cH\times \cH=:\cH^2$, the correlation between $z_{ij}$ and $z_{k\ell}$ is given by
\begin{align}
\rho_{(i,j),(k,\ell)} := \textrm{Corr}(z_{ij},z_{k\ell}) = \frac{\sigma_{ik}\omega_{j\ell}}{\sigma_i\sigma_k\omega_{j}\omega_{\ell}},
\end{align}
which becomes $|\rho_{(i,j),(k,\ell)}|\asymp |\sigma_{ik}\omega_{j\ell}|$ under Condition \ref{ass:mineig2}. Obviously, we have $\rho_{(i,j),(i,j)}=1$ for all $(i,j)\in\cH$, but the ``off-diagonal'' correlations in 
\begin{align*}
\cH_{\text{off}}^2 := \left\{((i,j),(k,\ell))\in\cH^2: (i,j)\not=(k,\ell) \right\}
\end{align*}
with $|\cH_{\text{off}}^2|=|\cH^2|-|\cH|$ take non-trivial values. To manipulate them, impose the following condition. 

\begin{con}\normalfont\label{ass:corr}
There exists a partition of $\cH_{\text{off}}^2$ denoted as $\cH_{\text{off}}^2=\cH_{w}^2\cup \cH_{s}^2$  such that for some constant $c>0$ and $\bar{\rho}:=1-\ep$ with arbitrary constant $\ep\in(0,1)$,
\begin{align*}
|\rho_{(i,j),(k,\ell)}| 
\in 
\begin{cases}
\left[0,c/(\log N)^2\right] & \text{for } ~((i,j),(k,\ell)) \in \cH_{w}^2~~~ \text{(weak correlations)}, \\
\left(c/(\log N)^2,\bar{\rho}\right] & \text{for } ~((i,j),(k,\ell)) \in \cH_{s}^2~~~ \text{(strong correlations)},
\end{cases}
\end{align*}
where $|\cH_{w}^2|=|\cH_{\text{off}}^2|-|\cH_{s}^2|$ 
and $|\cH_{s}^2|=O\left(|\cH_{\text{off}}^2|/(\log N)^2\right)$.
\end{con}
Condition \ref{ass:corr} says that most of the correlations are ``weak'' so that they are close to zero, while some are allowed to be ``strong'' enough to be non-vanishing. Because $\tT_{ij}$'s are asymptotically normal as shown in Theorem \ref{thm:t}, Condition \ref{ass:corr} ensures that most of $\tT_{ij}$'s are asymptotically independent. This condition can still be satisfied by many VAR models. To see this, note that $|\rho_{(i,j),(k,\ell)}|\asymp |\sigma_{ik}\omega_{j\ell}|$ and $\sigma_{ik},\omega_{j\ell}=O(1)$ under Condition 5. Therefore, for example, if either $\bSigma_u=(\sigma_{ik})$ or $\bOmega=(\omega_{j\ell})$ is approximately sparse, Condition \ref{ass:corr} can be satisfied regardless of the structure of the other. Here, neither $\bSigma_u$ nor $\bOmega$ needs to be exactly sparse.


\begin{thm}[FDR control: normal distribution]\label{thm:fdr-t}
Suppose $\nu>4$ and $\max_{i\in[N]}\bar{v}_i=O(T^{-\kappa_1})$ for some constant $\kappa_1\in(0,1/2)$. If Conditions \ref{ass:subG}--\ref{ass:corr} hold, then for any predetermined level $q\in[0,1]$, Procedure \ref{proc:asy} with either $t$-statistics in \eqref{Tstat} achieves the following: If $\tt_0$ is given by \eqref{proc:t0}, the obtained $\hat{\cS}(\tt_0)$ satisfies
$\limsup_{N,T\to\infty}\dFDR \leq q$ and 
$\lim_{N,T\to\infty} \Pro\left(\dFDP \leq q + \ep \right) =1$ for any $\ep>0$. 
If $\tt_0$ is given by \eqref{proc:t01}, the obtained $\hat{\cS}(\tt_0)$ satisfies
$\limsup_{N,T\to\infty} \dFWER \leq q$.
\end{thm}

Recall $\bar{v}_i=\bar{r}_i+M_\omega^2 \lambda $ defined in Theorem \ref{thm:asynormal}. The condition, $\max_{i\in[N]}\bar{v}_i=O(T^{-\kappa_1})$, reduces to 
$(s_\omega^6+s_\omega^2 \bar{s}^2)^{1/(1-2\kappa_1)}\ell_{NT}\ll T$, where $\bar{s}=\max_{i\in[N]}s_i$, if $r=0$ in Condition \ref{ass:invest}. 

We next show the dFDR control of Procedure \ref{proc:boot} with setting $\tilde{\cS}=\hat{\cS}_{\normalfont{\textsf{L}}}$. For this purpose, we need a distributional assumption on the wild bootstrap. 

\begin{con}\label{ass:wildboot}
$\{\zeta_t\}$ is a sequence of i.i.d.\ sub-Gaussian random variables with $\E \zeta_t = 0$ and $\E \zeta_t^2 =1$, where for every $T>0$, there exists some constant $c_\zeta>0$ such that for all $x>0$, $\max_{t\in[T]}\Pro\left( |\zeta_{t}| > x \right) \leq 2\exp(-x^2/c_\zeta)$.
\end{con}

Define the bootstrap distribution as
\begin{align*}
\bbQ^{*}(\tt) &= \frac{1}{|\hat{\cS}_{\normalfont{\textsf{L}}}^c|}\sum_{(i,j)\in \hat{\cS}_{\normalfont{\textsf{L}}}^c} \Pro^*\left( \tT_{ij}^{*}>\tt \right)
~~\text{with}~~
\Pro^*( \tT_{ij}^{*}>\tt )=\plim_{B\to \infty}\frac{1}{B}\sum_{b=1}^B1 \{\tT_{ij}^{*(b)}>\tt \},
\end{align*}
where $|\hat{\cS}_{\normalfont{\textsf{L}}}^c|=KN^2-\hat{s}$. 
Let $\bar{\mu}= \max_{i\in[N]}\bar{v}_i + s_\omega M_\omega^{3-2r}\lambda^{1-r} \log^2(N\vee T)$.

\begin{thm}[FDR control: Bootstrap]\label{thm:bootstrap}
Suppose  $\nu>4$ and $\bar{\mu}=O(T^{-\kappa_1})$ for some constant $\kappa_1\in(0,1/2)$. 
If Conditions \ref{ass:subG}--\ref{ass:wildboot} hold, then for any preassigned level $q\in[0,1]$, Procedure \ref{proc:boot} with either $t$-statistics in \eqref{Tstat*} achieves the same results as Theorem \ref{thm:fdr-t}. 
\end{thm}

The condition, $\bar{\mu}=O(T^{-\kappa_1})$, reduces to $(s_\omega^8+s_\omega^2 \bar{s}^2)^{1/(1-2\kappa_1)}\ell_{NT}\ll T$ if $r=0$ in Condition \ref{ass:invest}. 

To achieve Theorem \ref{thm:bootstrap}, we should show that $\tt_0$ using $\bbQ^*$ is asymptotically the same as using $Q$. Towards this goal, it suffices to verify that the event, $|\bbQ^*(\tt)/Q(\tt)-1|$ and $|\{1-\bbQ^*(-\tt)\}/Q(\tt)-1|$ converge to zero uniformly in $\tt\in[0,\bar{\tt}]$, occurs with high probability. This proof will complete in two steps. First, show that $\tT_{ij}^*$ with the ``sandwich'' s.e.\ can be approximated by the self-normalized sum, $\sum_{t=1}^T\hat{u}_{it}^*\bx_t'\hat{\bomega}_j/\sum_{t=1}^T(\hat{u}_{it}^*\bx_t'\hat{\bomega}_j)^2$. Then, verify that it can be uniformly normally approximated in the relative error using the Cram\'er-type large deviation theory of \cite{JingShaoWang2003}. The additional cost $s_\omega M_\omega^{3-2r}\lambda^{1-r} \log^2(N\vee T)$ in $\bar{\mu}$ is due to the approximation error in the first step. 


\subsubsection{Theory for the power guarantee}

We next investigate the asymptotic power of our Procedures \ref{proc:asy} and \ref{proc:boot}. A condition on the signal strength is required to distinguish the nonzero elements from zeros. 
\begin{con}\label{ass:betamin} For $\cS\subset \cH$, it holds that $\min_{(i,j)\in\cS}|\phi_{ij}|/(\sigma_i\omega_{ j}) \geq 4\sqrt{2\log(KN^2)/T}$. 
\end{con}

\begin{thm}[Power guarantee]\label{thm:power-t}
	Suppose $\nu>4$ and $\max_{i\in[N]}\bar{v}_i=O(T^{-\kappa_1})$ for some constant $\kappa_1\in(0,1/2)$. If Conditions \ref{ass:subG}--\ref{ass:mineig2} and \ref{ass:betamin} hold, then both Procedures \ref{proc:asy} and \ref{proc:boot} achieve $\lim_{N,T\to\infty}\dPower = 1$.
\end{thm}


This guarantees that both Procedures \ref{proc:asy} and \ref{proc:boot} do not asymptotically miss any important relation between variables, whichever threshold of either \eqref{proc:t0} or \eqref{proc:t01} is selected. 

\subsubsection{Theory for the robustification}\label{subsubsec:robust}

As noted in Section \ref{ssec:robust}, Condition \ref{ass:corr} sometimes seems annoying though it is crucial for Procedures \ref{proc:asy} and \ref{proc:boot}. Even without the condition, Procedure \ref{proc:eBH} with transformed statistics \eqref{e-var} can still provide reasonable discoveries.

\begin{thm}\label{thm:rob_fdr_pwr}
If Conditions \ref{ass:subG}--\ref{ass:mineig2} and $\bar{v}=o(1)$ hold. 
For a given level $q\in[0,1]$, the set of discoveries, $\hat{\cS}_R(h^*)$, obtained by Procedure \ref{proc:eBH} with e-variables \eqref{e-var} satisfies the following:
\begin{enumerate}
\item [(a)] For any continuous function $f:\bbR\to\bbR_+$ such that 
$f(\tT_{ij})$ is uniformly integrable (UI) for every $(i,j)\in \cS^c$, we have $\limsup_{N,T\to\infty}\FDR \leq q$. 
\item [(b)]For any monotonically increasing continuous function $f:\bbR\to\bbR_+$ such that $f(\sqrt{T})/|\cH|\to\infty$ for every $(i,j)\in \cS$, we have $\lim_{N,T\to\infty}\Po = 1$. 
\end{enumerate}
\end{thm}

It is not difficult to choose $f$ that admits the condition of (b). As for (a), the UI of $f(\tT_{ij})$ (as a random sequence indexed by $(T,N,K)$ for every $(i,j)\in\cS^c$) is pivotal for our e-based robust FDR control. For instance, we can show that $|\tT_{ij}|^p$ and $\exp(c|\tT_{ij}|^\alpha)$ are UI for some  constants $p,c,\alpha>0$; see Section \ref{ssec:eval} of Supplementary Material. 

The theory for e-based robustification still has room for exploration. First, Theorem \ref{thm:rob_fdr_pwr} does not investigate the directional FDR and power. Second, an optimal functional form of $f$ is not well understood. These issues will be addressed as future research topics.

\section{Monte Carlo Experiments\label{sec:MC}}

We investigate the finite sample behavior of the proposed procedures. 

\subsection{First experiment: Stability}

We begin with a small experiment to illustrate the superiority of our approach to those by the (adaptive) lasso in terms of selection stability. 
An $N\times1$ vector ${\by}_{t}$ for $t=-50,\dots,T$ is generated from the stationary VAR(1) model for $N=20$ and $T=100$ with $\mathbf{u}_{t}\sim \text{i.i.d.}N(\mathbf{0},\mathbf{I}_{N})$. The sparse $\bPhi$ is generated as in
Section \ref{subs:design} in such a way that the diagonal and $j$-diagonals for $j=\pm 1,\pm 2$, are non-zero. 

Figure \ref{figure:prelim} shows heatmaps of the frequencies with which the ($i,j$)th element of $\bPhi$ is discovered as nonzero by different methods in 1000 replications. Figure \ref{figure:prelim}(a) shows a heatmap of the absolute values of elements in $\bPhi_1$,
and Figures \ref{figure:prelim}(b), (c), and (d) report the results using the lasso, adaptive lasso, and Procedure \ref{proc:asy} with the target FDR level $q=0.2$, respectively. The overall frequency of false positives by the lasso is very high, and for certain elements unreasonably so (shown by dark blue). The adaptive lasso reduces overall false positives, but cannot do so sufficiently for the elements wrongly selected in large numbers by lasso. In addition, some nonzero $\phi_{ij}$'s are selected too infrequently by lasso, which is largely inherited by adaptive lasso (shown by very weak red). 
Meanwhile, when Procedure \ref{proc:asy} is used, the frequency of false positives is evenly spread among the elements and well controlled. Furthermore, nonzero elements are selected more evenly than the (adaptive) lasso.
These suggest that the proposed inferential methods can provide more stable and reproducible results for discovering the network Granger causality than existing popular estimation methods.

\graphicspath{ {./images/} }
\begin{figure}[h!]
	\centering
	\begin{minipage}{0.24\hsize}
		\centering
		\includegraphics[width=1.00\linewidth]{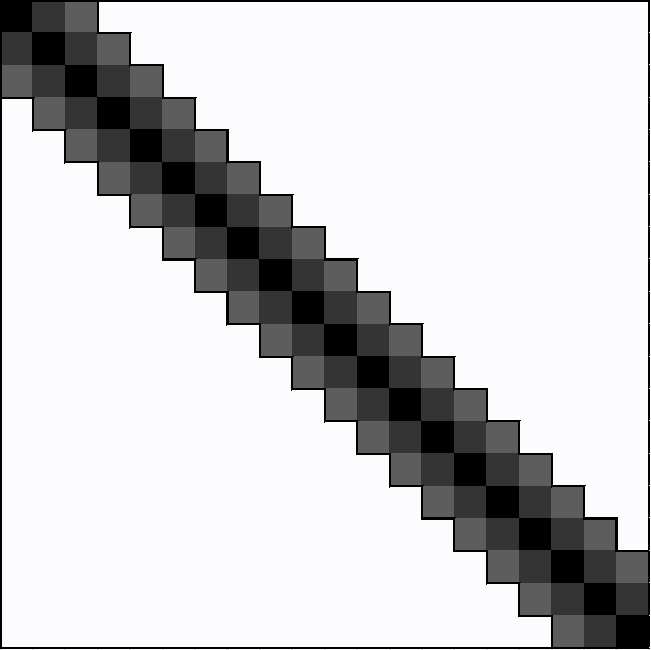}
		\subcaption{\footnotesize{$\bPhi_1$ absolute value}}
	\end{minipage}
	\begin{minipage}{0.24\hsize}
		\centering
		\includegraphics[width=1.00\linewidth]{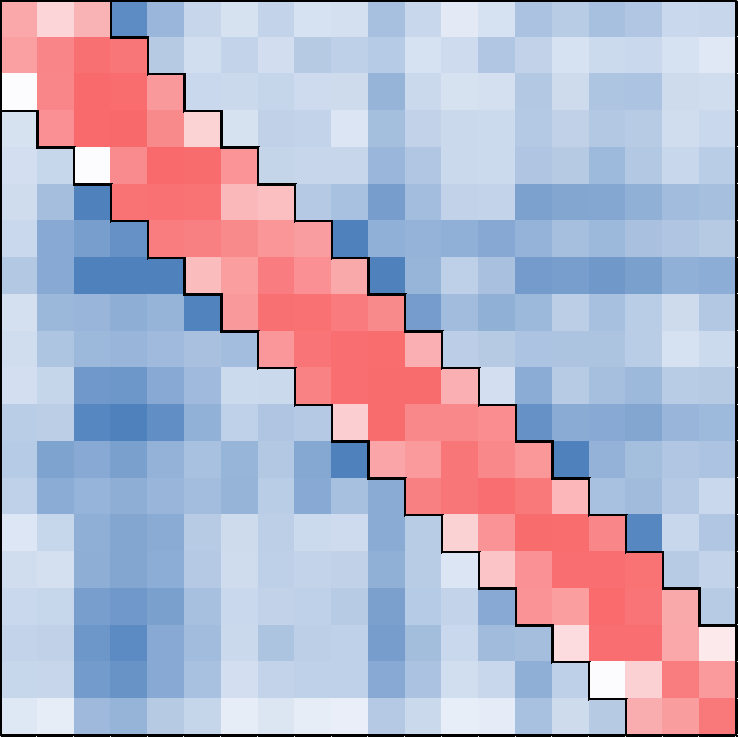}
		\subcaption{\footnotesize{lasso}}
	\end{minipage}
	\begin{minipage}{0.24\hsize}
		\centering
		\includegraphics[width=1.00\linewidth]{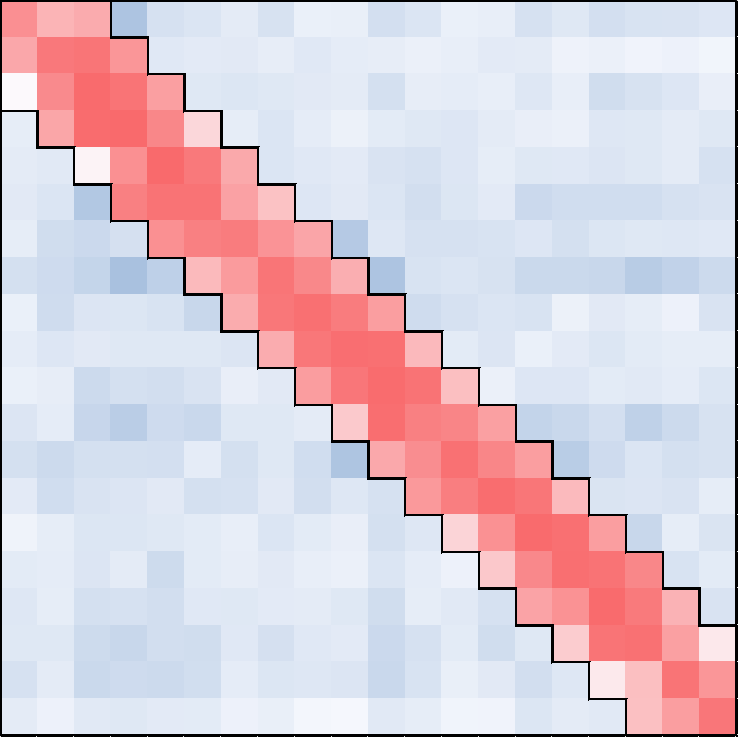}
		\subcaption{\footnotesize{adaptive lasso}}
	\end{minipage} 
	\begin{minipage}{0.24\hsize}
		\centering
		\includegraphics[width=1.00\linewidth]{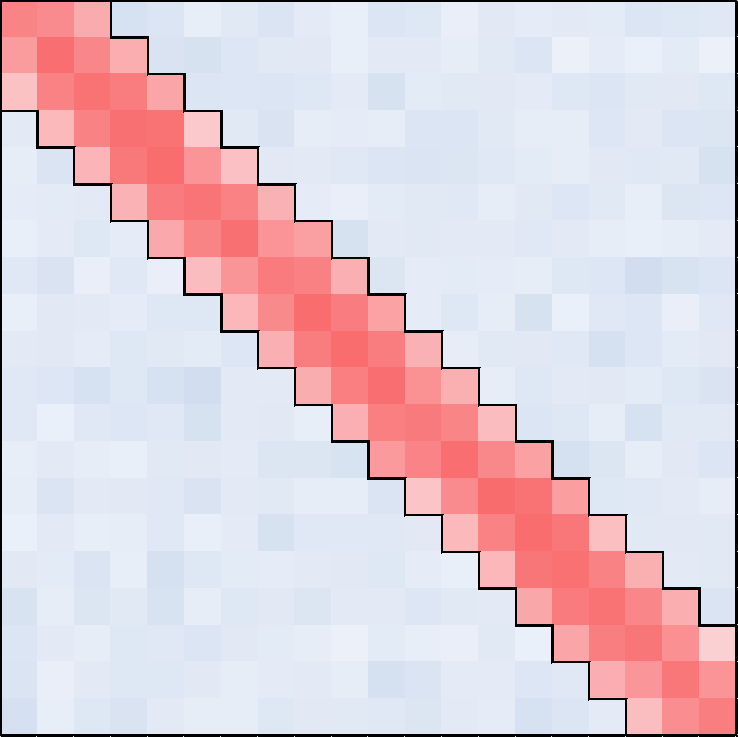}
		\subcaption{\footnotesize{multiple test,$q=0.2$}}
	\end{minipage} 
    \raggedright
       \includegraphics[width=0.1\linewidth]{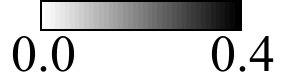}
 \hspace{50mm}
 	\includegraphics[width=0.2\linewidth]{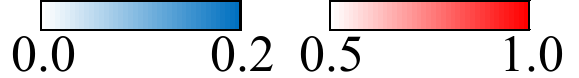}
	\caption{Heatmap of selection frequencies in  $\bPhi_1$; the blue and red cells indicate the frequencies of false and true discoveries, respectively.}
	\label{figure:prelim}
\end{figure}

Another practical advantage of our inference method is that by varying the value of $q$, elements of $\bPhi$ with different levels of significance can be visualized. By decreasing the value of $q$, the overall rejection frequency naturally decreases. In particular, the smaller the value of $q$, the more statistically significant elements are more likely to be selected; see the results with $q=0.1$ in Figure \ref{figure:prelim2} in Supplementary Material. 
Clearly, point estimation methods, such as the (adaptive) lasso, cannot provide this feature. 


\subsection{Second experiment: Performance of our methods}

\subsubsection{Design}\label{subs:design}
We consider the $N\times1$ vector generated as $\mathbf{y}_{t}=\sum_{\ell=1}^{K}{\bPhi}_{\ell
}\mathbf{y}_{t-\ell}+\mathbf{u}_{t}$ for $t=-(50+K),\dots,T$ with $\mathbf{y}%
_{-(50+K)}=\mathbf{0}$, and the set $\{\mathbf{y}_{-(50+K)},\dots,\mathbf{y}_{-K}\}$
is discarded. Define $\mathbf{u}_{t}={\bSigma}_{u}^{1/2}{\bve
}_{t}$, where 
${\bSigma}_{u}=\sigma^2 \bI_n$ with $\sigma=1$.
Another construction of $\bSigma_u$ with non-zero off-diagonals is considered in Supplementary Material. 
Two different distributions of $\varepsilon_{ti}$ are considered: 
(i) Standard normal, $\varepsilon_{ti}\sim \text{i.i.d.}N\left(  0,1\right)  $; (ii) Standardized mixture normal, $\varepsilon_{ti}=(\eta_{ti}-\mu_{\eta})/\sigma_{\eta}$ with $\eta_{ti}=q_{ti}\xi_{ti}+\left(  1-q_{ti}\right)\zeta_{ti}$, $\mu_{\eta}=\mathbb{E}\eta_{ti}$, and $\sigma_{\eta}^{2}=\mathbb{E}\eta_{ti}^{2}-\mu_{\eta}^{2}$, where $q_{ti}\sim \text{i.i.d.}Ber(\pi)$, $\xi_{ti}\sim \text{i.i.d.}N(  \mu_{\xi},\sigma_{\xi}^{2})$, and $\zeta_{ti}\sim \text{i.i.d.}N(\mu_{\zeta},\sigma_{\zeta}^{2})$. We have chosen the parameter values $\mu_{\xi}=2$, $\sigma_{\xi}=2$, $\mu_{\zeta}=4$, $\sigma_{\zeta}=10$, and $\pi=0.9$, which give $\mu_{\eta}=0.4$ and $\sigma_{\eta}=3.88$. The resulting error variable, $\varepsilon_{ti}$, is unimodal yet it has skewness 1.86 and kurtosis 3.53.

We focus on the model with $K=1$. The coefficient matrix ${\bPhi}=(\phi_{ij})$ is constructed as follows. First generate $\mathbf{\Psi}=(\psi_{ij})$ with $\psi_{ij}=\rho^{1+|i-j|/4}1\{| i-j| \leq m\}$. Now
form an $N\times N$ sign matrix 
$\boldsymbol{\Upsilon}^{(r)}=(\upsilon_{ij}^{(r)})$ for $r=1,2,\dots$, where
$\upsilon_{ij}^{(r)}=2\varphi_{ij}^{(r)}-1$ with $\varphi_{ij}^{(r)}\sim \text{i.i.d.}Ber(0.5)$. Repeatedly compute ${\bPhi}^{(r)}={\bPsi}\circ \boldsymbol{\Upsilon}^{(r)}$ for $r=1,2,\dots$ until $\lambda_{\max}({\bPhi}^{(r)})\leq0.96$. If the process stops at the $R$th repetition,  set ${\bPhi}={\bPhi}^{(R)}$. We consider $\rho=0.4$ and $m\in\{2,4,7\}$.
Observe that, except for the first and last rows, we have $s_{i}=5, 9, 15$ for
$m=2,4,7$, respectively. Set $q=0.1$ and consider all the combinations of $N\in\{50,100,200,300\}$ and $T\in\{200,300\}$. 

The model is estimated for each $i$th row using the R package, ``\texttt{glmnet}.'' The bootstrap procedure uses the same values of the lasso tuning parameters at the estimation stage. The CLIME estimator is constructed by the R package, ``\texttt{fastclime}.'' 
All the results are based on 1000 replications and 100 bootstrap samples. 

\subsubsection{Results}\label{subs:res}
Table \ref{table:fdr_het1} summarizes the results of Procedures \ref{proc:asy} and \ref{proc:boot} with
${\bSigma}_{u}=\bI_n$  
and  $\hat{m}_{ij}=\hat{\sigma}_i\sqrt{\hat{\bomega}'_{j} \hat{\bSigma}_x\hat{\bomega}_{j}}$.
It contains three panels for $m=2,4,7$; The larger the value of $m$, the larger the number of nonzeros in ${\bPhi}$. The results with $\bSigma_u$ with nonzero off-diagonals are qualitatively very similar; see Table \ref{table:fdr_het2} in Supplementary Material. 

As can be seen, both procedures control the dFDR to be around the predetermined level $q=0.1$, maintaining the high power. In particular, the good performance for large $N\geq T$ is very encouraging. The asymptotic threshold tends to produce slightly larger dFDR's than $q$ while the bootstrap one is more conservative. The mixture normal error moderately amplifies the tendency. Despite this conservativeness, the power based on the bootstrap threshold is almost identical to that based on the asymptotic one. The exaggeration and the conservativeness by both procedures are mitigated if $T$ is getting large. As expected, when the nonzero elements in $\bPhi$ increase (i.e. $m$ rises), the power goes down. However, it quickly rises as $T$ increases. 

The results of Procedure \ref{proc:eBH} based on the transformed t-statistic, $f(\tT_{ij})$, is reported in Table \ref{table:fdr_het1_ev}. Two transformation functions $f(x)$ are considered: $|x|^{p}$ with $p=10$ and $\exp(c|x|^{\alpha})$ with $c=3$ and $\alpha=1$. Both give very similar FDR and power. Their FDR are more conservative than the asymptotic and bootstrap FDR, but the associated reduction in power is relatively small. In conclusion, the e-BH procedure can be a reliable and robust alternative.

\setlength{\tabcolsep}{3.3pt}%
\renewcommand{\arraystretch}{.95}
\begin{table}[htb!]
	\caption{Directional FDR and power using asymptotic and bootstrap thresholds for $q=0.1$, with cross-sectionally uncorrelated errors}	
	\label{table:fdr_het1}
	\centering
	
	\begin{tabular}
		[c]{rlrccrccrccrcc}\hline
		\multicolumn{14}{c}{$m=2$ ($\max s_{i}=5$)}\\\hline
		&  &  & \multicolumn{5}{c}{$T=200$} &  & \multicolumn{5}{c}{$T=300$%
		}\\\cline{4-8}\cline{10-14}
		&  &  & \multicolumn{2}{c}{asymptotic} &  & \multicolumn{2}{c}{bootstrap} &  &
		\multicolumn{2}{c}{asymptotic} &  & \multicolumn{2}{c}{bootstrap}%
		\\\cline{4-5}\cline{7-8}\cline{10-11}\cline{13-14}
		&  &  & dFDR & PWR &  & dFDR & PWR &  & dFDR & PWR &  & dFDR & PWR\\
		\multicolumn{14}{l}{Standard Normal Error}\\
		& $N=50$ &  & \multicolumn{1}{r}{9.3} & \multicolumn{1}{r}{97.5} &  &
		\multicolumn{1}{r}{6.7} & \multicolumn{1}{r}{96.8} &  &
		\multicolumn{1}{r}{8.8} & \multicolumn{1}{r}{99.8} &  &
		\multicolumn{1}{r}{7.0} & \multicolumn{1}{r}{99.7}\\
		& $N=100$ &  & \multicolumn{1}{r}{10.7} & \multicolumn{1}{r}{94.7} &  &
		\multicolumn{1}{r}{6.9} & \multicolumn{1}{r}{93.2} &  &
		\multicolumn{1}{r}{9.9} & \multicolumn{1}{r}{99.4} &  &
		\multicolumn{1}{r}{7.4} & \multicolumn{1}{r}{99.2}\\
		& $N=200$ &  & \multicolumn{1}{r}{10.7} & \multicolumn{1}{r}{91.9} &  &
		\multicolumn{1}{r}{7.7} & \multicolumn{1}{r}{90.5} &  &
		\multicolumn{1}{r}{10.3} & \multicolumn{1}{r}{99.0} &  &
		\multicolumn{1}{r}{8.4} & \multicolumn{1}{r}{98.8}\\
		& $N=300$ &  & \multicolumn{1}{r}{10.4} & \multicolumn{1}{r}{89.6} &  &
		\multicolumn{1}{r}{6.8} & \multicolumn{1}{r}{87.3} &  &
		\multicolumn{1}{r}{9.7} & \multicolumn{1}{r}{98.5} &  &
		\multicolumn{1}{r}{8.5} & \multicolumn{1}{r}{98.4}\\
		\multicolumn{14}{l}{Mixture Normal Error}\\
		& $N=50$ &  & \multicolumn{1}{r}{10.4} & \multicolumn{1}{r}{94.2} &  &
		\multicolumn{1}{r}{5.7} & \multicolumn{1}{r}{91.9} &  &
		\multicolumn{1}{r}{9.5} & \multicolumn{1}{r}{99.0} &  &
		\multicolumn{1}{r}{6.3} & \multicolumn{1}{r}{98.7}\\
		& $N=100$ &  & \multicolumn{1}{r}{12.6} & \multicolumn{1}{r}{89.5} &  &
		\multicolumn{1}{r}{5.5} & \multicolumn{1}{r}{84.5} &  &
		\multicolumn{1}{r}{11.5} & \multicolumn{1}{r}{97.8} &  &
		\multicolumn{1}{r}{6.5} & \multicolumn{1}{r}{96.8}\\
		& $N=200$ &  & \multicolumn{1}{r}{13.1} & \multicolumn{1}{r}{87.0} &  &
		\multicolumn{1}{r}{6.6} & \multicolumn{1}{r}{82.5} &  &
		\multicolumn{1}{r}{11.7} & \multicolumn{1}{r}{97.2} &  &
		\multicolumn{1}{r}{7.6} & \multicolumn{1}{r}{96.4}\\
		& $N=300$ &  & \multicolumn{1}{r}{14.0} & \multicolumn{1}{r}{84.4} &  &
		\multicolumn{1}{r}{5.5} & \multicolumn{1}{r}{77.4} &  &
		\multicolumn{1}{r}{12.6} & \multicolumn{1}{r}{96.2} &  &
		\multicolumn{1}{r}{7.7} & \multicolumn{1}{r}{95.1}\\\hline
		\multicolumn{14}{c}{$m=4$ ($\max s_{i}=9$)}\\\hline
		&  &  & \multicolumn{5}{c}{$T=200$} &  & \multicolumn{5}{c}{$T=300$%
		}\\\cline{4-8}\cline{10-14}
		&  &  & \multicolumn{2}{c}{asymptotic} &  & \multicolumn{2}{c}{bootstrap} &  &
		\multicolumn{2}{c}{asymptotic} &  & \multicolumn{2}{c}{bootstrap}%
		\\\cline{4-5}\cline{7-8}\cline{10-11}\cline{13-14}
		&  &  & dFDR & PWR &  & dFDR & PWR &  & dFDR & PWR &  & dFDR & PWR\\
		\multicolumn{14}{l}{Standard Normal Error}\\
		& $N=50$ &  & \multicolumn{1}{r}{8.3} & \multicolumn{1}{r}{89.1} &  &
		\multicolumn{1}{r}{6.7} & \multicolumn{1}{r}{87.9} &  &
		\multicolumn{1}{r}{7.9} & \multicolumn{1}{r}{96.7} &  &
		\multicolumn{1}{r}{6.9} & \multicolumn{1}{r}{96.4}\\
		& $N=100$ &  & \multicolumn{1}{r}{8.5} & \multicolumn{1}{r}{83.7} &  &
		\multicolumn{1}{r}{8.5} & \multicolumn{1}{r}{83.6} &  &
		\multicolumn{1}{r}{8.3} & \multicolumn{1}{r}{94.7} &  &
		\multicolumn{1}{r}{8.5} & \multicolumn{1}{r}{94.8}\\
		& $N=200$ &  & \multicolumn{1}{r}{9.3} & \multicolumn{1}{r}{78.1} &  &
		\multicolumn{1}{r}{9.6} & \multicolumn{1}{r}{78.4} &  &
		\multicolumn{1}{r}{8.8} & \multicolumn{1}{r}{91.9} &  &
		\multicolumn{1}{r}{9.6} & \multicolumn{1}{r}{92.3}\\
		& $N=300$ &  & \multicolumn{1}{r}{10.3} & \multicolumn{1}{r}{73.0} &  &
		\multicolumn{1}{r}{8.8} & \multicolumn{1}{r}{71.8} &  &
		\multicolumn{1}{r}{9.4} & \multicolumn{1}{r}{89.3} &  &
		\multicolumn{1}{r}{9.9} & \multicolumn{1}{r}{89.5}\\
		\multicolumn{14}{l}{Mixture Normal Error}\\
		& $N=50$ &  & \multicolumn{1}{r}{8.5} & \multicolumn{1}{r}{85.1} &  &
		\multicolumn{1}{r}{5.8} & \multicolumn{1}{r}{82.5} &  &
		\multicolumn{1}{r}{7.8} & \multicolumn{1}{r}{94.8} &  &
		\multicolumn{1}{r}{6.1} & \multicolumn{1}{r}{93.9}\\
		& $N=100$ &  & \multicolumn{1}{r}{8.9} & \multicolumn{1}{r}{79.9} &  &
		\multicolumn{1}{r}{7.2} & \multicolumn{1}{r}{78.2} &  &
		\multicolumn{1}{r}{8.7} & \multicolumn{1}{r}{92.4} &  &
		\multicolumn{1}{r}{7.6} & \multicolumn{1}{r}{91.9}\\
		& $N=200$ &  & \multicolumn{1}{r}{11.5} & \multicolumn{1}{r}{74.7} &  &
		\multicolumn{1}{r}{8.2} & \multicolumn{1}{r}{71.7} &  &
		\multicolumn{1}{r}{10.6} & \multicolumn{1}{r}{88.9} &  &
		\multicolumn{1}{r}{8.8} & \multicolumn{1}{r}{88.0}\\
		& $N=300$ &  & \multicolumn{1}{r}{12.0} & \multicolumn{1}{r}{69.9} &  &
		\multicolumn{1}{r}{6.9} & \multicolumn{1}{r}{65.3} &  &
		\multicolumn{1}{r}{11.0} & \multicolumn{1}{r}{86.1} &  &
		\multicolumn{1}{r}{8.9} & \multicolumn{1}{r}{85.1}\\\hline
		\multicolumn{14}{c}{$m=7$ ($\max s_{i}=15$)}\\\hline
		&  &  & \multicolumn{5}{c}{$T=200$} &  & \multicolumn{5}{c}{$T=300$%
		}\\\cline{4-8}\cline{10-14}
		&  &  & \multicolumn{2}{c}{asymptotic} &  & \multicolumn{2}{c}{bootstrap} &  &
		\multicolumn{2}{c}{asymptotic} &  & \multicolumn{2}{c}{bootstrap}%
		\\\cline{4-5}\cline{7-8}\cline{10-11}\cline{13-14}
		&  &  & dFDR & PWR &  & dFDR & PWR &  & dFDR & PWR &  & dFDR & PWR\\
		\multicolumn{14}{l}{Standard Normal Error}\\
		& $N=50$ &  & \multicolumn{1}{r}{6.8} & \multicolumn{1}{r}{68.6} &  &
		\multicolumn{1}{r}{6.0} & \multicolumn{1}{r}{67.5} &  &
		\multicolumn{1}{r}{6.7} & \multicolumn{1}{r}{80.4} &  &
		\multicolumn{1}{r}{6.1} & \multicolumn{1}{r}{79.8}\\
		& $N=100$ &  & \multicolumn{1}{r}{7.6} & \multicolumn{1}{r}{60.6} &  &
		\multicolumn{1}{r}{8.0} & \multicolumn{1}{r}{61.0} &  &
		\multicolumn{1}{r}{7.6} & \multicolumn{1}{r}{74.2} &  &
		\multicolumn{1}{r}{8.0} & \multicolumn{1}{r}{74.6}\\
		& $N=200$ &  & \multicolumn{1}{r}{9.5} & \multicolumn{1}{r}{53.0} &  &
		\multicolumn{1}{r}{9.0} & \multicolumn{1}{r}{52.5} &  &
		\multicolumn{1}{r}{8.6} & \multicolumn{1}{r}{67.8} &  &
		\multicolumn{1}{r}{9.1} & \multicolumn{1}{r}{68.2}\\
		& $N=300$ &  & \multicolumn{1}{r}{9.4} & \multicolumn{1}{r}{50.1} &  &
		\multicolumn{1}{r}{9.3} & \multicolumn{1}{r}{50.1} &  &
		\multicolumn{1}{r}{8.5} & \multicolumn{1}{r}{65.3} &  &
		\multicolumn{1}{r}{10.1} & \multicolumn{1}{r}{66.4}\\
		\multicolumn{14}{l}{Mixture Normal Error}\\
		& $N=50$ &  & \multicolumn{1}{r}{6.4} & \multicolumn{1}{r}{65.5} &  &
		\multicolumn{1}{r}{4.9} & \multicolumn{1}{r}{63.2} &  &
		\multicolumn{1}{r}{6.1} & \multicolumn{1}{r}{77.9} &  &
		\multicolumn{1}{r}{5.1} & \multicolumn{1}{r}{76.6}\\
		& $N=100$ &  & \multicolumn{1}{r}{7.9} & \multicolumn{1}{r}{58.1} &  &
		\multicolumn{1}{r}{6.8} & \multicolumn{1}{r}{56.7} &  &
		\multicolumn{1}{r}{7.9} & \multicolumn{1}{r}{71.8} &  &
		\multicolumn{1}{r}{7.2} & \multicolumn{1}{r}{71.1}\\
		& $N=200$ &  & \multicolumn{1}{r}{9.4} & \multicolumn{1}{r}{51.1} &  &
		\multicolumn{1}{r}{7.5} & \multicolumn{1}{r}{49.3} &  &
		\multicolumn{1}{r}{9.0} & \multicolumn{1}{r}{65.6} &  &
		\multicolumn{1}{r}{8.3} & \multicolumn{1}{r}{65.0}\\
		& $N=300$ &  & \multicolumn{1}{r}{10.5} & \multicolumn{1}{r}{48.9} &  &
		\multicolumn{1}{r}{7.5} & \multicolumn{1}{r}{46.5} &  &
		\multicolumn{1}{r}{9.6} & \multicolumn{1}{r}{63.4} &  &
		\multicolumn{1}{r}{9.1} & \multicolumn{1}{r}{62.9}\\\hline
	\end{tabular}
\end{table}

\setlength{\tabcolsep}{3.3pt}%
\renewcommand{\arraystretch}{.95}
\begin{table}[htb!]
	\caption{FDR and power using e-BH thresholds for $q=0.1$, with cross-sectionally uncorrelated errors}	
	\label{table:fdr_het1_ev}
	\centering

\begin{tabular}
[c]{rlrcccccrccccc}\hline
\multicolumn{14}{c}{$m=2$ ($\max s_{i}=5$)}\\\hline
&  &  & \multicolumn{5}{c}{$T=200$} &  & \multicolumn{5}{c}{$T=300$%
}\\\cline{4-8}\cline{10-14}%
\multicolumn{1}{c}{} & \multicolumn{1}{c}{$f(x)$} & \multicolumn{1}{c}{} &
\multicolumn{2}{c}{$|x|^{p}$} &  & \multicolumn{2}{c}{$\exp(c|x|^{\alpha})$} &
\multicolumn{1}{c}{} & \multicolumn{2}{c}{$|x|^{p}$} &  &
\multicolumn{2}{c}{$\exp(c|x|^{\alpha})$}\\\cline{4-5}\cline{7-8}%
\cline{10-11}\cline{13-14}
&  &  & FDR & PWR & \multicolumn{1}{r}{} & FDR & PWR &  & FDR & PWR &
\multicolumn{1}{r}{} & FDR & PWR\\
\multicolumn{8}{l}{Standard Normal Error} & \multicolumn{1}{l}{} &
\multicolumn{1}{l}{} & \multicolumn{1}{l}{} & \multicolumn{1}{l}{} &
\multicolumn{1}{l}{} & \multicolumn{1}{l}{}\\
& $N=50$ &  & \multicolumn{1}{r}{1.8} & \multicolumn{1}{r}{93.2} &
\multicolumn{1}{r}{} & \multicolumn{1}{r}{1.3} & \multicolumn{1}{r}{92.0} &  &
\multicolumn{1}{r}{1.4} & \multicolumn{1}{r}{99.0} & \multicolumn{1}{r}{} &
\multicolumn{1}{r}{0.9} & \multicolumn{1}{r}{98.7}\\
& $N=100$ &  & \multicolumn{1}{r}{2.2} & \multicolumn{1}{r}{88.4} &
\multicolumn{1}{r}{} & \multicolumn{1}{r}{1.5} & \multicolumn{1}{r}{86.5} &  &
\multicolumn{1}{r}{1.8} & \multicolumn{1}{r}{98.0} & \multicolumn{1}{r}{} &
\multicolumn{1}{r}{1.2} & \multicolumn{1}{r}{97.5}\\
& $N=200$ &  & \multicolumn{1}{r}{1.7} & \multicolumn{1}{r}{82.9} &
\multicolumn{1}{r}{} & \multicolumn{1}{r}{1.2} & \multicolumn{1}{r}{80.8} &  &
\multicolumn{1}{r}{1.5} & \multicolumn{1}{r}{96.9} & \multicolumn{1}{r}{} &
\multicolumn{1}{r}{1.0} & \multicolumn{1}{r}{96.3}\\
& $N=300$ &  & \multicolumn{1}{r}{1.4} & \multicolumn{1}{r}{78.6} &
\multicolumn{1}{r}{} & \multicolumn{1}{r}{1.1} & \multicolumn{1}{r}{76.6} &  &
\multicolumn{1}{r}{1.2} & \multicolumn{1}{r}{95.7} & \multicolumn{1}{r}{} &
\multicolumn{1}{r}{0.8} & \multicolumn{1}{r}{95.0}\\
\multicolumn{8}{l}{Mixture Normal Error} & \multicolumn{1}{l}{} &
\multicolumn{1}{l}{} & \multicolumn{1}{l}{} & \multicolumn{1}{l}{} &
\multicolumn{1}{l}{} & \multicolumn{1}{l}{}\\
& $N=50$ &  & \multicolumn{1}{r}{2.5} & \multicolumn{1}{r}{87.7} &
\multicolumn{1}{r}{} & \multicolumn{1}{r}{1.8} & \multicolumn{1}{r}{86.0} &  &
\multicolumn{1}{r}{1.9} & \multicolumn{1}{r}{97.1} & \multicolumn{1}{r}{} &
\multicolumn{1}{r}{1.4} & \multicolumn{1}{r}{96.4}\\
& $N=100$ &  & \multicolumn{1}{r}{3.1} & \multicolumn{1}{r}{80.6} &
\multicolumn{1}{r}{} & \multicolumn{1}{r}{2.3} & \multicolumn{1}{r}{78.3} &  &
\multicolumn{1}{r}{2.4} & \multicolumn{1}{r}{94.7} & \multicolumn{1}{r}{} &
\multicolumn{1}{r}{1.7} & \multicolumn{1}{r}{93.6}\\
& $N=200$ &  & \multicolumn{1}{r}{2.4} & \multicolumn{1}{r}{75.5} &
\multicolumn{1}{r}{} & \multicolumn{1}{r}{1.8} & \multicolumn{1}{r}{73.1} &  &
\multicolumn{1}{r}{2.0} & \multicolumn{1}{r}{92.9} & \multicolumn{1}{r}{} &
\multicolumn{1}{r}{1.4} & \multicolumn{1}{r}{91.9}\\
& $N=300$ &  & \multicolumn{1}{r}{2.3} & \multicolumn{1}{r}{71.3} &
\multicolumn{1}{r}{} & \multicolumn{1}{r}{1.8} & \multicolumn{1}{r}{69.4} &  &
\multicolumn{1}{r}{2.3} & \multicolumn{1}{r}{71.3} & \multicolumn{1}{r}{} &
\multicolumn{1}{r}{1.5} & \multicolumn{1}{r}{90.1}\\\hline
\multicolumn{14}{c}{$m=4$ ($\max s_{i}=9$)}\\\hline
&  &  & \multicolumn{5}{c}{$T=200$} &  & \multicolumn{5}{c}{$T=300$%
}\\\cline{4-8}\cline{10-14}%
\multicolumn{1}{c}{} & \multicolumn{1}{c}{$f(x)$} & \multicolumn{1}{c}{} &
\multicolumn{2}{c}{$|x|^{p}$} &  & \multicolumn{2}{c}{$\exp(c|x|^{\alpha})$} &
\multicolumn{1}{c}{} & \multicolumn{2}{c}{$|x|^{p}$} &  &
\multicolumn{2}{c}{$\exp(c|x|^{\alpha})$}\\\cline{4-5}\cline{7-8}%
\cline{10-11}\cline{13-14}
&  &  & FDR & PWR & \multicolumn{1}{r}{} & FDR & PWR &  & FDR & PWR &
\multicolumn{1}{r}{} & FDR & PWR\\
\multicolumn{8}{l}{Standard Normal Error} & \multicolumn{1}{l}{} &
\multicolumn{1}{l}{} & \multicolumn{1}{l}{} & \multicolumn{1}{l}{} &
\multicolumn{1}{l}{} & \multicolumn{1}{l}{}\\
& $N=50$ &  & \multicolumn{1}{r}{1.7} & \multicolumn{1}{r}{79.5} &
\multicolumn{1}{r}{} & \multicolumn{1}{r}{1.2} & \multicolumn{1}{r}{77.4} &  &
\multicolumn{1}{r}{1.4} & \multicolumn{1}{r}{91.8} & \multicolumn{1}{r}{} &
\multicolumn{1}{r}{1.0} & \multicolumn{1}{r}{90.7}\\
& $N=100$ &  & \multicolumn{1}{r}{1.6} & \multicolumn{1}{r}{72.4} &
\multicolumn{1}{r}{} & \multicolumn{1}{r}{1.1} & \multicolumn{1}{r}{69.6} &  &
\multicolumn{1}{r}{1.5} & \multicolumn{1}{r}{88.7} & \multicolumn{1}{r}{} &
\multicolumn{1}{r}{0.9} & \multicolumn{1}{r}{87.2}\\
& $N=200$ &  & \multicolumn{1}{r}{1.7} & \multicolumn{1}{r}{66.5} &
\multicolumn{1}{r}{} & \multicolumn{1}{r}{1.2} & \multicolumn{1}{r}{63.9} &  &
\multicolumn{1}{r}{1.5} & \multicolumn{1}{r}{85.1} & \multicolumn{1}{r}{} &
\multicolumn{1}{r}{1.0} & \multicolumn{1}{r}{83.4}\\
& $N=300$ &  & \multicolumn{1}{r}{1.8} & \multicolumn{1}{r}{60.3} &
\multicolumn{1}{r}{} & \multicolumn{1}{r}{1.3} & \multicolumn{1}{r}{58.1} &  &
\multicolumn{1}{r}{1.5} & \multicolumn{1}{r}{81.4} & \multicolumn{1}{r}{} &
\multicolumn{1}{r}{1.0} & \multicolumn{1}{r}{79.6}\\
\multicolumn{8}{l}{Mixture Normal Error} & \multicolumn{1}{l}{} &
\multicolumn{1}{l}{} & \multicolumn{1}{l}{} & \multicolumn{1}{l}{} &
\multicolumn{1}{l}{} & \multicolumn{1}{l}{}\\
& $N=50$ &  & \multicolumn{1}{r}{1.9} & \multicolumn{1}{r}{74.2} &
\multicolumn{1}{r}{} & \multicolumn{1}{r}{1.4} & \multicolumn{1}{r}{71.9} &  &
\multicolumn{1}{r}{1.5} & \multicolumn{1}{r}{88.4} & \multicolumn{1}{r}{} &
\multicolumn{1}{r}{1.1} & \multicolumn{1}{r}{86.9}\\
& $N=100$ &  & \multicolumn{1}{r}{1.8} & \multicolumn{1}{r}{67.5} &
\multicolumn{1}{r}{} & \multicolumn{1}{r}{1.3} & \multicolumn{1}{r}{64.5} &  &
\multicolumn{1}{r}{1.7} & \multicolumn{1}{r}{85.0} & \multicolumn{1}{r}{} &
\multicolumn{1}{r}{1.1} & \multicolumn{1}{r}{83.1}\\
& $N=200$ &  & \multicolumn{1}{r}{2.1} & \multicolumn{1}{r}{62.5} &
\multicolumn{1}{r}{} & \multicolumn{1}{r}{1.5} & \multicolumn{1}{r}{59.9} &  &
\multicolumn{1}{r}{1.9} & \multicolumn{1}{r}{81.0} & \multicolumn{1}{r}{} &
\multicolumn{1}{r}{1.3} & \multicolumn{1}{r}{79.0}\\
& $N=300$ &  & \multicolumn{1}{r}{2.2} & \multicolumn{1}{r}{56.6} &
\multicolumn{1}{r}{} & \multicolumn{1}{r}{1.6} & \multicolumn{1}{r}{54.4} &  &
\multicolumn{1}{r}{1.9} & \multicolumn{1}{r}{77.1} & \multicolumn{1}{r}{} &
\multicolumn{1}{r}{1.3} & \multicolumn{1}{r}{75.1}\\\hline
\multicolumn{14}{c}{$m=7$ ($\max s_{i}=15$)}\\\hline
&  &  & \multicolumn{5}{c}{$T=200$} &  & \multicolumn{5}{c}{$T=300$%
}\\\cline{4-8}\cline{10-14}%
\multicolumn{1}{c}{} & \multicolumn{1}{c}{$f(x)$} & \multicolumn{1}{c}{} &
\multicolumn{2}{c}{$|x|^{p}$} &  & \multicolumn{2}{c}{$\exp(c|x|^{\alpha})$} &
\multicolumn{1}{c}{} & \multicolumn{2}{c}{$|x|^{p}$} &  &
\multicolumn{2}{c}{$\exp(c|x|^{\alpha})$}\\\cline{4-5}\cline{7-8}%
\cline{10-11}\cline{13-14}
&  &  & FDR & PWR & \multicolumn{1}{r}{} & FDR & PWR &  & FDR & PWR &
\multicolumn{1}{r}{} & FDR & PWR\\
\multicolumn{8}{l}{Standard Normal Error} & \multicolumn{1}{l}{} &
\multicolumn{1}{l}{} & \multicolumn{1}{l}{} & \multicolumn{1}{l}{} &
\multicolumn{1}{l}{} & \multicolumn{1}{l}{}\\
& $N=50$ &  & \multicolumn{1}{r}{1.3} & \multicolumn{1}{r}{55.8} &
\multicolumn{1}{r}{} & \multicolumn{1}{r}{0.9} & \multicolumn{1}{r}{53.7} &  &
\multicolumn{1}{r}{1.1} & \multicolumn{1}{r}{69.2} & \multicolumn{1}{r}{} &
\multicolumn{1}{r}{0.8} & \multicolumn{1}{r}{67.3}\\
& $N=100$ &  & \multicolumn{1}{r}{1.5} & \multicolumn{1}{r}{48.7} &
\multicolumn{1}{r}{} & \multicolumn{1}{r}{1.0} & \multicolumn{1}{r}{46.5} &  &
\multicolumn{1}{r}{1.4} & \multicolumn{1}{r}{63.7} & \multicolumn{1}{r}{} &
\multicolumn{1}{r}{0.9} & \multicolumn{1}{r}{61.5}\\
& $N=200$ &  & \multicolumn{1}{r}{1.6} & \multicolumn{1}{r}{42.3} &
\multicolumn{1}{r}{} & \multicolumn{1}{r}{1.1} & \multicolumn{1}{r}{40.2} &  &
\multicolumn{1}{r}{1.4} & \multicolumn{1}{r}{58.0} & \multicolumn{1}{r}{} &
\multicolumn{1}{r}{0.9} & \multicolumn{1}{r}{55.9}\\
& $N=300$ &  & \multicolumn{1}{r}{1.7} & \multicolumn{1}{r}{39.8} &
\multicolumn{1}{r}{} & \multicolumn{1}{r}{1.2} & \multicolumn{1}{r}{38.0} &  &
\multicolumn{1}{r}{1.4} & \multicolumn{1}{r}{55.7} & \multicolumn{1}{r}{} &
\multicolumn{1}{r}{1.0} & \multicolumn{1}{r}{53.9}\\
\multicolumn{8}{l}{Mixture Normal Error} & \multicolumn{1}{l}{} &
\multicolumn{1}{l}{} & \multicolumn{1}{l}{} & \multicolumn{1}{l}{} &
\multicolumn{1}{l}{} & \multicolumn{1}{l}{}\\
& $N=50$ &  & \multicolumn{1}{r}{1.3} & \multicolumn{1}{r}{52.0} &
\multicolumn{1}{r}{} & \multicolumn{1}{r}{0.9} & \multicolumn{1}{r}{49.6} &  &
\multicolumn{1}{r}{1.1} & \multicolumn{1}{r}{65.9} & \multicolumn{1}{r}{} &
\multicolumn{1}{r}{0.8} & \multicolumn{1}{r}{63.8}\\
& $N=100$ &  & \multicolumn{1}{r}{1.6} & \multicolumn{1}{r}{45.9} &
\multicolumn{1}{r}{} & \multicolumn{1}{r}{1.1} & \multicolumn{1}{r}{43.5} &  &
\multicolumn{1}{r}{1.5} & \multicolumn{1}{r}{60.7} & \multicolumn{1}{r}{} &
\multicolumn{1}{r}{1.1} & \multicolumn{1}{r}{58.4}\\
& $N=200$ &  & \multicolumn{1}{r}{1.9} & \multicolumn{1}{r}{40.2} &
\multicolumn{1}{r}{} & \multicolumn{1}{r}{1.3} & \multicolumn{1}{r}{38.1} &  &
\multicolumn{1}{r}{1.7} & \multicolumn{1}{r}{55.2} & \multicolumn{1}{r}{} &
\multicolumn{1}{r}{1.1} & \multicolumn{1}{r}{53.0}\\
& $N=300$ &  & \multicolumn{1}{r}{1.9} & \multicolumn{1}{r}{38.1} &
\multicolumn{1}{r}{} & \multicolumn{1}{r}{1.4} & \multicolumn{1}{r}{36.4} &  &
\multicolumn{1}{r}{1.7} & \multicolumn{1}{r}{53.2} & \multicolumn{1}{r}{} &
\multicolumn{1}{r}{1.2} & \multicolumn{1}{r}{51.3}\\\hline
\end{tabular}

\end{table}
\section{Two Empirical Applications}\label{sec:ee}

We apply our proposed methods to two large datasets to discover the underlying Granger-causal networks. Section \ref{sec:ee1} investigates the large macroeconomic and financial variables. Section \ref{sec:ee2} analyzes the regional house price growths in the UK. All the results in this section are based on $K=1$ and $\hat{m}_{ij}=\hat{\sigma}_i\sqrt{\hat{\bomega}'_{j} \hat{\bSigma}_x\hat{\bomega}_{j}}$.

\subsection{Large macroeconomic variables\label{sec:ee1} }

The FRED-MD data file of May 2019 is obtained from
McCracken's website, and the variables are transformed to be stationary as instructed by
\cite{MN2016}. The data consists of a balanced panel of 128 monthly series
spanning the period from June 1999 to May 2019. All series are standardized
before the analysis. Following \cite{MN2016}, the series are categorized into
eight groups:
\textbf{G1}, Output and Income; \textbf{G2}, Labour Market; \textbf{G3}, Consumption, Orders and Inventories; \textbf{G4}, Housing; \textbf{G5}, Interest and Exchange Rate; \textbf{G6}, Prices; \textbf{G7}, Money and
Credit; \textbf{G8}, Stock Market. (The group order is different from \cite{MN2016}.) The variables are numbered from 1 to 128,
and the descriptions for all the variables are reported in Table \ref{table:fredmdv} in Supplementary Material.

We estimate a VAR(1) model. The asymptotic and bootstrap thresholds for the $t$-ratios with $q=0.05$ were 2.88 and 4.41, respectively. Here we summarize the result with Procedure \ref{proc:boot} in Figure \ref{figure:nwFREDMDq05bts} as a network Granger causality diagram. The result with Procedure \ref{proc:asy} is reported in Supplementary Material. The nodes represent the variables, and their colors show the eight categories. The size of a node indicates the number of variables it significantly predicts; the larger the node size, the more variables it can predict. The arrows show the direction of the Granger causality. The self-lag effects are excluded in the figure. The main dynamic inter-relations are clustered within the groups, yet interesting interlinkages between the variable groups are observed. In particular, Price variables are clustered together, and eleven price variables are caused by variable 112, \textit{Real M2 Money Supply}. Price variable 94, \textit{Crude Oil}, also Granger-causes seven other price variables. Finally, \textit{Real Manufactures and Trade Industries Sales} (variable 49) Granger-cause three producer price indices (variables 90,\ 91,\ 92). These findings make a lot of sense from an economic point of view. In addition, it is easy to identify the variables that cause many other variables (many edges come out from the node) and those which are caused by many other variables (the node surrounded by many pointing arrows); see the \textit{Housing} variable cluster, for example.  

\noindent
\begin{figure}[h!]
	\begin{minipage}{1\hsize}
		\centering
		\includegraphics[width=130mm]
		{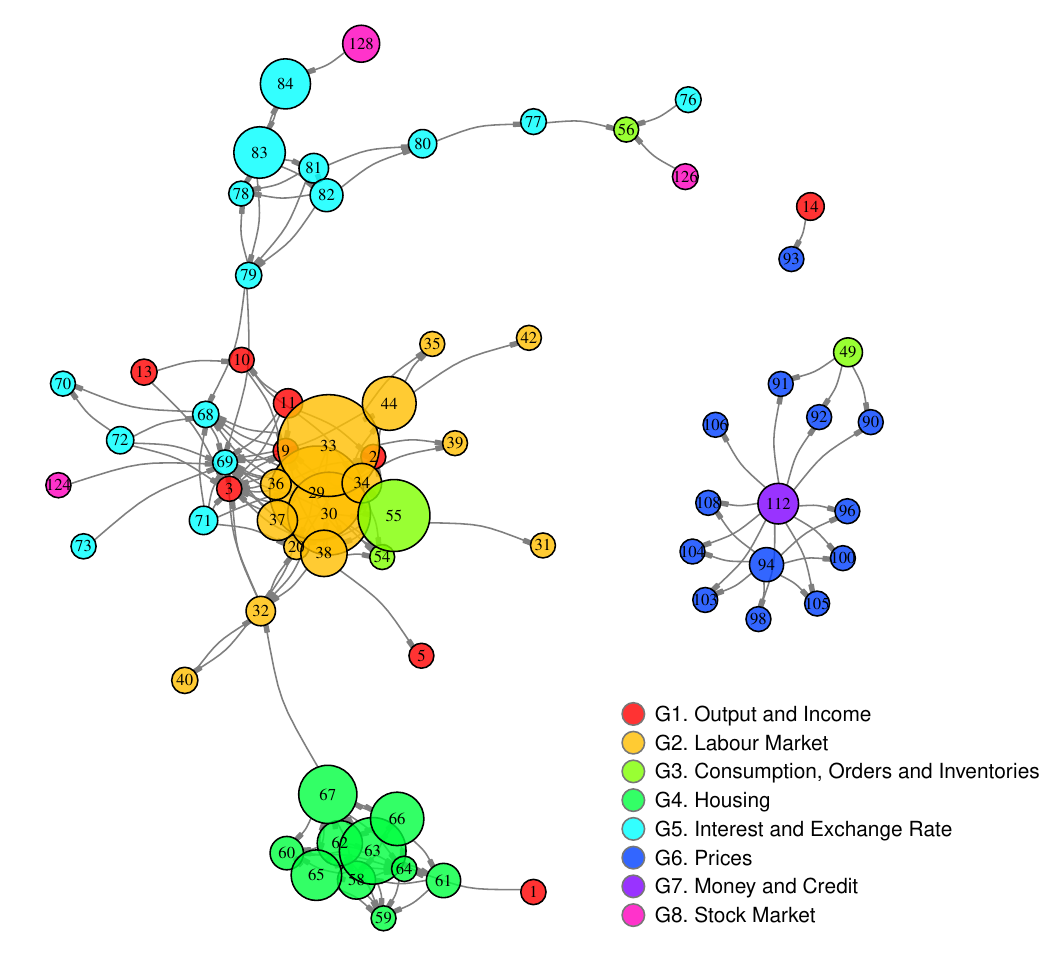}
		\vspace{0mm}
		\caption{\footnotesize{Network Granger-causality: 128 macroeconomic variables, bootstrap $t_0$, $q=0.05$}}
		\label{figure:nwFREDMDq05bts}
		\vspace{-3mm}
	\end{minipage}
\end{figure}

\subsection{UK regional house price growths\label{sec:ee2}}

We obtained the monthly average house prices at the local authority district level, published in November 2021 by HM Land Registry in the UK. Before analysis, we seasonally adjusted the prices, then deflated them by the UK consumer price index (CPI).\footnote{The CPI index (D7BT, not seasonally adjusted) is obtained from Office for National Statistics, UK. The house prices and the CPI are seasonally adjusted using the R package ``\texttt{seasonal}.'' } 
Denoting by $HP_{it}$ the seasonally adjusted real house price of the district $i$ at month $t$, the monthly house price growth is computed as $\Delta hp_{it}=\log(HP_{it}/HP_{it-1})$. In this analysis, we choose variables of 86 districts of Scotland, Wales, and the London area, spanning
209 months, from February 2004 to June 2021. The variables are numbered from 1
to 86, and the full list of the district names is reported in Table \ref{table:ukhpv} in Supplementary Materials. All the series are demeaned before the analysis. 

We estimate a VAR(1) model. The asymptotic and bootstrap thresholds for the $t$-ratios with $q=0.05$ were 2.73 and 4.22, respectively. Following the previous subsection, we summarize the result with the bootstrap threshold in Figure \ref{figure:nwUKHPq05bts} as a network Granger causality diagram. The result with the asymptotic threshold is reported in Supplementary Material.

\begin{figure}[!h]
	\begin{minipage}{1\hsize}
		\centering
		\includegraphics[width=130mm]
		{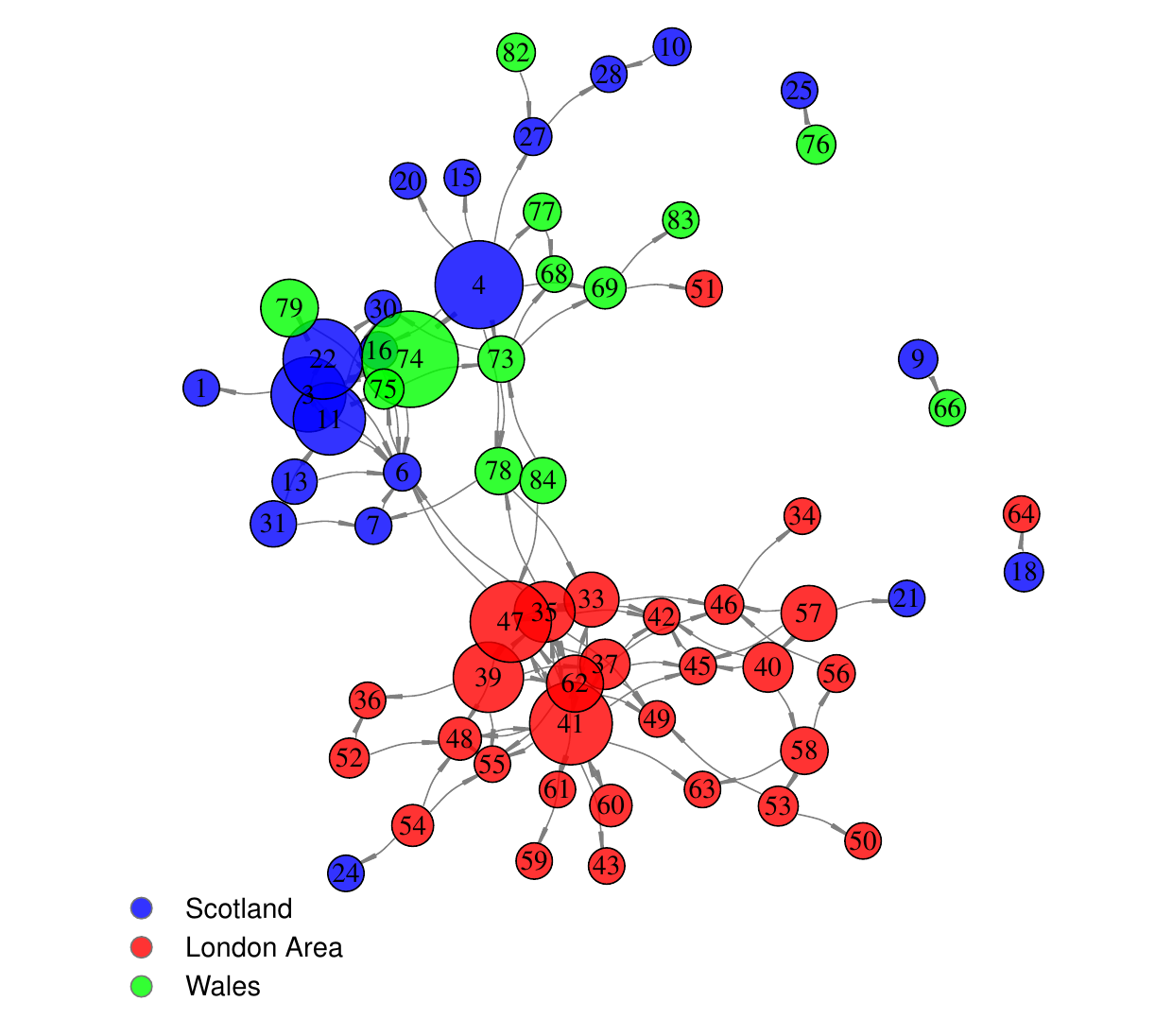}
		\vspace{0mm}
		\caption{\footnotesize{Network Granger causality: 86 UK regional house prices, bootstrap $t_0$, $q=0.05$}}
		\label{figure:nwUKHPq05bts}
		\vspace{-3mm}
	\end{minipage}
\end{figure}

The results show that house price growth causality networks are more or less clustered in each of the three regions. However, there are interesting inter-regional Granger-causal relationships. The London network Granger-causes variables 21 and 24, which are the house price growth in the City of Edinburgh and West Lothian. These two regions comprise a large part of Edinburgh, the capital of Scotland. Furthermore, the London regional network is shown to Granger-cause and also to be Granger-caused by Cardiff, the capital of Wales, variable 78. These results suggest that house price fluctuations in London interact dynamically with those in the capitals of other countries in the UK. 

There is one unique inter-regional Granger Causality. Falkirk in Scotland (Variable 6) has a significant concentration of arrows from London, Scotland, and Wales. A little searching leads to the BBC News of 21 July 2020 reporting that the Scottish and UK governments have pledged 90 million pounds in ``Growth Deal'' funds to stimulate the economy around Falkirk.\footnote{https://www.bbc.co.uk/news/uk-scotland-tayside-central-53471904} The deal was signed off by the UK, Scottish governments, and Falkirk council on 21 December 2021.\footnote{https://www.bbc.co.uk/news/uk-scotland-scotland-business-59734937} This large deal may have attracted investment to the Falkirk property market and was statistically identified as a Granger Causality for house price fluctuations in Falkirk.

\section{Conclusion}\label{sec:concl}

This paper has proposed multiple testing procedures that control the FDR for discovering the network Granger causality in high-dimensional VAR models. 
The validity of our inference-based framework is supported by the theory, simulation studies, and two empirical applications. We hope that the methods will enable us to stably discover networks inherent in various high-dimensional time series, serving as a clue for new theories in their domains.


We can consider some directions for future research:  
(i) \cite{Eichler2007} investigates a \textit{path diagram} that is composed of the Granger-causal network and contemporaneous connections in the error covariance matrix. As considered in \cite{BarigozziBrownlees2019}, extending our methods to include detection of the contemporaneous correlation networks is interesting and important; 
(ii) We have not included a serious discussion of lag order selection.  \cite{Stokes2017} argue that both too large and too small lag order selections will lead to spurious Granger causality findings. Apparently, investigating the ``optimal'' choice of $K$, such as combining the regularization-based lag selection method of \cite{NicholsonEtAl2020} with our framework, can be another important topic. 



\section*{Supplementary Material}
Section A: Proof of Theorems,
Section B: Proof of Propositions, 
Section C: Lemmas and their proofs, 
Section D: Precision Matrix Estimation, 
Section E: Additional Results on Robustification, 
Section F: Additional Experimental Results, 
Section G: List of Variable Names in Empirical Applications, 
Section H: Additional Results of Empirical Applications.

\section*{Acknowledgments}
The authors thank the editor, associate editor, and anonymous referees for their valuable comments and suggestions, which significantly improved the paper. 
This work was supported by JSPS KAKENHI Grant Numbers 19K13665, 20H01484, 20H05631, 21H00700, and 21H04397. The authors report there are no competing interests to declare.


\bibliographystyle{chicago}
\bibliography{references_GVAR}

\newpage
\appendix
\setcounter{page}{1}
\setcounter{section}{0}
\renewcommand{\theequation}{A.\arabic{equation}}
\setcounter{equation}{0}
\begin{center}
	{\Large Supplementary Material for} \\[7mm]
	\textbf{\Large Discovering the Network Granger Causality \\
		in Large Vector Autoregressive Models} \\[10mm]
	\textsc{\large Yoshimasa Uematsu$^*$} {\large and} \textsc{\large Takashi Yamagata$^\dagger$} \\[5mm]
	*\textit{\large Department of Social Data Science, Hitotsubashi University} \\[1mm]
	$\dagger$\textit{\large Department of Economics and Related Studies, University of York} \\[1mm]
	$\dagger$\textit{\large Institute of Social Economic Research, Osaka University}
\end{center}

\section{Proofs of Theorems} \label{sec:proof}


\subsection{Proof of Theorem \ref{thm:t}}

\begin{proof}[Proof]
	Fix any $(i,j)\in\cH$ such that $\bar{v}_i=o(1)$. Denote $m_{ij}=\sqrt{\sigma_i^2\omega_{jj}}=\sqrt{\sigma_i^2{\bomega}_{j}'\bSigma_x{\bomega}_{j}}$ and $\hat{m}_{ij}=\hat{\sigma}_i\sqrt{\hat{\bomega}_{j}'\hat{\bSigma}_x\hat{\bomega}_{j}}$ or $\hat{\sigma}_i\hat{\omega}_{j}$. 
	By the construction of the debiased lasso estimator with Proposition \ref{thm:asynormal} and Lemma \ref{lem:omegahat1}, the $t$-statistic is written as
	\begin{align*}
		\tT_{ij}-\frac{\sqrt{T}\phi_{ij}}{\hat{m}_{ij}}
		&= \frac{ \sqrt{T} ( \hat{\phi}_{ij} - \phi_{ij} )}{\hat{m}_{ij}} 
		= \frac{z_{ij}+r_{ij}}{m_{ij}}\left\{1+\left(\frac{m_{ij}}{\hat{m}_{ij}}-1\right)\right\} \\
		&=\left\{\frac{z_{ij}}{m_{ij}}+O\left(\bar{r}_i\right)\right\}\left\{1+O\left(M_\omega^2\lambda\right)\right\},
	\end{align*}
	which holds with probability at least $1-O((N\vee T)^{-\nu})$. Here, $\bar{r}_i$ and $M_\omega^2\lambda$ are asymptotically negligible (with high probability) by the assumed condition. 
	
	Next, prove the asymptotic normality of $z_{ij}/m_{ij}$ for each $(i,j)\in\cH$. 
	Recall that
	\begin{align*}
		\frac{z_{ij}}{m_{ij}} = \sum_{t=1}^T\xi_{Tt}^{(i,j)}, ~~~~~
		\xi_{Tt}^{(i,j)}=\frac{u_{it}\bx_t'{\bomega}_{j}}{\sqrt{T}m_{ij}}.
	\end{align*}
	By a simple calculation, we have $\Var(\xi_{Tt}^{(i,j)})=1/T$ and $\sum_{t=1}^T(\xi_{Tt}^{(i,j)})^2\to_p1$. Furthermore, by an application of Lemma \ref{lem:tail_yu}, we obtain
	\begin{align*}
		\max_t|\xi_{Tt}^{(i,j)}| \leq T^{-1/2}\max_t\|u_{it}\bx_t\|_\infty\|{\bomega}_{j}\|_1/m_{ij}
		\lesssim M_\omega T^{-1/2}\sqrt{\log^3(N\vee T)}
	\end{align*}
	with high probability, which implies $\max_{t}|\xi_{Tt}^{(i,j)}|\to_p0$ when $M_\omega^2\lambda=o(1)$. Thus McLeish's central limit theorem is applicable to achieve $\tT_{ij}-\sqrt{T}\phi_{ij}/\hat{m}_{ij}\to_d N(0,1)$ for each $(i,j)\in\cH$ such that $\bar{r}_i+M_\omega^2\lambda=o(1)$. This completes the proof. 
\end{proof}

\subsection{Proof of Theorem \ref{thm:fdr-t}}
\begin{proof}
	Throughout this proof, let $p=|\cH|=KN^2$ denote the total number of parameters in the (augmented) coefficient matrix, $\bPhi$. 
	Define 
	\begin{align*}
		&\cH_{\leq0}=\cH\cap\{(i,j):\phi_{ij}\leq 0\},~~~\cH_{\geq0}=\cH\cap\{(i,j):\phi_{ij}\geq 0\}, \\
		&\cS_{<0}=\cH\cap\{(i,j):\phi_{ij}< 0\},~~~\cS_{>0}=\cH\cap\{(i,j):\phi_{ij}> 0\}.
	\end{align*}
	By the definition of $s=|\cS|=|\cH\cap \{(i,j):\phi_{ij}\not=0\}|$, there is some sequence $\pi=\pi_{np}\in[0,1]$ such that $|\cS_{<0}|=\pi s$ and $|\cS_{>0}|=(1-\pi)s$. We also have $|\cH_{\leq0}|=p-(1-\pi)s$ and $|\cH_{\geq0}|=p-\pi s$. 
	
	\textit{Case 1.} Consider when \eqref{proc:t0} does not exist and $\tt_0=\sqrt{2\log p}$. First, we observe that 
	\begin{align}
		\dFDR(\tt_0) 
		&\leq \dFWER 
		= \Pro \left( \sum_{(i,j)\in\hat\cS(\tt_0)}1\{\sgn(\hat{\phi}_{ij})\not=\sgn(\phi_{ij})\}\geq 1 \right) \notag\\
		&\leq \Pro \left( \sum_{(i,j)\in\cH_{\leq0}}1\{\tT_{ij}\geq \tt_0\} \geq 1 \right)
		+ \Pro \left( \sum_{(i,j)\in\cH_{\geq0}}1\{\tT_{ij}\leq -\tt_0\} \geq 1 \right). \label{case1-1}
	\end{align}
	The first probability of \eqref{case1-1} is further bounded as
	\begin{align*}
		&\Pro \left( \sum_{(i,j)\in\cH_{\leq0}}1\{\tT_{ij}\geq \tt_0\} \geq 1 \right)
		\leq \sum_{(i,j)\in\cH_{\leq0}} \Pro \left( \tT_{ij}\geq \tt_0 \right) \\
		&\quad \leq \sum_{(i,j)\in\cS^c} \Pro \left( \cZ_{ij}\geq \tt_0-\delta_1 \right) + 
		\sum_{(i,j)\in\cS_{<0}} \Pro \left( \cZ_{ij}+\sqrt{T}\phi_{ij}/{\hat{m}_{ij}}\geq \tt_0-\delta_1 \right)  \\
  &\qquad\qquad\qquad\qquad\qquad\qquad\qquad\qquad\qquad\qquad\qquad\qquad + pO((N\vee T)^{-\nu+2})\\
		&\quad \leq \sum_{(i,j)\in\cS^c} \Pro \left( \cZ_{ij}\geq \tt_0-\delta_1 \right) + 
		\sum_{(i,j)\in\cS_{<0}} \Pro \left( \cZ_{ij}\geq \tt_0-\delta_1 \right)  + pO((N\vee T)^{-\nu+2})\\
		&\quad = (p-s+\pi s) Q(\tt_0-\delta_1)  + pO((N\vee T)^{-\nu+2})\\
		&\quad \leq p Q(\tt_0-\delta_1) + pO((N\vee T)^{-\nu+2}),
	\end{align*}
	where the second inequality holds for some positive sequence $\delta_1=O(T^{-\kappa})$ for some $\kappa\in(0,\kappa_1]$ by Lemma \ref{lem:strongapp}(a) and $Q(\tt)=\Pro(\cZ>\tt)$ is the upper tail probability of a standard normal random variable $\cZ$. Because $Q(\tt)\leq (\tt\sqrt{2\pi})^{-1}\exp(-\tt^2/2)$ and $\tt_0-\delta_1=\sqrt{2\log p}+ o(1)$, we have
	\begin{align*}
		p Q(\tt_0-\delta_1)
		\lesssim p(\sqrt{\log p})^{-1}\exp(-\log p + o(1) ) 
		=O(1/\sqrt{\log p}). 
	\end{align*}
	We also obtain $pO((N\vee T)^{-\nu+2})=o(1)$ for $\nu>4$. 
	The second probability of \eqref{case1-1} is bounded in the same way. We thus conclude $\dFDR(\tt_0) \leq \dFWER= o(1)$. 

	\textit{Case 2.} Consider when $\tt_0$ is given by \eqref{proc:t0}. Define
	\begin{align*}
		V = \sup_{\tt\in[0,\bar{\tt}]}\left|\frac{\sum_{(i,j)\in\cH_{\leq0}}\left[1\{\tT_{ij}\geq \tt\}-Q(\tt)\right] + \sum_{(i,j)\in\cH_{\geq0}}\left[1\{\tT_{ij}\leq -\tt\}-Q(\tt)\right]}{2pQ(\tt)}\right|.
	\end{align*}
	Then we have
	\begin{align*}
		\dFDP(\tt_0) 
		&= \frac{|\{(i,j)\in\hat{\cS}(\tt_0):\sgn(\hat{\phi}_{ij})\not=\sgn(\phi_{ij})\}|}{|\hat{\cS}(\tt_0)|\vee 1} \\
		&= \frac{2pQ(\tt_0)}{|\hat{\cS}(\tt_0)|\vee 1}\frac{\sum_{(i,j)\in\cH_{\leq0}}1\{\tT_{ij}\geq \tt_0\} + \sum_{(i,j)\in\cH_{\geq0}}1\{\tT_{ij}\leq -\tt_0\}}{2pQ(\tt_0)} \\
		&\leq q \left[\frac{\sum_{(i,j)\in\cH_{\leq0}}1\{\tT_{ij}\geq \tt_0\} + \sum_{(i,j)\in\cH_{\geq0}}1\{\tT_{ij}\leq -\tt_0\}}{2pQ(\tt_0)} \right] \\
		&\leq q\left[ V + \frac{(|\cH_{\leq0}|+|\cH_{\geq0}|)Q(\tt_0)}{2pQ(\tt_0)} \right] 
		\leq  q\left(V + 1\right).
	\end{align*}
	Thus the dFDR and dFDP are controlled if we prove $V=o_p(1)$ in view of Fatou's lemma and Markov's inequality, respectively. 
	Note that 
	\begin{align*}
		V &\leq \sup_{\tt\in[0,\bar{\tt}]}\left|\frac{\sum_{(i,j)\in\cH_{\leq0}}\left[1\{\tT_{ij}\geq \tt\}-Q(\tt)\right]}{2pQ(\tt)}\right| \\
		&\qquad \qquad \qquad + \sup_{\tt\in[0,\bar{\tt}]}\left|\frac{\sum_{(i,j)\in\cH_{\geq0}}\left[1\{\tT_{ij}\leq -\tt\}-Q(\tt)\right]}{2pQ(\tt)}\right|.
	\end{align*}
	We only prove that the first term is $o_p(1)$ by symmetry.

	To this end, consider discretization. That is, we partition $[0,\bar{\tt}]$ into small intervals, $0=\tt_0 < \tt_1< \dots <\tt_h = \bar{\tt}=(2\log p -a\log\log p)^{1/2}$, such that $\tt_m-\tt_{m-1}=v_p$ for $m\in\{1,\dots,h-1\}$ and $\tt_h-\tt_{h-1}\leq v_p$, where $v_p=(\log p\log\log p)^{-1/2}$. 
	Then a simple calculation gives $1/h \leq v_p/\bar{\tt} =O( 1/(\log p \sqrt{\log\log p}))$. 
	Fix arbitrary $m\in\{1,\dots,h\}$. For any $\tt\in[\tt_{m-1},\tt_m]$, we have
	\begin{align*}
		\frac{\sum_{(i,j)\in\cH_{\leq0}}1\{\tT_{ij}\geq \tt\}}{2pQ(\tt)}
		\leq \frac{\sum_{(i,j)\in\cH_{\leq0}}1\{\tT_{ij}\geq \tt_{m-1}\}}{2pQ(\tt_{m-1})}\frac{Q(\tt_{m-1})}{Q(\tt_{m})}
	\end{align*}
	and 
	\begin{align*}
		\frac{\sum_{(i,j)\in\cH_{\leq0}}1\{\tT_{ij}\geq \tt\}}{2pQ(\tt)}
		\geq \frac{\sum_{(i,j)\in\cH_{\leq0}}1\{\tT_{ij}\geq \tt_{m}\}}{2pQ(\tt_{m})}\frac{Q(\tt_{m})}{Q(\tt_{m-1})}.
	\end{align*}
	Lemma 7.2 of \cite{JJ2019} gives
	\begin{align*}
		\frac{Q(\tt_{m-1})}{Q(\tt_{m})}
		\leq\frac{Q(\tt_{m}-v_p)}{Q(\tt_{m})}
		= 1+O(v_p+v_p\bar{\tt})=1+o(1)
	\end{align*}
	uniformly in $m\in\{1,\dots,h\}$. Thus the proof completes if the following is true:
	\begin{align*}
		\tilde{V} := \max_{m\in\{1,\dots,h\}}\left|\frac{\sum_{(i,j)\in\cH_{\leq0}}\left[1\{\tT_{ij}\geq \tt_m\}-Q(\tt_m)\right]}{2pQ(\tt_m)}\right|=o_p(1).
	\end{align*}
	
	Prove $\tilde{V}=o_p(1)$. 
	Fix arbitrary $\ep>0$. Then we have
	\begin{align*}
		\Pro\left(\tilde{V}> \ep \right) 
		&\leq h \max_{m\in\{1,\dots,h\}} \Pro \left( \left|\frac{\sum_{(i,j)\in\cH_{\leq0}}\left[1\{\tT_{ij}\geq \tt_m\}-Q(\tt_m)\right]}{2pQ(\tt_m)}\right| > \ep\right) \\
		&\leq h \max_{m\in\{1,\dots,h\}} \E \left|\frac{\sum_{(i,j)\in\cH_{\leq0}}\left[1\{\tT_{ij}\geq \tt_m\}-Q(\tt_m)\right]}{2pQ(\tt_m)}\right|^2/\ep^2,
	\end{align*}
	where $h\lesssim \log p \sqrt{\log\log p}$. 
	Consider bounding the expectation uniformly in $m\in\{1,\dots,b\}$. Denote by $\sum$ and $\sum\sum$ the summations over $(i,j)\in\cH_{\leq0}$ and $(i,j),(k,\ell)\in\cH_{\leq0}$, respectively. By a simple calculus and Lemma \ref{lem:strongapp}(a)(b), we obtain
	\begin{align*}
		&\E \left[ \frac{\sum\sum\left[1\{\tT_{ij}\geq \tt_m\}-Q(\tt_m)\right]\left[1\{\tT_{k\ell}\geq \tt_m\}-Q(\tt_m)\right]}{4p^2Q(\tt_m)^2} \right] \\
		&= \frac{\sum\sum\Pro(\tT_{ij}\geq \tt_m, \tT_{k\ell}\geq \tt_m)}{4p^2Q(\tt_m)^2}
		- \frac{|\cH_{\leq0}|\sum\Pro(\tT_{ij}\geq\tt_m)}{2p^2Q(\tt_m)}
		+ \frac{|\cH_{\leq0}|^2}{4p^2} \\
		&\leq \frac{\sum\sum \Pro(\cZ_{ij}\geq \tt_m-\delta_1, \cZ_{k\ell}\geq \tt_m-\delta_1)}{4p^2Q(\tt_m)^2}
		- \frac{(1+O(s/p))Q(\tt_m+\delta_1)}{2Q(\tt_m)} +\frac{1}{4} \\
		&=: (i) + (ii) + 1/4,
	\end{align*}
	where $(\cZ_{ij},\cZ_{k\ell})$ is a bivariate standard normal random vector and  $\delta_1=O(T^{-\kappa})$ for some constant $\kappa>0$. 
	Thus we will conclude $\tilde{V}=o_p(1)$ if we show that $(i)\leq 1/4+o(1/h)$ and $(ii)\leq-1/2+o(1/h)$, where $1/h=O( 1/(\log p \sqrt{\log\log p}))$.

	First consider $(ii)$. Expand $Q(\tt_m+\delta_1)$ around $\delta_1=0$ by the mean value theorem. Then there exists $\delta_1^*$ between $0$ and $\delta_1$ such that 
	\begin{align}
		(ii) &= -\frac{Q(\tt_m+\delta_1)}{2Q(\tt_m)}(1+o(1))
		= -\frac{Q(\tt_m)+Q'(\tt_m+\delta_1^*)\delta_1}{2Q(\tt_m)}(1+o(1)) \notag\\
		&\leq -1/2 - \frac{(2\pi)^{-1/2}\exp\left\{-(\tt_m+\delta_1^*)^2/2\right\}\delta_1}{2(2\pi)^{-1/2}\tt_m^{-1}\exp\left\{-\tt_m^2/2\right\}}(1+o(1)) \notag\\
		&= -1/2 - \tt_m\exp\left\{-\delta_1^*(\tt_m+\delta_1^*/2)\right\}\delta_1(1+o(1))/2 \notag \\
		&= -1/2 + o(1/h), \label{(ii0)}
	\end{align}
	where the last equality holds uniformly in $m\in\{1,\dots,h\}$ because $\delta_1$ is polynomially decaying while $\tt_m$ is the logarithmic function for all $m\in\{1,\dots,h\}$.
	
	Next consider $(i)$ by decomposing the summation into two parts, $(i,j)=(k,\ell)$ and $(i,j)\not=(k,\ell)$. First we see the summation over $(i,j)=(k,\ell)$, which has $p$ entries. We have
	\begin{align}
		&\frac{\sum\sum \Pro(\cZ_{ij}\geq \tt_m-\delta_1, \cZ_{k\ell}\geq \tt_m-\delta_1)1\{(i,j)=(k,\ell)\}}{4p^2Q(\tt_m)^2} \notag\\
		&= \frac{|\cH_{\geq0}|}{4p} \frac{1}{pQ(\tt_m)} \frac{ Q(\tt_m-\delta_1)}{Q(\tt_m)}
		= o(1/h). \label{(I00)}
	\end{align}
	The last estimate is true because we have $|\cH_{\geq0}|/p=O(1)$ and $Q(\tt_m-\delta_1)/Q(\tt_m)=1+o(1/h)$ by the same reason as above, and  by \cite{SzarekWerner}, 
	\begin{align}
		\frac{1}{pQ(\tt_m)}
		&\leq \frac{\bar{\tt}+(\bar{\tt}^2+4)^{1/2}}{p2(2\pi)^{-1/2}\exp\left\{-\bar{\tt}^2/2\right\}}
		\lesssim \frac{\bar{\tt}}{p\exp\left\{-\bar{\tt}^2/2\right\}} \notag \\
		&\lesssim \log^{1/2}p\cdot\exp\{- \log\log^{a/2} p\}
		=O(\log^{1/2-a/2}p) = o(1/h) \label{szarek}
	\end{align}
	for any $a>3$.

	Finally we bound $(i)$ with summation over $(i,j)\not=(k,\ell)$, which contains $p^2-p$ entries. We have
	\begin{align}
		&\frac{\sum\sum \Pro(\cZ_{ij}\geq \tt_m-\delta_1, \cZ_{k\ell}\geq \tt_m-\delta_1)}{4p^2Q(\tt_m)^2}\notag \\
		&\leq \frac{\sum\sum \Pro(\cZ_{ij}\geq \tt_m-\delta_1, \cZ_{k\ell}\geq \tt_m-\delta_1)}{4p^2Q(\tt_m-\delta_1)^2}\frac{Q(\tt_m-\delta_1)^2}{Q(\tt_m)^2},\label{(I100)}
	\end{align}
	where $Q(\tt_m-\delta_1)^2/Q(\tt_m)^2=1+o(1/h)$ as above. 
	By Lemma \ref{lem:probratio} and the inequality $1/\sqrt{1-x^2}\leq 1+|x|/\sqrt{1-x^2}$ for any $x\in(-1,1)$, we obtain
	\begin{align}
		&\frac{\sum\sum \Pro(\cZ_{ij}\geq \tt_m-\delta_1, \cZ_{k\ell}\geq \tt_m-\delta_1)}{4p^2Q(\tt_m-\delta_1)^2}
		\leq \frac{1}{4p^2}\sum\sum \frac{1}{\sqrt{1-\rho_{(i,j),(k,\ell)}^2}} \notag\\
		&= 1/4+o(1/h)+\frac{1}{4p^2}\sum\sum \frac{|\rho_{(i,j),(k,\ell)}|}{\sqrt{1-\rho_{(i,j),(k,\ell)}^2}}.\label{(I10)}
	\end{align}
	We evaluate the sum using Condition \ref{ass:corr} with $\log N \asymp \log p$. Denote $\cH_w^2=\cH_{w1}\times \cH_{w2}$ and $\cH_s^2=\cH_{s1}\times \cH_{s2}$. Then it is decomposed as
	\begin{align*}
		\sum_{(i,j)\in\cH_{\leq 0}}\sum_{(k,\ell)\in\cH_{\leq 0}}
		=\sum_{(i,j)\in\cH_{\leq 0}\cap\cH_{w1}}\sum_{(k,\ell)\in\cH_{\leq 0}\cap\cH_{w2}}
		+\sum_{(i,j)\in\cH_{\leq 0}\cap\cH_{s1}}\sum_{(k,\ell)\in\cH_{\leq 0}\cap\cH_{s2}}.
	\end{align*}
	They are bounded by
	\begin{align}
		&\frac{1}{4p^2}\sum_{(i,j)\in\cH_{\leq 0}\cap\cH_{w1}}\sum_{(k,\ell)\in\cH_{\leq 0}\cap\cH_{w2}} \frac{|\rho_{(i,j),(k,\ell)}|}{\sqrt{1-\rho_{(i,j),(k,\ell)}^2}} \notag\\
		&\leq\frac{1}{4p^2}\sum_{(i,j)\in\cH_{w1}}\sum_{(k,\ell)\in\cH_{w2}} \frac{|\rho_{(i,j),(k,\ell)}|}{\sqrt{1-\rho_{(i,j),(k,\ell)}^2}} \notag\\
		&\leq \frac{p^2-p-|\cH_{s1}\times\cH_{s2}|}{4p^2} \frac{c/\log^2 p}{\sqrt{1-c^2/\log^4p}} 
		= O(1)O(1/\log^2 p)=o(1/h) \label{(I20)}
	\end{align}
	and
	\begin{align}
		&\frac{1}{4p^2}\sum_{(i,j)\in\cH_{\leq 0}\cap\cH_{s1}}\sum_{(k,\ell)\in\cH_{\leq 0}\cap\cH_{s2}} \frac{|\rho_{(i,j),(k,\ell)}|}{\sqrt{1-\rho_{(i,j),(k,\ell)}^2}}\notag \\
		&\leq\frac{1}{4p^2}\sum_{(i,j)\in\cH_{s1}}\sum_{(k,\ell)\in\cH_{s2}} \frac{|\rho_{(i,j),(k,\ell)}|}{\sqrt{1-\rho_{(i,j),(k,\ell)}^2}}\notag \\
		&\leq \frac{|\cH_{s1}\times\cH_{s2}|}{4p^2} \frac{|\bar{\rho}|}{\sqrt{1-\bar{\rho}^2}}
		= O(1/\log^2 p)O(1) = o(1/h). \label{(I30)}
	\end{align}
	From \eqref{(I00)} and \eqref{(I100)}--\eqref{(I30)}, we have $(i)=1/4+o(1/h)$. This completes the proof. 
\end{proof}

\subsection{Proof of Theorem \ref{thm:bootstrap}}
Before proceeding the proof, we recall the notation. In the proof, we mainly use the $t$-statistics with the ``sandwich'' s.e., defined as
\begin{align*}
	\tT_{ij}^* 
	= \frac{ \sqrt{T} \hat{\phi}_{ij}^*}{\hat{\sigma}^{*}_{i}\sqrt{\hat{\bomega}'_{j} \hat{\bSigma}_x\hat{\bomega}_{j}}}
\end{align*}
for $(i,j)\in\hat{\cS}_{\normalfont{\textsf{L}}}^c$, 
where 
\begin{align*}
	&\hat{\sigma}_i^{*2}=(T-\hat{s}_i)^{-1}\sum_{t=1}^T\hat{u}_{it}^{*2},~~~
	\hat{u}_{it}^{*2}={u}_{it}^{*2}-2{u}_{it}^*\bx_t'\bdelta_i^{*}+{\bdelta_i^{*}}'\bx_t\bx_t'\bdelta_i^{*}, \\
	&\bdelta_i^{*} = (\hat{\bphi}_{i\cdot}^{\normalfont{\textsf{L}}*} - \hat{\bphi}_{i\cdot}^{\normalfont{\textsf{L}}})', ~~~
	{u}_{it}^{*2}=\hat{u}_{it}^2\zeta_t^2. 
\end{align*}
The bootstrap debiased lasso estimator has been defined as
\begin{align*}
	\hat{\bPhi}^*&=\hat{\bPhi}^{\normalfont{\textsf{L}}*} + \left( \bY^* - \hat\bPhi^{\normalfont{\textsf{L}}*} \bX \right)\bX'\hat\bOmega/T \\
	&= \hat{\bPhi}^{\normalfont{\textsf{L}}} 
	+ \bU^*\bX'\bOmega/T
	+ \bU^*\bX'(\hat\bOmega-\bOmega)/T
	+ \left(\hat{\bPhi}^{\normalfont{\textsf{L}}}- \hat\bPhi^{\normalfont{\textsf{L}}*} \right)(\hat{\bSigma}_x\hat\bOmega-\bI).
\end{align*} 
By the definition, we have $\hat{\phi}_{ij}^{\normalfont{\textsf{L}}} =0$ for all $(i,j)\in\hat{\cS}_{\normalfont{\textsf{L}}}^c$, and hence
\begin{align*}
	\hat{\phi}_{ij}^*
	= \bu_i^*\bX'\bomega_j/T
	+ \bu_i^*\bX'(\hat\bomega_j-\bomega_j)/T
	+ \left(\hat{\bphi}_{i}^{\normalfont{\textsf{L}}}- \hat\bphi_{i}^{\normalfont{\textsf{L}}*} \right)(\hat{\bSigma}_x\hat\bomega_j-\be_j)
\end{align*} 
for all $(i,j)\in\hat{\cS}_{\normalfont{\textsf{L}}}^c$. 
We have also defined  
\begin{align*}
	\bbQ^*(\tt) &= \frac{1}{|\hat{\cS}_{\normalfont{\textsf{L}}}^c|}\sum_{(i,j)\in \hat{\cS}_{\normalfont{\textsf{L}}}^c}  \Pro^*\left( \tT_{ij}^* > \tt \right)~~\text{(conditional on the observations)},
\end{align*}
where $|\hat{\cS}_{\normalfont{\textsf{L}}}^c|=p-\hat{s}$ with $p=KN^2$, and $\Pro^*$ indicates the probability measure induced by the bootstrap.

\begin{proof}
	It is sufficient to consider the case when $\tt_0$ is computed by \eqref{proc2:t0}. From the beginning of \textit{Case 2} in the proof of Theorem \ref{thm:fdr-t} with the same argument, the FDR is decomposed as 
	\begin{align*}
		&\dFDP(\tt_0) \\
		&= \frac{p\left\{\bbQ^*(\tt_0)+1-\bbQ^*(-\tt_0)\right\}}{|\hat{\cS}(\tt_0)|\vee 1}\cdot\frac{2Q(\tt_0)}{\bbQ^*(\tt_0)+1-\bbQ^*(-\tt_0)} \\
		&\qquad \qquad \qquad \qquad \qquad \cdot \frac{\sum_{(i,j)\in\cI_{\leq0}}1\{\tT_{ij}\geq \tt_0\} + \sum_{(i,j)\in\cI_{\geq0}}1\{\tT_{ij}\leq -\tt_0\}}{2pQ(\tt_0)} \\
		&\leq q \cdot\frac{2}{\bbQ^*(\tt_0)/Q(\tt_0)+\{1-\bbQ^*(-\tt_0)\}/Q(\tt_0)}\left\{ V + 1 + o(1)\right\},
	\end{align*}
	where $V$ is defined in the proof of Theorem \ref{thm:fdr-t} and has been shown $V=o_p(1)$. Hence, it suffices to prove that the event, 
	\begin{align*}
		\sup_{\tt\in[0,\bar{\tt}]}\left|\frac{\bbQ^*(\tt)}{Q(\tt)}-1\right| +  \sup_{\tt\in[0,\bar{\tt}]}\left|\frac{1-\bbQ^*(-\tt)}{Q(\tt)}-1\right| =o(1),
	\end{align*}
	occurs with high probability. 
	
	Define 
	\begin{align*}
		S_{ij}^* &= T^{-1/2}\sum_{t=1}^T{u}_{it}^*\bx_t'\hat{\bomega}_{j}, \\
		R_{ij}^* &= T^{-1/2}\bu_i^*\bX'(\hat{\bomega}_{j}-{\bomega}_{j}) + \sqrt{T}(\hat{\bphi}_i^L-\hat{\bphi}_i^{L*})(\hat{\bSigma}_x\hat{\bomega}_{j}-\be_j), \\	\hat{m}_{ij}^{*2}&=\hat{\sigma}_i^{*2}\hat{\bomega}_{j}'\hat{\bSigma}_x\hat{\bomega}_{j},~~~
		\tilde{m}_{ij}^{*2} = T^{-1}\sum_{t=1}^T{u}_{it}^{*2}\hat{\bomega}_{j}'\bx_t\bx_t'\hat{\bomega}_{j},
	\end{align*}
	
	Then by the construction, we have $\tT_{ij}^*=(S_{ij}^*+R_{ij}^*)/\hat{m}_{ij}^{*}$. 
	We first check if $\tT_{ij}^*$ is asymptotically equivalent to the self-normalized sum,  $S_{ij}^*/\tilde{m}_{ij}^*$, and then verify that it can be uniformly normally approximated in the relative error. 
	
	Define the events, 
	\begin{align*}
		A_1^* = \left\{\max_{i\in[N]}\max_{j\in[KN]}\left|\frac{\tilde{m}_{ij}^{*}}{\hat{m}_{ij}^*} -1 \right| \lesssim \bar{\mu}_1 \right\}, ~~~
		A_2^* = \left\{\max_{i\in[N]}\max_{j\in[KN]} \left|\frac{R_{ij}^*}{\hat{m}_{ij}^*}\right| \lesssim \bar{\mu}_2 \right\},
	\end{align*}
	where
	\begin{align*}
		\bar{\mu}_1 &= \max\left\{\bar{\tau}_1,\bar{\tau}_2,\bar{\tau}_3\right\}, \\
		\bar{\mu}_2 &= \max\left\{ \lambda^{2-r} M_\omega^{2-2r}s_\omega,\ \bar{s}\lambda M_\omega \log^{3/2}(N\vee T) \right\}
	\end{align*}
	with $\bar{\tau}_1$, $\bar{\tau}_2$, and $\bar{\tau}_3$ defined in Lemmas \ref{lem:sigma*}, \ref{lem:u*2-hhat2}, and \ref{lem:uhat2-sig2}, respectively. 
	Then the event,
	\begin{align*}
		A := \left\{ \hat{\cS}_{\normalfont{\textsf{L}}}\supset \cS \right\}
		\cap \left\{ \Pro^*\left( A_1^{*c} \right) 
		= O((NT)^{-\nu}) \right\} 
		\cap 
		\left\{ \Pro^*\left( A_2^{*c} \right)=O((NT)^{-\nu}) \right\},
	\end{align*}
	occurs with high probability by Lemmas \ref{lem:boot}, \ref{lem:boot2}, and \ref{lem:screen}. 
	By a simple computation, we have $\bar{\mu}_1\vee\bar{\mu}_2\asymp \bar{\mu}$ with 
	\begin{align*}
		\bar{\mu} &=  \max\left\{ 
		M_\omega^{3-2r}\lambda^{1-r}s_\omega \log^2(N\vee T),~ 
		M_\omega\bar{s}\lambda \log^{3/2}(N\vee T),~ 
		M_\omega^2\lambda \log(N\vee T) \right\}.
	\end{align*}
	Because $\bar{\mu}$ tends to zero polynomially, we observe that conditional on $A$,
	\begin{align}
		\bbQ^*(\tt) 
		&= \frac{1}{|\hat{\cS}_{\normalfont{\textsf{L}}}^c|}\sum_{(i,j)\in \hat{\cS}_{\normalfont{\textsf{L}}}^c} \Pro^*\left(\frac{S_{ij}^*+R_{ij}^*}{\hat{m}_{ij}^{*}}>\tt\right) \notag \\
		&= \frac{1}{|\hat{\cS}_{\normalfont{\textsf{L}}}^c|}\sum_{(i,j)\in \hat{\cS}_{\normalfont{\textsf{L}}}^c} \Pro^*\left(\frac{S_{ij}^*}{\tilde{m}_{ij}^{*}} \left\{ 1 +  \left(\frac{\tilde{m}_{ij}^{*}}{\hat{m}_{ij}^*} -1 \right) \right\}>\tt-\frac{R_{ij}^*}{\hat{m}_{ij}^{*}},~ A_1^*\cap A_2^*\right) \notag\\
		&\qquad \qquad \qquad \qquad \qquad \qquad \qquad \qquad \qquad \qquad + O((N\vee T)^{-\nu})  \notag \\
		&=\frac{1}{|\hat{\cS}_{\normalfont{\textsf{L}}}^c|}\sum_{(i,j)\in \hat{\cS}_{\normalfont{\textsf{L}}}^c} \Pro^*\left( \frac{S_{ij}^*}{\tilde{m}_{ij}^{*}}  > \frac{\tt - O(\bar{\mu}_2)}{1+O(\bar{\mu}_1)} \right) + O((N\vee T)^{-\nu}) \label{(1)}
	\end{align}
	with high probability, where the terms, $O(\bar{\mu}_1)$, $O(\bar{\mu}_2)$, and $O((N\vee T)^{-\nu})$, depend neither on $\tt$ nor $(i,j)$. 
	Since $(i,j)\in\hat{\cS}_{\normalfont{\textsf{L}}}^c$ implies $(i,j)\in\cS^c$ on event $A$, Lemma \ref{Jiang} entails the normal approximation of the self-normalized sum,
	\begin{align}
		&\Pro^*\left( \frac{S_{ij}^*}{\tilde{m}_{ij}^{*}}  > \frac{\tt - O(\bar{\mu}_2)}{1+O(\bar{\mu}_1)} \right) \notag \\
		&= Q\left(\frac{\tt - O(\bar{\mu}_2)}{1+O(\bar{\mu}_1)}\right)\left\{ 1+O\left( \frac{M_\omega\log NT}{T^{1/2}} \right)\left(1+\frac{\tt-O(\bar{\mu}_2)}{1+O(\bar{\mu}_1)}\right)^{3} \right\},\label{(2)}
	\end{align}
	with high probability. 
	Combining \eqref{(1)} and \eqref{(2)} and dividing the both sides by $Q(\tt)$ yield, with high probability,
	\begin{align*}
		\frac{\bbQ^*(\tt)}{Q(\tt)}
		&= \frac{Q\left(\frac{\tt - O(\bar{\mu}_2)}{1+O(\bar{\mu}_1)}\right)}{Q(\tt)}\left\{ 1+O\left( \frac{M_\omega\log NT}{T^{1/2}} \right)\left(1+\frac{\tt-O(\bar{\mu}_2)}{1+O(\bar{\mu}_1)}\right)^{3} \right\} \\
		&+\frac{O((N\vee T)^{-\nu})}{Q(\tt)}.
	\end{align*}
	By the assumed condition, $\bar{\mu}_1$, $\bar{\mu}_2$ and $M_\omega/T^{1/2}$  decay polynomially  while $\tt\in[0,\bar{\tt}]$ at most diverges logarithmically. Thus we obtain uniformly in $\tt\in[0,\bar{\tt}]$,
	\begin{align*}
		1+O\left( \frac{M_\omega\log NT}{T^{1/2}} \right)\left(1+\frac{\tt-O(\bar{\mu}_2)}{1+O(\bar{\mu}_1)}\right)^{3} 
		=1 + o(1)
	\end{align*}
	and, 
	by an application of Lemma 7.2 of \cite{JJ2019}, 
	\begin{align*}
		\frac{Q\left(\frac{\tt - O(\bar{\mu}_2)}{1+O(\bar{\mu}_1)}\right)}{Q(\tt)}
		= 1 + O(\bar{\mu}_2)(1+\tt)+O(\bar{\mu}_1)(1+\tt^2)
		= 1 + o(1)
	\end{align*} 
	uniformly in $\tt\in[0,\bar{\tt}]$. 
	By \eqref{szarek}, we have
	\begin{align*}
		\sup_{\tt\in[0,\bar{\tt}]}\frac{O((N\vee T)^{-\nu})}{Q(\tt)}
		=\frac{O((N\vee T)^{-\nu})}{Q(\bar{\tt})}
		=O((N\vee T)^{-\nu+2}).
	\end{align*}

	Consequently, it holds that
	\begin{align*}
		\sup_{\tt\in[0,\bar{\tt}]}\left| \frac{\bbQ^*(\tt)}{Q(\tt)} - 1 \right| 
		=o(1)
	\end{align*}
	with high probability for any $\nu>2$. The same argument is applied to showing $|\{1-\bbQ^*(-\tt)\}/Q(\tt)-1|=o(1)$ uniformly in $\tt\in[0,\bar{\tt}]$. This completes the proof. 
\end{proof}

\subsection{Proof of Theorem \ref{thm:power-t}}

\begin{proof}
	We use the same notation as in the proof of Theorem \ref{thm:fdr-t}. 
	Let $\bar{\tt}_0=\sqrt{2\log p}$ denote the upper bound of the critical value, $\tt_0$. By the definition of power and monotonicity of probability, we have
	\begin{align*}
		\Po(\tt_0)
		&= \E\left[ \frac{|\{(i,j)\in\hat{\cS}(\tt_0):\sgn(\hat{\phi}_{ij})=\sgn(\phi_{ij})\}|}{s\vee 1} \right]\\
		&= \frac{1}{s}\sum_{(i,j)\in\cS_{<0}}\Pro\left( \tT_{ij} \leq -\tt_0 \right) 
		+ \frac{1}{s}\sum_{(i,j)\in\cS_{>0}}\Pro\left( \tT_{ij} \geq \tt_0 \right) \\
		&\geq \frac{1}{s}\sum_{(i,j)\in\cS_{<0}}\Pro\left( \tT_{ij} \leq -\bar{\tt}_0 \right) 
		+ \frac{1}{s}\sum_{(i,j)\in\cS_{>0}}\Pro\left( \tT_{ij} \geq \bar{\tt}_0 \right).
	\end{align*}
	Consider the second probability in the lower bound. Proposition \ref{thm:asynormal} gives
	\begin{align*}
		\tT_{ij}=\sqrt{T}\hat{\phi}_{ij}/\hat{m}_{ij}=(\sqrt{T}\phi_{ij}+z_{ij}+r_{ij})/\hat{m}_{ij},
	\end{align*}
	where $\phi_{ij}=0$ if and only if $(i,j)\in\cS^c$. It holds that
	\begin{align*}
		&\max_{(i,j)\in\cS_{>0}}\Pro\left( \tT_{ij} \leq \bar{\tt}_0\right) \\
		&\leq \max_{(i,j)\in\cS_{>0}}\Pro\left( \frac{z_{ij}+r_{ij}}{\hat{m}_{ij}} + \frac{m_{ij}}{\hat{m}_{ij}}\frac{\sqrt{T}\phi_{ij}^0}{m_{ij}} \leq \bar{\tt}_0,~ \frac{m_{ij}}{\hat{m}_{ij}}>\frac{1}{2} \right) +\max_{(i,j)\in\cS_{>0}}\Pro\left( \frac{m_{ij}}{\hat{m}_{ij}}\leq \frac{1}{2} \right) \\
		&\leq \max_{(i,j)\in\cS_{>0}}\Pro\left(\cZ_{ij}\leq \bar{\tt}_0-\min_{(i,j)\in\cS_{>0}}\frac{\sqrt{T}\phi_{ij}^0}{2m_{ij}} +\delta_1\right) + O\left((N\vee T)^{-\nu}\right) \\
		&\leq \Phi\left(-\sqrt{2\log p} + \delta_1 \right) + O\left((N\vee T)^{-\nu}\right),
	\end{align*}
	where the second inequality follows from Lemmas \ref{lem:omegahat1} and \ref{lem:strongapp}, and the third inequality holds by $\bar{\tt}_0=\sqrt{2\log p}$ and Condition \ref{ass:betamin}. Since $\delta_1\sqrt{\log p}=o(1)$ by the assumption, the Gaussian probability in the upper bound tends to zero. The same result is obtained for the first probability. Thus the power goes to unity because $|\cS_{<0}|+|\cS_{>0}|=s$. This completes the proof. 
\end{proof}

\subsection{Proof of Theorem \ref{thm:rob_fdr_pwr}}
\begin{proof}[Proof of (a)]
	For any $(i,j)\in\hat{\cS}_R(h^*)$, we have $qE_{h(i,j)}/|{\cH}|\geq 1/|\hat{\cS}_R(h^*)|$ by the property of e-BH. Thus the FDR of $\hat{\cS}:=\hat{\cS}_R(h^*)$ is bounded by
	\begin{align*}
		\FDR &= \E\left[\sum_{(i,j)\in\cS^c}\frac{1\{(i,j)\in\hat{\cS}\}}{|\hat{\cS}|}\right] 
		\leq \E\left[\sum_{(i,j)\in\cS^c}\frac{qE_{h(i,j)}1\{(i,j)\in\hat{\cS}\}}{|{\cH}|}\right] \\
		&\leq \frac{1}{|{\cH}|}\sum_{(i,j)\in\cS^c}q\E[E_{h(i,j)}]
		\leq q\max_{(i,j)\in\cS^c}\E\left[\frac{f(\tT_{ij})}{\E f(\cZ)}\right],
	\end{align*}
	where $\cZ\sim N(0,1)$. 
	Theorem \ref{thm:t} with Conditions \ref{ass:subG}--\ref{ass:mineig2} and $\bar{v}=o(1)$ achieves  $f(\tT_{ij})\to_d f(\cZ)$ for all $(i,j)\in\cS^c$ for any continuous function $f$ by the continuous mapping theorem. Therefore, for every $(i,j)\in \cS^c$, the assumed uniform integrability of $f(\tT_{ij})$ implies  $\E[f(\tT_{ij})]\to \E[f(\cZ)]$. This makes the upper bound of FDR be $q+o(1)$, which completes the proof.
\end{proof}

\begin{proof}[Proof of (b)]
	For any $(i,j)\in \hat{\cS}_R(h^*)$, 
	we have $qE_{h(i,j)}\geq |{\cH}|/|\hat{\cS}_R(h^*)|$. This implies $qE_{(h^*)}/|{\cH}|\geq 1/|\hat{\cS}|$.
	Therefore, the power is lower bounded by
	\begin{align*}
		\Po &= \E\left[ \sum_{(i,j)\in\cS}\frac{1\{(i,j)\in\hat{\cS}\}}{s} \right]
		= \frac{1}{s} \sum_{(i,j)\in\cS} \Pro\left(E_{h(i,j)}\geq E_{(h^*)} \right) \\
		&\geq \frac{1}{s} \sum_{(i,j)\in\cS} \Pro\left(E_{h(i,j)}\geq \frac{|{\cH}|}{q} \right)
		\geq \min_{(i,j)\in\cS} \Pro\left(f(\tT_{ij})\geq \frac{\E f(\cZ)|{\cH}|}{q} \right) \\
		&\geq \min_{(i,j)\in\cS} \Pro\left( |\tT_{ij}|\gtrsim f^{-1}(|{\cH}|) \right),
	\end{align*}
	where $f^{-1}$ exists by the monotonicity. 
	From Theorem \ref{thm:t} and its proof, $|\tT_{ij}|/\sqrt{T}$ is uniformly tight for every $(i,j)\in\cS$ when $\bar{v}=o(1)$. Therefore, the lower bound of $\text{Power}$ tends to unity by the condition of $f$. This completes the proof. 
\end{proof}

\section{Proofs of Propositions}

\subsection{Proof of Proposition \ref{thm:errbound}}
\begin{proof}
	Derive the non-asymptotic error bound for the Lasso estimator. 
	First define two events:
	\begin{align*}
		\mathcal{E}_1 = \left\{ \left\| T^{-1}\bX\bU' \right\|_{\max} \leq \lambda/2 \right\}, ~~~
		\mathcal{E}_2 = \left\{ \left\|\hat{\bSigma}_x-\bSigma_x\right\|_{\max} \leq b\lambda/2 \right\},
	\end{align*}
	where $\lambda=8bc_{uu}\sqrt{2(\nu+7)^3\log^3(N\vee T)/T}$ with any positive constant $\nu$. We work on event $\cE_1\cap \cE_2$ since Lemmas \ref{lem:tail_yu} and \ref{lem:tail_yy} guarantee that $\cE_1\cap \cE_2$ occurs with  probability at least $1-O((N\vee T)^{-\nu})$.

	Define $\bdelta_{i\cdot} = \hat{\bphi}_{i\cdot}^{\normalfont{\textsf{L}}}-\bphi_{i\cdot}$. 
	Because $\hat{\bphi}_{i\cdot}^{\normalfont{\textsf{L}}}$ minimizes the objective function, we have 
	\begin{align*}
		(2T)^{-1}\| \by_{i\cdot} - \hat{\bphi}_{i\cdot}^{\normalfont{\textsf{L}}}\bX \|_{2}^2 + \lambda\|\hat{\bphi}_{i\cdot}^{\normalfont{\textsf{L}}}\|_1 
		&\leq (2T)^{-1}\| \by_{i\cdot}- \bphi_{i\cdot}\bX \|_{2}^2 + \lambda\|\bphi_{i\cdot} \|_1, 
	\end{align*}
	which is equivalently written as
	\begin{align*}
		(2T)^{-1}\|\bdelta_{i\cdot}\bX \|_{2}^2  
		&\leq T^{-1}\bu_{i\cdot}\bX'\bdelta_{i\cdot}' + \lambda\|\bphi_{i\cdot} \|_1- \lambda\|\hat{\bphi}_{i\cdot}^{\normalfont{\textsf{L}}}\|_1. 
	\end{align*}
	By H\"older's inequality and the triangle inequality, we have
	\begin{align*}
		(2T)^{-1}\| \bdelta_{i\cdot}\bX \|_{2}^2  
		&\leq \left| T^{-1}\bu_{i\cdot}\bX'\bdelta_{i\cdot}'\right| + \lambda\| \bphi_{i\cdot} \|_1 
		- \lambda\| \hat{\bphi}_{i\cdot}^{\normalfont{\textsf{L}}}\|_1 \\
		&\leq \| T^{-1}\bu_{i\cdot}\bX'\|_{\infty}  \|\bdelta_{i\cdot}\|_1
		+ \lambda \|\bdelta_{i\cdot}\|_1.
	\end{align*}
	On event $\cE_1$, we thus obtain the upper bound of $\| \bdelta_{i\cdot}\bX \|_{2}^2$ as
	\begin{align}
		T^{-1}\| \bdelta_{i\cdot}\bX \|_{2}^2  
		\leq 3\lambda \| \bdelta_{i\cdot}\|_1. \label{1234}
	\end{align}

	Next we bound $\|\bdelta_{i\cdot}\bX\|_{2}^2$ from below. 
	Lemma \ref{lem:negahban} states that $\bdelta_{i\cdot}$ lies in $\cD=\{\bdelta_{i\cdot}\in\mathbb{R}^{1\times KN}:\|\bdelta_{i\cS_i^c}\|_1\leq 3\|\bdelta_{i\cS_i}\|_1\}$ on $\mathcal{E}_1$. Hence under $\mathcal{E}_1$ and $\mathcal{E}_2$, Lemma \ref{lem:E2} entails 
	\begin{align}
		T^{-1}\|\bdelta_{i\cdot}\bX \|_{2}^2 /\|\bdelta_{i\cdot}\|_{2}^2 
		\geq \gamma -16s_i\|\hat{\bSigma}_x-\bSigma_x\|_{\max} 
		\geq \gamma-8bs_i\lambda. \label{basic4}
	\end{align}
	By \eqref{1234}, \eqref{basic4}, Lemma \ref{lem:negahban}, and the Cauchy--Schwarz inequality, we have
	\begin{align*}
		(\gamma-8bs_i\lambda) \| \bdelta_{i\cdot}\|_{2}^2
		&\leq 3\lambda \|\bdelta_{i\cdot} \|_1 
		= 3\lambda \|\bdelta_{i\cS_i}\|_1 + 3\lambda \|\bdelta_{i\cS_i^c}\|_1 \\
		&\leq 12 \lambda \|\bdelta_{i\cS_i}\|_1 
		\leq 12 s_i^{1/2} \lambda \|\bdelta_{i\cS_i}\|_2
		\leq 12 s_i^{1/2} \lambda \|\bdelta_{i\cdot}\|_{2}.
	\end{align*}
	This concludes that
	\begin{align*}
		\|\bdelta_{i\cdot}\|_{2}
		\leq\frac{ 12s_i^{1/2} \lambda}{\gamma-8bs_i\lambda}. 
	\end{align*}

	Next we derive the prediction error bound. 
	By Lemma \ref{lem:negahban} and the Cauchy--Schwarz inequality again, the error bound in the element-wise $\ell_1$-norm is given by
	\begin{align}
		\|\bdelta_{i\cdot}\|_1 
		&= \|\bdelta_{i\cS_i}\|_1 + \|\bdelta_{i\cS_i^c}\|_1 \notag \\
		&\leq 4 \|\bdelta_{i\cS_i}\|_1 
		\leq 4s_i^{1/2}\|\bdelta_{i\cS_i}\|_{2} 
		\leq 4s_i^{1/2}\|\bdelta_{i\cdot}\|_{2} 
		\leq \frac{48s_i \lambda}{\gamma-8bs_i\lambda}. \label{basic8}
	\end{align}
	From \eqref{1234} and \eqref{basic8}, the prediction error bound is obtained as
	\begin{align*}
		T^{-1}\| \bdelta_{i\cdot}\bX \|_{2}^2  
		\leq 3 \lambda \|\bdelta_{i\cdot}\|_1 
		\leq \frac{144s_i \lambda^2}{\gamma-8bs_i\lambda}. 
	\end{align*}
	This completes the proof. 
\end{proof}

\subsection{Proof of Proposition \ref{thm:asynormal}}
\begin{proof}
	The first assertion is trivial by the construction of the debiased lasso estimator. We derive the upper bound of $|r_{ij}|$. Observe that  for each $i\in[N]$,
	\begin{align*}
		|r_{ij}| &\leq \left\|T^{-1/2}\bu_{i\cdot}\bX'\right\|_\infty\left\|\hat{\bomega}_{j}-{\bomega}_{j}\right\|_1
		+ \sqrt{T}\left\|\hat\bphi_{i\cdot}^{\normalfont{\textsf{L}}} - \bphi_{i\cdot} \right\|_1\left\|\hat{\bSigma}_{x}\hat{\bomega}_{j}-\be_{j}\right\|_\infty \\
		&\leq \left\|T^{-1/2}\bU\bX'\right\|_{\max}\max_j\left\|\hat{\bomega}_{j}-{\bomega}_{j}\right\|_1
		+ \sqrt{T}\left\|\hat\bphi_{i\cdot}^{\normalfont{\textsf{L}}} - \bphi_{i\cdot} \right\|_1\left\|\hat{\bSigma}_{x}\hat\bOmega-\bI_{KN}\right\|_{\max} 
	\end{align*}
	uniformly in $j\in[KN]$. 
	Let $\lambda_1=b\|\bOmega\|_{\ell_1}\lambda/2$, and consider the event 
	\begin{align*}
		\{\|\bX\bU'/T\|_{\max}\leq \lambda/2\}\cap \{\|\hat{\bSigma}_x-\bSigma_x\|_{\max}\leq \lambda_1/\|\bOmega\|_{\ell_1}\}. 
	\end{align*}
	Then this occurs with probability at least $1-O((N\vee T)^{-\nu})$ by Lemmas \ref{lem:tail_yu} and \ref{lem:tail_yy} with the proof of Proposition \ref{thm:errbound}. On this event, the proof of Theorem 6 in \cite{CaiEtAl2011} establishes the bounds, $\max_{j}\|\hat{\bomega}_{j}-{\bomega}_{j}\|_1\lesssim (\|\bOmega\|_{\ell_1}\lambda_1)^{1-r}s_\omega$ and $\|\hat\bSigma_x\hat\bOmega-\bI_{KN}\|_{\max}\leq \lambda_1$, under Condition \ref{ass:invest}. Therefore, together with the Lasso bound derived in Proposition \ref{thm:errbound} and $\lambda_1\lesssim M_\omega\lambda$, we obtain
	\begin{align*}
		\max_{j\in[KN]}|r_{ij}|
		&\lesssim \sqrt{T}\lambda(\|\bOmega\|_{\ell_1}\lambda_1)^{1-r}s_\omega/2 
		+ \frac{48s_i\sqrt{T} \lambda\lambda_1}{\gamma-8bs_i\lambda} \\
		&\lesssim \sqrt{T}\lambda M_\omega^{1-r}(M_\omega\lambda)^{1-r}s_\omega + s_i\sqrt{T}\lambda^2 M_\omega \\
		&\lesssim \left(s_\omega M_\omega^{2-2r}\lambda^{1-r}+s_iM_\omega\lambda\right)\sqrt{\log^3(N\vee T)}~(=:\bar{r}_i)
	\end{align*}
	for each $i\in[N]$ such that $s_i\lambda=o(1)$. This completes the proof. 
\end{proof}

\section{Lemmas and their Proofs}

\begin{lem}\label{lem:subE}
	If Condition \ref{ass:subG} is true, then for any $i,j\in[N]$ and $s,t\in[T]$, there exist constants $c_{uu},c_m>0$ such that
	\begin{align*}
		(a)~~&\Pro\left(\left|\frac{1}{T}\sum_{t=1}^T(u_{it}u_{jt}-\E[u_{it}u_{jt}])\right|> x\right) \leq 2\exp\left\{ -\frac{1}{2}\left(\frac{Tx^2}{c_{uu}^2}\wedge\frac{Tx}{c_{uu}}\right)\right\},\\
		(b)~~&\E\left[ \max_{i\in[N]}\max_{t\in[T]}|u_{it}|^m \right] \leq c_m\log^{m/2} (N\vee T),
	\end{align*}
	where $m\in\mathbb{N}$ is an arbitrary fixed constant. 
\end{lem}
\begin{proof}[Proof of Lemma \ref{lem:subE}]
	Prove $(a)$. Since $u_{it}$ is sub-Gaussian by Condition \ref{ass:subG}, Lemma 2.7.7 of \cite{Vershynin2018} entails that for any $i,j\in[N]$,  $\{u_{it}u_{jt}-\E[u_{it}u_{jt}]\}_t$ is a sequence of i.i.d.\ (centered) sub-exponential random variables. Thus the result directly follows by Bernstein's inequality \citep[Theorem 2.8.1]{Vershynin2018}. 
	
	Prove $(b)$. For any fixed $m\in\mathbb{N}$, the sub-Gaussian tail property for $u_{it}$ in Condition \ref{ass:subG} implies for all $x>0$, 
	\begin{align*}
		\Pro(|u_{it}|^m>x)=\Pro(|u_{it}|>x^{1/m})\leq 2\exp(-x^{2/m}/c_u).
	\end{align*}
	Using this tail probability with the union bound, we have
	\begin{align*}
		&\E\left[ \max_{i,t\in\mathbb{N}}\frac{|u_{it}|^m}{(1+\log it)^{m/2}} \right]
		\leq \int_0^\infty \Pro\left( \max_{i,t\in\mathbb{N}}\frac{|u_{it}|^m}{(1+\log it)^{m/2}} > x \right) \diff x \\
		&\quad \leq \int_0^{(2c_u)^{m/2}} \diff x + 2\sum_{i=1}^\infty\sum_{t=1}^\infty \int_{(2c_u)^{m/2}}^\infty \exp\left( -\frac{x^{2/m}(1+\log i+\log t)}{c_u} \right) \diff x \\
		&\quad \leq (2c_u)^{m/2} + 2\sum_{i=1}^\infty i^{-2}\sum_{t=1}^\infty t^{-2} \int_{(2c_u)^{m/2}}^\infty \exp\left( -\frac{x^{2/m}}{c_u} \right) \diff x,
	\end{align*}
	where the upper bound is further bounded by some positive constants. Thus there exists some constant $M>0$ such that 
	\begin{align*}
		M &\geq \E\left[ \max_{i,t\in\mathbb{N}}\frac{|u_{it}|^m}{(1+\log it)^{m/2}} \right] \\
		&\geq \E\left[ \max_{i\in[N],t\in[T]}\frac{|u_{it}|^m}{(1+\log it)^{m/2}} \right]
		\geq \E\left[ \max_{i\in[N],t\in[T]}|u_{it}|^m \right]\frac{1}{(1+\log NT)^{m/2}}. 
	\end{align*}
	Replacing the constant factor appropriately gives the result. This completes the proof. 
\end{proof}

\begin{lem}\label{lem:VMA}
	If Condition \ref{ass:stab} is true, then there exists a constant $\delta\in(0,1)$ such that for any monotonically divergent sequence $r_T>0$ with sufficiently large $T$,
	\begin{align}
		\sum_{\ell=r_T}^\infty \|\bB_\ell\|_{\infty} \leq \frac{\delta^{r_T}}{1-\delta}.\label{cond:sum}
	\end{align}
	In particular, the summability condition in \eqref{cond:sum0} holds.
\end{lem}
\begin{proof}[Proof of Lemma \ref{lem:VMA}]
	The VAR($K$) model is written as the VAR(1) companion form with coefficient
	\begin{align*}
		\bA = 
		\begin{pmatrix}
			\bPhi_1 & \bPhi_2 & \dots & \bPhi_{K-1} & \bPhi_K \\
			\bI_{N}	& \bzero 	  & \dots & \bzero & \bzero \\
			\bzero 		& \bI_{N} 	& \dots & \bzero & \bzero \\
			\vdots      & \vdots      &       & \vdots & \vdots \\
			\bzero      & \bzero      & \dots & \bI_{N}   & \bzero
		\end{pmatrix}.
	\end{align*}
	Condition \ref{ass:stab} entails that the spectral radius of $\bA$, defined as $\rho(\bA)=\max_{j\in[KN]}|\lambda_j(\bA)|$, is strictly less than one uniformly in $N$. This implies that the VAR(1) model is invertible to the VMA($\infty$). Therefore, we have the representation $\by_t = \sum_{\ell=0}^\infty \bB_\ell \bu_{t-\ell}$, where $\bB_0=\bI_N$ and $\bB_\ell=\bJ'\bA^\ell\bJ$ with $\bJ'=(\bI_N, \bzero_{N\times(KN-N)})$ for $\ell=1,2,\dots$. Note that 
	\begin{align*}
		\sum_{\ell=0}^\infty \|\bB_\ell\|_{\infty} 
		\leq \sum_{\ell=0}^\infty \|\bJ\|_{\infty}^2 \|\bA^\ell\|_{\infty} 
		= \sum_{\ell=0}^\infty \|\bA^\ell\|_{\infty}. 
	\end{align*}
	Again, since we have $\rho(\bA)<1$ uniformly in $N$ by Condition \ref{ass:stab}, we can always pick up a small constant $\ep>0$ such that  $\delta:=\rho(\bA)+\ep$ lies in $(0,1)$. We apply the Gelfand formula \citep{HornJohnson2012}; for this choice of $\ep$, there exists $T^*$ such that $\|\bA^{r_T}\|_\infty\leq \delta^{r_T}$ for all $T\geq T^*$, where $r_T>0$ is any monotonically diverging sequence as $T\to \infty$.
	Therefore, for any $T\geq T^*$ we have
	\begin{align*}
		\sum_{\ell=r_T}^\infty \|\bA^\ell\|_{\infty}
		\leq \sum_{\ell=r_T}^\infty\delta^\ell
		=\frac{\delta^{r_T}}{1-\delta},
	\end{align*}
	which proves \eqref{cond:sum}. 
	Finally, because $\limsup_{\ell\to\infty}\|\bA^\ell\|_\infty^{1/\ell}\leq \delta<1$, 
	the summability condition in \eqref{cond:sum0} holds by the root test. 
	This completes the proof. 
\end{proof}

\begin{lem}\label{lem:tail_yu}
	Let $b=\sum_{\ell=0}^{\infty}\left\|\bB_\ell\right\|_\infty$. 
	If Conditions \ref{ass:subG} and \ref{ass:stab} are true, then for any $\bar{c}_1>0$ and arbitrary divergent sequence $r_T>0$ with sufficiently large $T$, we have for all $x>0$,
	\begin{align*}
		&\Pro\left( \left\|T^{-1}\bX\bU'\right\|_{\max} > x \right) \\
		&\leq 2KN^2 \exp\left(-\frac{x^2T}{2\bar{c}_1^2}\right)
		+ 2r_TKN^2T\exp\left(-\frac{\bar{c}_1}{4bc_{uu}}\right)
		+ 2c_1^2KT\frac{\delta^{r_T}\log N}{\bar{c}_1(1-\delta)},
	\end{align*}
	where constants $\delta\in(0,1)$ and $c_{uu},c_1>0$ are given by Lemmas \ref{lem:subE} and \ref{lem:VMA}, respectively.
	In particular, if we set $x=\sqrt{2\bar{c}_1^2(\nu+3)\log(N\vee T)/T}$, $\bar{c}_1=4bc_{uu}(\nu+5)\log (N\vee T)$, and $r_T=T$ with any constant $\nu>0$, then the upper bound becomes $O((N\vee T)^{-\nu})$. 
\end{lem}
\begin{proof}[Proof of Lemma \ref{lem:tail_yu}]
	Following the definition, we have for any $x,\bar{c}_1>0$,
	\begin{align*}
		&\Pro\left(\left\|T^{-1}\bX\bU'\right\|_{\max}>x\right) 
		= \Pro\left(\max_{k\in[K]}\left\|T^{-1}\sum_{t=1}^T\by_{t-k}\bu_{t}'\right\|_{\max}>x \right)\\
		&\leq \Pro\left(\max_{k\in[K]}\left\|T^{-1}\sum_{t=1}^T\by_{t-k}\bu_{t}'\right\|_{\max}>x \mid \max_{k\in[K]}\max_{t\in[T]}\|\by_{t-k}\bu_{t}'\|_{\max}\leq \bar{c}_1\right) \\
		&\qquad + \Pro\left(\max_{k\in[K]}\max_{t\in[T]}\|\by_{t-k}\bu_{t}'\|_{\max}> \bar{c}_1\right).  
	\end{align*}
	
	Consider the first term. Note for any $k\in[K]$ that $\{\by_{t-k}\bu_t'\}$ forms a martingale difference sequence with respect to the filtration, $\cF_t=\sigma\{\bu_{t-s} :s=0,1,\dots\}$. Thus the union bound and the Azuma-Hoeffding inequality \citep{Bercu2015} give
	\begin{align*}
		&\Pro\left(\max_{k\in[K]}\left\|T^{-1}\sum_{t=1}^T\by_{t-k}\bu_{t}'\right\|_{\max}>x \mid \max_{k\in[K]}\max_{t\in[T]}\|\by_{t-k}\bu_{t}'\|_{\max}\leq \bar{c}_1\right) \\
		&\leq KN^2\max_{k\in[K]}\max_{i,j\in[N]}\Pro\left(\left|T^{-1}\sum_{t=1}^Ty_{i,t-k}u_{jt}\right|>x \mid \max_{k\in[K]}\max_{t\in[T]}\|\by_{t-k}\bu_{t}'\|_{\max}\leq \bar{c}_1\right)\\
		&\leq 2KN^2 \exp\left(-\frac{x^2T}{2\bar{c}_1^2}\right).
	\end{align*}
	
	For the second term, applying the triangle and H\"older's inequalities yields
	\begin{align*}
		\left\|\by_{t-k}\bu_t'\right\|_{\max}
		=\left\|\sum_{\ell=0}^\infty \bB_\ell \bu_{t-k-\ell}\bu_t'\right\|_{\max} 
		\leq \left( \sum_{\ell=0}^{r-1} + \sum_{\ell=r}^\infty \right)\left\|\bB_\ell\right\|_\infty \left\|\bu_{t-k-\ell}\bu_{t}'\right\|_{\max},
	\end{align*}
	and hence, 
	\begin{align*}
		&\Pro\left(\max_{k\in[K]}\max_{t\in[T]}\|\by_{t-k}\bu_{t}'\|_{\max}> \bar{c}_1\right) \\
		&\leq \Pro\left(\max_{k\in[K]}\max_{t\in[T]} \sum_{\ell=0}^{r-1} \left\|\bB_\ell\right\|_\infty \left\|\bu_{t-k-\ell}\bu_{t}'\right\|_{\max}> \bar{c}_1/2 \right) \\
		&\qquad+ \Pro\left(\max_{k\in[K]}\max_{t\in[T]} \sum_{\ell=r}^\infty \left\|\bB_\ell\right\|_\infty \left\|\bu_{t-k-\ell}\bu_{t}'\right\|_{\max} > \bar{c}_1/2 \right). 
	\end{align*}
	Consider the probability with $\sum_{\ell=0}^{r-1}$. The union bound and the summability condition in Lemma \ref{lem:VMA} yield
	\begin{align*}
		&\Pro\left( \max_{k\in[K]}\max_{t\in[T]}\sum_{\ell=0}^{r-1}\left\|\bB_\ell\right\|_\infty \left\|\bu_{t-k-\ell}\bu_{t}'\right\|_{\max} > \bar{c}_1/2 \right) \\
		&\leq rKT \max_{k\in[K]}\max_{t\in[T]}\max_{\ell=0,\dots,r-1}\Pro\left(  \left\|\bu_{t-k-\ell}\bu_{t}'\right\|_{\max} > \frac{\bar{c}_1}{2\sum_{\ell=0}^{r-1}\left\|\bB_\ell\right\|_\infty} \right) \\
		&\leq rKN^2T \max_{k\in[K+r-1]}\max_{t\in[T]}\max_{i,j\in[N]}\Pro\left(  \left|u_{i,t-k}u_{jt}\right| > \frac{\bar{c}_1}{2b} \right) \\
		&\leq 2rKN^2T\exp\left(-\frac{\bar{c}_1}{4bc_{uu}}\right),
	\end{align*}
	where the last inequality holds by Lemma \ref{lem:subE}$(a)$ since $u_{i,t-k}u_{jt}$ is sub-exponential \citep[Proposition 2.7.1 and Lemma 2.7.7]{Vershynin2018}. 
	Consider the probability with $\sum_{\ell=r}^\infty$. By the union bound and the Markov inequality, we have
	\begin{align*}
		&\Pro\left( \max_{k\in[K]}\max_{t\in[T]}\sum_{\ell=r}^{\infty}\left\|\bB_\ell\right\|_\infty \left\|\bu_{t-k-\ell}\bu_{t}'\right\|_{\max} > \bar{c}_1/2 \right) \\
		&\leq  2KT\sum_{\ell=r}^{\infty}\left\|\bB_\ell\right\|_\infty \E\left[ \left\|\bu_{t-k-\ell}\bu_{t}'\right\|_{\max}\right]/\bar{c}_1\\
		&\leq 2KT\frac{\delta^r}{\bar{c}_1(1-\delta)} \E\left[\left\|\bu_{t}\right\|_{\max}\right]^2
		\leq 2c_1^2KT\frac{\delta^r\log N}{\bar{c}_1(1-\delta)},
	\end{align*}
	where the last inequality holds by Lemma \ref{lem:subE}(b). 
	Combining these upper bounds completes the proof. 
\end{proof}

\begin{lem}\label{lem:tail_yy}
	Let $b=\sum_{\ell=0}^{\infty}\left\|\bB_\ell\right\|_\infty$. 
	If Conditions \ref{ass:subG} and \ref{ass:stab} are true, then for any $\bar{c}_2\geq c_{uu}$ and arbitrary divergent sequence $r_T>0$ with sufficiently large $T$, we have for all $x>0$,
	\begin{align*}
		&\Pro\left( \left\|T^{-1}\bX\bX'-\E[T^{-1}\bX\bX']\right\|_{\max} > x \right) \\
		&\leq 2K^2r_T^2N^2\exp\left(-\frac{x^2T}{8b^4\bar{c}_2^2}\right) 
		+ 2K^2r_T^2N^2T\exp\left(-\frac{\bar{c}_2}{2c_{uu}}\right)
		+ 6K^2bc_2\frac{\delta^{r_T} \log N}{x(1-\delta)},
	\end{align*}
	where constants $\delta\in(0,1)$ and $c_{uu},c_2>0$ are given by Lemmas \ref{lem:subE} and \ref{lem:VMA}, respectively. 
	In particular, if we set $x=\sqrt{8(\nu+6)b^4\bar{c}_2^2\log(N\vee T)/T}$, $\bar{c}_2=2(\nu+7)c_{uu}\log (N\vee T)$, and $r_T=T$ with any constant $\nu>0$, the upper bound becomes $O((N\vee T)^{-\nu})$. 
\end{lem}
\begin{proof}[Proof of Lemma \ref{lem:tail_yy}]
	For $g,h\in[K]$, define
	\begin{align*}
		\bW_{g,h}=T^{-1}\sum_{t=1}^T\left(\bu_{t-g}\bu_{t-h}'-\E[\bu_{t-g}\bu_{t-h}']\right).
	\end{align*}
	By the VMA($\infty$) representation in Lemma \ref{lem:VMA}, we have
	\begin{align*}
		\by_{t-g}\by_{t-h}'=\sum_{\ell=0}^\infty\sum_{m=0}^\infty \bB_\ell\bu_{t-g-\ell}\bu_{t-h-m}'\bB_m', 
	\end{align*}
	so that
	\begin{align*}
		&\left\|T^{-1}\bX\bX'-\E[T^{-1}\bX\bX']\right\|_{\max}
		\leq \max_{g,h\in[K]}\left\|T^{-1}\sum_{t=1}^T\left(\by_{t-g}\by_{t-h}'-\E[\by_{t-g}\by_{t-h}']\right)\right\|_{\max}\\
		&\leq \max_{g,h\in[K]} \sum_{\ell=0}^\infty\sum_{m=0}^\infty \left\| \bB_\ell \bW_{g+\ell,h+m}\bB_m'\right\|_{\max} \\
		&\leq \max_{g,h\in[K]} \left( \sum_{\ell=0}^{r-1}+\sum_{\ell=r}^{\infty}\right) \left( \sum_{m=0}^{r-1}+\sum_{m=r}^{\infty}\right) 
		\|\bB_\ell\|_\infty \left\|\bW_{g+\ell,h+m}\bB_m'\right\|_{\max} \\
		&\leq \max_{g,h\in[K]} \left( \sum_{\ell=0}^{r-1}\sum_{m=0}^{r-1} + 3\sum_{\ell=0}^{\infty}\sum_{m=r}^{\infty} \right) \|\bB_\ell\|_\infty\|\bB_m\|_\infty \left\|\bW_{g+\ell,h+m}\right\|_{\max}.
	\end{align*}
	Therefore, we obtain
	\begin{align*}
		&\Pro\left( \left\|T^{-1}\bX\bX'-\E[T^{-1}\bX\bX']\right\|_{\max} > x \right) \\
		&\leq \Pro\left( \max_{g,h\in[K]} \sum_{\ell=0}^{r-1}\sum_{m=0}^{r-1} \|\bB_\ell\|_\infty\|\bB_m\|_\infty \left\|\bW_{g+\ell,h+m}\right\|_{\max} > x/2 \right) \\
		&\quad+ \Pro\left( \max_{g,h\in[K]} 3\sum_{\ell=0}^{\infty}\sum_{m=r}^{\infty} \|\bB_\ell\|_\infty\|\bB_m\|_\infty \left\|\bW_{g+\ell,h+m}\right\|_{\max} > x/2 \right).
	\end{align*}

	We consider the first probability. The union bound gives
	\begin{align*}
		&\Pro\left(\max_{g,h\in[K]} \sum_{\ell=0}^{r-1}\sum_{m=0}^{r-1} \|\bB_\ell\|_\infty\|\bB_m\|_\infty \left\|\bW_{g+\ell,h+m}\right\|_{\max} > x/2\right)\\
		&\leq K^2\max_{g,h\in[K]}\Pro\left(\max_{\ell,m=0,\dots,r-1}\left\|\bW_{g+\ell,h+m}\right\|_{\max} \sum_{\ell=0}^{r-1}\sum_{m=0}^{r-1} \|\bB_\ell\|_\infty\|\bB_m\|_\infty  > x/2\right)\\
		&\leq K^2r^2\max_{g,h\in[K]}\max_{\ell,m=0,\dots,r-1}\Pro\left(\left\|\bW_{g+\ell,h+m}\right\|_{\max} > x/(2b^2)\right)\\
		&=K^2r^2 \max_{k\in[K+r-1]}\Pro\left( \left\|\bW_{k,1}\right\|_{\max} > x/(2b^2) \right).
	\end{align*}
	If $k=1$, each component of $\bW_{1,1}$ is a sample average of the i.i.d.\ random variables, $\{u_{it}u_{jt}-\E[u_{it}u_{jt}]\}_t$. Clearly this is a martingale difference sequence with respect to filtration  $\cF_t=\sigma\{u_{i,t-s},u_{j,t-s}:s=0,1,\dots\}$. If $k\geq 2$, we have $\E[u_{i,t-k}u_{j,t-1}]=0$ and the sequence, $\{u_{i,t-k+1}u_{jt}\}_t$, is also martingale difference with respect to $\cF_t$. Therefore, the Azuma-Hoeffding inequality \citep{Bercu2015} with the conditioning argument and the union bound as in the proof of Lemma \ref{lem:tail_yu}, we have
	\begin{align*}
		&\max_{k\in[K+r-1]}\Pro\left(\left\|\bW_{k,1}\right\|_{\max} > x/(2b^2) \right) \\
		&\leq \max_{k\in[K+r-1]}\Pro\left(\left\|\bW_{k,1}\right\|_{\max} > x/(2b^2) \mid \max_{t\in[T]} \left\|\bu_{t-k}\bu_{t-1}'-\E[\bu_{t-k}\bu_{t-1}'] \right\|_{\max} \leq \bar{c}_2 \right) \\
		&\quad + \max_{k\in[K+r-1]}\Pro\left( \max_{t\in[T]}\left\|\bu_{t-k}\bu_{t-1}'-\E[\bu_{t-k}\bu_{t-1}'] \right\|_{\max} > \bar{c}_2\right) \\
		&\leq 2N^2\exp\left(-\frac{x^2T}{8b^4\bar{c}_2^2}\right) + 2N^2T\exp\left(-\frac{\bar{c}_2}{2c_{uu}}\right).
	\end{align*}

	We next consider the second probability. The union bound and the Markov inequality with Lemma \ref{lem:VMA} give
	\begin{align*}
		&\Pro\left(3\max_{g,h\in[K]} \sum_{\ell=0}^{\infty}\sum_{m=r}^{\infty} \|\bB_\ell\|_\infty\|\bB_m\|_\infty \left\|\bW_{g+\ell,h+m}\right\|_{\max} > x/2 \right)\\
		&\leq 6K^2\max_{g,h\in[K]}\sum_{\ell=0}^{\infty}\sum_{m=r}^{\infty} \|\bB_\ell\|_\infty\|\bB_m\|_\infty \E\left\|\bW_{g+\ell,h+m}\right\|_{\max}/x \\
		&\leq 6K^2\frac{b\delta^r}{x(1-\delta)}\max_{k\in\{1,2\}}\E\left[\left\|\bu_{t-k}\bu_{t-1}'-\E[\bu_{t-k}\bu_{t-1}']\right\|_{\max}\right].
	\end{align*}
	By Lemma \ref{lem:subE}, the last expectation is evaluated as
	\begin{align*} \max_{k\in\{1,2\}}\E\left[\left\|\bu_{t-k}\bu_{t-1}'-\E\bu_{t-k}\bu_{t-1}'\right\|_{\max}\right]
		\leq c_2 \log N.
	\end{align*}
	Combining the obtained results so far yields the desired inequality. This completes the proof. 
\end{proof}

\begin{lem} \label{lem:negahban}
	Let $\bdelta_{i\cdot}'\in\bbR^{KN}$ denote the $i$-th column vector of $\bDelta'\in\bbR^{KN\times N}$ with $\bDelta=\hat{\bPhi}-\bPhi^0$. If inequality \eqref{1234} is true, 
	then it holds that
	\begin{align*}
		\| \bdelta_{i\cS_i^c} \|_1\leq 3\| \bdelta_{i\cS_i} \|_1.
	\end{align*}
\end{lem}
\begin{proof}[Proof of Lemma \ref{lem:negahban}]
	Note that $\bphi=\bphi_{\cS}$ and $\bv=\bv_{\cS}+\bv_{\cS^c}$ for any $KN$-dimensional vector $\bv$. Hence the lower bound of $\|\hat{\bphi}_{i\cdot}\|_1$ is computed as 
	\begin{align*}
		\|\hat{\bphi}_{i\cdot}\|_1
		&=\|\bphi_{i\cdot}+\bdelta_{i\cdot}\|_1
		= \|\bphi_{i\cS_i}+\bdelta_{i\cS_i}+\bdelta_{i\cS_i^c}\|_1
		\geq \|\bphi_{i\cS_i} + \bdelta_{i\cS_i^c}\|_1 - \| \bdelta_{i\cS_i} \|_1\\
		&= \|\bphi_{i\cS_i}\|_1 + \|\bdelta_{i\cS_i^c}\|_1 - \| \bdelta_{i\cS_i} \|_1
		= \|\bphi_{i\cdot}\|_1 + \|\bdelta_{i\cS_i^c}\|_1 - \| \bdelta_{i\cS_i} \|_1,
	\end{align*}
	where the last equality holds by decomposability of the $\ell_1$-norm. 
	Thus we obtain
	\begin{align*}
		\|\bphi_{i\cdot}\|_1 - \|\hat{\bphi}_{i\cdot}\|_1 
		\leq \| \bdelta_{i\cS_i} \|_1 - \|\bdelta_{i\cS_i^c}\|_1.
	\end{align*}
	From \eqref{1234}, we have
	\begin{align*}
		0\leq (2T)^{-1}\| \bdelta_{i\cdot}\bX \|_{2}^2  
		&\leq 2^{-1}\lambda \|\bdelta_{i\cdot}\|_1 + \lambda\| \bphi_{i\cdot} \|_1 - \lambda\| \hat{\bphi}_{i\cdot} \|_1 \\
		&\leq 2^{-1}\lambda \|\bdelta_{i\cS_i}\|_1 + 2^{-1}\lambda \|\bdelta_{i\cS_i^c}\|_1 + \lambda\| \bdelta_{i\cS_i} \|_1 - \lambda\|\bdelta_{i\cS_i^c}\|_1 \\
		&= (3/2)\lambda \|\bdelta_{i\cS_i}\|_1 - 2^{-1}\lambda \|\bdelta_{i\cS_i^c}\|_1,
	\end{align*}
	giving the desired inequality, $\| \bdelta_{i\cS_i^c} \|_1\leq 3\| \bdelta_{i\cS_i} \|_1$. This completes the proof. 
\end{proof}

\begin{lem}\label{lem:E2}
	Suppose the same conditions as Proposition \ref{thm:errbound}. We have
	\begin{align*}
		\min_{\bdelta_{i\cdot}\in \cD_i} \frac{T^{-1}\|\bdelta_{i\cdot}\bX\|_{\F}^2}{\|\bdelta_{i\cdot}\|_{2}^2} \geq \gamma - 16s_i\left\|T^{-1}\bX\bX'-\bSigma_x\right\|_{\max},
	\end{align*}
	where $\cD_i=\{\bdelta\in\mathbb{R}^{KN}:\|\bdelta_{\cS_i^c}\|_1\leq 3\|\bdelta_{\cS_i}\|_1\}$, and $\gamma>0$ is some constant. 
\end{lem}
\begin{proof}[Proof of Lemma \ref{lem:E2}]
	Let $\bdelta_{i\cdot}'\in\bbR^{KN}$ denote the $i$th column vector of $\bDelta'\in\bbR^{KN\times N}$. We see that 
	\begin{align}
		\frac{T^{-1}\|\bdelta_{i\cdot}\bX\|_{2}^2}{\|\bdelta_{i\cdot}\|_{2}^2}
		= \frac{\bdelta_{i\cdot}\E [T^{-1}\bX\bX']\bdelta_{i\cdot}'}{\bdelta_{i\cdot}\bdelta_{i\cdot}'}
		-\frac{\bdelta_{i\cdot}\left(\E[T^{-1}\bX\bX'] - T^{-1}\bX\bX'\right)\bdelta_{i\cdot}'}{\bdelta_{i\cdot}\bdelta_{i\cdot}'}.\label{E2:01}
	\end{align}
	Then it is sufficient to show that uniformly in $\bdelta_{i\cdot}\in\cD_i$ the first term is bounded from below by some constant and the second term converges to zero with high probability.

	We consider the first term of \eqref{E2:01}. 
	In view of the Rayleigh quotient with Condition \ref{ass:mineig}, 
	we obtain
	\begin{align}
		\min_{\bdelta_{i\cdot}\in \cD_i}\frac{\bdelta_{i\cdot}\E [T^{-1}\bX\bX']\bdelta_{i\cdot}'}{\bdelta_{i\cdot}\bdelta_{i\cdot}'}
		\geq \gamma.\label{E2:02}
	\end{align}

	We next consider the second term of \eqref{E2:01}. By applying H\"older's inequality twice, the numerator is bounded as
	\begin{align}
		&\left|\bdelta_{i\cdot}'\left(\E[T^{-1}\bX\bX'] - T^{-1}\bX\bX'\right)\bdelta_{i\cdot}' \right|
		\leq \left\|\bdelta_{i\cdot}\left(\E[T^{-1}\bX\bX'] - T^{-1}\bX\bX'\right)\right\|_{\max}\|\bdelta_{i\cdot}\|_1 \notag\\
		&\quad\leq \left\|\E[T^{-1}\bX\bX'] - T^{-1}\bX\bX'\right\|_{\max} \|\bdelta_{i\cdot}\|_1^2.\label{E2:03}
	\end{align}
	We further compute the upper bound of $\|\bdelta_{i\cdot}\|_1^2$. Lemma \ref{lem:negahban} yields 
	\begin{align}
		\|\bdelta_{i\cdot}\|_1^2
		= \left(\|\bdelta_{\cS_i}\|_1+\|\bdelta_{\cS_i^c}\|_1\right)^2
		\leq 16\|\bdelta_{\cS_i}\|_1^2
		\leq 16s_i\|\bdelta_{\cS_i}\|_2^2 
		\leq 16s_i\|\bdelta_{i\cdot}\|_2^2. \label{E2:04}
	\end{align}
	Combining \eqref{E2:03} and \eqref{E2:04} and rearranging the terms yield
	\begin{align}
		\max_{\bdelta_{i\cdot}\in \cD}\frac{ \bdelta_{i\cdot}\left(\E[T^{-1}\bX\bX'] - T^{-1}\bX\bX'\right)\bdelta_{i\cdot}'}{\|\bdelta_{i\cdot}\|_2^2} 
		\leq 16s_i\left\|\E[T^{-1}\bX\bX'] - T^{-1}\bX\bX'\right\|_{\max}.\label{E2:05}
	\end{align}
	From \eqref{E2:01}, \eqref{E2:02}, and \eqref{E2:05}, we obtain 
	\begin{align*}
		\min_{\bdelta_{i\cdot}\in \cD}\frac{T^{-1}\|\bdelta_{i\cdot}\bX\|_{2}^2}{\|\bdelta_{i\cdot}\|_{2}^2}
		\geq \gamma - 16s_i \left\|\E[T^{-1}\bX\bX'] - T^{-1}\bX\bX'\right\|_{\max}.
	\end{align*}
	This completes the proof. 
\end{proof}

\begin{lem}\label{lem:omegahat1} If Conditions \ref{ass:subG}--\ref{ass:mineig2} are true, then the following inequality holds for all $i\in[N]$ such that $\bar{r}_i=o(1)$ with probability at least $1-O((N\vee T)^{-\nu})$:
	\begin{align*}
		(a)~&~\max_{j\in[KN]}\left| \hat\sigma_i^2\hat{\bomega}_{j}'\hat\bSigma_x\hat{\bomega}_{j} - \sigma_i^2\omega_{j}^2 \right| 
		\lesssim M_\omega^2 \lambda, \\
		(b)~&~\max_{j\in[KN]}\left| \hat\sigma_i^2\hat\omega_j^2 - \sigma_i^2\omega_{j}^2 \right| 
		\lesssim M_\omega^2 \lambda, \\
		(c)~&~\max_{j\in[KN]}\left| \frac{\sigma_i\omega_{j}}{\hat\sigma_i\sqrt{\hat{\bomega}_{j}'\hat\bSigma_x\hat{\bomega}_{j}}}-1\right| \lesssim \frac{M_\omega^2 \lambda}{\left| \gamma^2 -O(\sqrt{M_\omega^2 \lambda}) \right|},\\
		(d)~&~\max_{j\in[KN]}\left| \frac{\sigma_i\omega_{j}}{\hat{\sigma}_i\hat\omega_{j}}-1\right| \lesssim \frac{M_\omega^2 \lambda}{\left| \gamma^2 -O(\sqrt{M_\omega^2 \lambda}) \right|}.
	\end{align*}
	In consequence, if $\bar{v}_i=\bar{r}_i+M_\omega^2 \lambda$ is $o(1)$, all the upper bounds are $o(1)$. 
\end{lem}
\begin{proof}[Proof of Lemma \ref{lem:omegahat1}]
	For any $(i,j)\in\cH$, we consider two bounds:
	\begin{align*}
		(a)&~~	\left| \hat\sigma_i^2\hat{\bomega}_{j}'\hat\bSigma_x\hat{\bomega}_{j} 
		- \sigma_i^2\omega_{j}^2 \right| 
		\leq \left(\left|\hat\sigma_i^2-\sigma_i^2\right|+\sigma_i^2 \right)\left| \hat{\bomega}_{j}'\hat\bSigma_x\hat{\bomega}_{j} -
		\omega_{j}^2 \right|
		+ \left|\hat\sigma_i^2-\sigma_i^2\right|\omega_{j}^2, \\
		(b)&~~\left| \hat\sigma_i^2\hat\omega_{j}^2
		- \sigma_i^2\omega_{j}^2 \right| 
		\leq \left(\left|\hat\sigma_i^2-\sigma_i^2\right|+\sigma_i^2 \right)\left| \hat\omega_{j}^2 -\omega_{j}^2 \right|
		+ \left|\hat\sigma_i^2-\sigma_i^2\right|\omega_{j}^2.
	\end{align*}
	To complete the proof, we derive upper bounds of 
	\begin{align*}
		(i):=\max_j|\hat{\bomega}_{j}'\hat\bSigma_x\hat{\bomega}_{j} -\omega_{j}^2|,~~
		(i)':=\max_j\left| \hat\omega_{j}^2 -\omega_{j}^2 \right|,~~\text{and}~~ (ii):=|\hat\sigma_i^2-\sigma_i^2|.
	\end{align*}
	First, $(i)'$ is bounded as
	\begin{align*}
		(i)'
		\lesssim \max_{j}\|{\bomega}_{j}\|\lambda_1 \lesssim M_\omega^2\lambda
	\end{align*}	
	by Theorem 6 of \cite{CaiEtAl2011}. Next, bound $(i)$. Since $w_{jj}={\bomega}_{j}'\be_j$, we have
	\begin{align*}
		(i)
		&\leq \max_j|(\hat{\bomega}_{j}-{\bomega}_{j})'\hat\bSigma_x\hat{\bomega}_{j}| + |{\bomega}_{j}'(\be_j-\hat\bSigma_x\hat{\bomega}_{j})| \\
		&\leq \max_j|(\hat{\bomega}_{j}-{\bomega}_{j})'(\hat\bSigma_x\hat{\bomega}_{j}-\be_j)| +\max_j|(\hat{\bomega}_{j}-{\bomega}_{j})'\be_j| + \max_j|{\bomega}_{j}'(\be_j-\hat\bSigma_x\hat{\bomega}_{j})| \\
		&\leq \max_j\|\hat{\bomega}_{j}-{\bomega}_{j}\|_1\|\hat\bSigma_x\hat{\bomega}_{j}-\be_j\|_\infty +\max_j\|\hat{\bomega}_{j}-{\bomega}_{j}\|_\infty + \max_j\|{\bomega}_{j}\|_1\|\be_j-\hat\bSigma_x\hat{\bomega}_{j}\|_\infty \\
		&\leq \max_j\|\hat{\bomega}_{j}-{\bomega}_{j}\|_1\lambda_1 + 2M_\omega\lambda_1.
	\end{align*}
	By the proof of Proposition \ref{thm:asynormal}, it holds with probability at least $1-O((N\vee T)^{-\nu})$ that
	\begin{align*}
		\max_{j}\|\hat{\bomega}_{j}-{\bomega}_{j}\|_1
		\lesssim (\max_j\|{\bomega}_{j}\|_{\ell_1}\lambda_1)^{1-r}s_\omega
		\lesssim M_\omega^{2-2r}\lambda^{1-r}s_\omega,
	\end{align*}
	where we have used $\lambda_1=b\max_j\|{\bomega}_{j}\|_{1}\lambda/2$. Then Conditions \ref{ass:invest} and \ref{ass:mineig2} yield 
	\begin{align*}
		(i) \lesssim M_\omega^{3-2r}\lambda^{2-r}s_\omega + M_\omega^2\lambda.
	\end{align*}

	Finally bound $(ii)$. We have
	\begin{align*}
		(ii)
		\leq T^{-1}\left|\sum_{t=1}^T(\hat{u}_{it}^2-u_{it}^2)\right| + \left|T^{-1}\sum_{t=1}^T({u}_{it}^2-\E u_{it}^2)\right|. 
	\end{align*}
	Consider the first term. Because $\hat{\bu}_{i\cdot} - \bu_{i\cdot} = - (\hat{\bphi}_{i\cdot}^L -  \bphi_{i\cdot})\bX$, we see that 
	\begin{align*}
		&T^{-1}\left|\sum_{t=1}^T(\hat{u}_{it}^2-u_{it}^2)\right|
		\leq T^{-1}\left|\sum_{t=1}^T(\hat{u}_{it}-u_{it})^2\right| + 2T^{-1}\left|\sum_{t=1}^T(\hat{u}_{it}-u_{it})u_{it}\right| \\
		&\leq T^{-1}\left\| (\hat{\bphi}_{i\cdot}^L -  \bphi_{i\cdot})\bX\right\|_2^2
		+ 2T^{-1}|(\hat{\bphi}_{i\cdot}^L -  \bphi_{i\cdot})\bX\bu_{i\cdot}'| \\
		&\leq T^{-1}\left\| (\hat{\bphi}_{i\cdot}^L -  \bphi_{i\cdot})\bX\right\|_2^2
		+ 2T^{-1}\|\hat{\bphi}_{i\cdot}^L -  \bphi_{i\cdot}\|_1\|\bX\bU'\|_{\max} \\
		&\lesssim s_i\lambda^2 + s_i\lambda\lambda = 2s_i\lambda^2,
	\end{align*}
	where the last inequality holds by Proposition \ref{thm:errbound} $(b)$$(c)$ and Lemma \ref{lem:tail_yu}. To evaluate the second term, putting $x=\sqrt{2\nu c_{uu}^2T^{-1}\log N}$ in Lemma \ref{lem:subE} $(a)$ yields
	\begin{align*}
		\Pro\left(\max_{i}\left|\frac{1}{T}\sum_{t=1}^T\left(u_{it}^2-\E u_{it}^2\right)\right|\gtrsim \sqrt{\frac{\log N}{T}}=o(\lambda)\right) 
		=O(N^{-\nu}).
	\end{align*}
	Combining these inequalities gives $(ii)\lesssim s_i\lambda^2 + o(\lambda) = o(\lambda)$ since $s_i\lambda=o(1)$ by the assumed condition $\bar{r}_i=o(1)$. 
	
	Consequently we obtain
	\begin{align*}
		(a)~~\left\{(ii)+\sigma_i^2\right\}(i)+(ii)\omega_{j}^2
		&\lesssim \left(M_\omega^2 \lambda + M_\omega^{3-2r}\lambda^{2-r}s_\omega \right) + o(\lambda) \\
		&
		\lesssim M_\omega^2\lambda + \min_{i\in[N]}\bar{r}_i^2
	\end{align*}
	and
	\begin{align*}
		(b)~~\left\{(ii)+\sigma_i^2\right\}(i)'+(ii)\omega_{j}^2
		&\lesssim M_\omega^2 \lambda + o(\lambda) \asymp M_\omega^2 \lambda.
	\end{align*}

	Denote $m_{ij}=\sigma_i\sqrt{{\bomega}_{j}'\bSigma_x{\bomega}_{j}}=\sigma_i\omega_j$ and $\hat{m}_{ij}$ is either $\hat{\sigma}_i\sqrt{\hat{\bomega}_{j}'\hat{\bSigma}_x\hat{\bomega}_{j}}$ or $\hat{\sigma}_i\hat{\omega}_{ j}$. For any $r$ such that $|\hat{m}_{ij}^2-m_{ij}^2|\leq r$, we have $\hat{m}_{ij}^2 \geq m_{ij}^2 - r$, and hence, $\hat{m}_{ij} \geq m_{ij} - \sqrt{r}$. Therefore, 
	\begin{align*}
		r\geq \left| m_{ij}^2 - \hat{m}_{ij}^2\right| 
		&= \left| m_{ij}\hat{m}_{ij} + \hat{m}_{ij}^2\right|\left| \frac{m_{ij}}{\hat{m}_{ij}} - 1\right| \\
		&\geq \left| m_{ij} \left(m_{ij} - \sqrt{r} \right) + m_{ij}^2 - r\right|\left| \frac{m_{ij}}{\hat{m}_{ij}} - 1\right| \\
		&= \left| 2m_{ij}^2 -O\left(\sqrt{r}\right) \right|\left| \frac{m_{ij}}{\hat{m}_{ij}} - 1\right|,
	\end{align*}
	which gives the inequality, 
	\begin{align*}
		\left| \frac{m_{ij}}{\hat{m}_{ij}}-1\right| \lesssim \frac{r}{\left| m_{ij}^2 -O(\sqrt{r}) \right|}.
	\end{align*}
	Since $\min_{i,j}m_{ij}^2\geq \gamma^2$ by Condition \ref{ass:mineig2}, this completes the proof. 
\end{proof}

\begin{lem}\label{lem:strongapp} Suppose $\nu>2$ and $\max_{i\in[N]}\bar{v}_i=O(T^{-\kappa_1})$ for some constants $\kappa_1\in(0,1/2)$. If Conditions \ref{ass:subG}--\ref{ass:mineig2} are assumed, then for each $(i,j),(k,\ell)\in\cH$ and any $x>0$, the following inequalities hold:
	\begin{align*}
		(a)~&~\Pro(\cZ_{ij}\geq x + \delta_1) - O((N\vee T)^{-\nu+2}) \\
		&\qquad \qquad \leq \Pro(\tT_{ij}-\sqrt{T}\phi_{ij}/\hat{m}_{ij}\geq x) 
		\leq \Pro(\cZ_{ij}\geq x - \delta_1) + O((N\vee T)^{-\nu+2}) , \\
		(b)~&~\Pro(\tT_{ij}-\sqrt{T}\phi_{ij}/\hat{m}_{ij}\geq x, \tT_{k\ell}-\sqrt{T}\phi_{k\ell}/\hat{m}_{k\ell}\geq x) \\
		&\qquad \qquad \leq \Pro(\cZ_{ij}\geq x-\delta_1, \cZ_{k\ell}\geq x-\delta_1) + O((N\vee T)^{-\nu+2}),
	\end{align*}
	where $(\cZ_{ij},\cZ_{k\ell})$ is a bivariate standard normal random vector with the covariance (correlation) $\rho_{(i,j),(k,\ell)} = \sigma_{ik}\omega_{j\ell}/(\sigma_i\sigma_k\omega_{j}\omega_{\ell})$ and  $\delta_1=O(T^{-\kappa})$ for some $\kappa\in(0,\kappa_1]$ is some positive sequence. 
\end{lem}
\begin{proof}[Proof of Lemma \ref{lem:strongapp}]
	Prove $(b)$. Fix $(i,j)$ and $(k,\ell)$ arbitrary and let $\tilde{z}_{ij}=T^{-1/2}\sum_{t=1}^T\xi_t^{(i,j)}$ with $\xi_t^{(i,j)}=u_{it}\bx_t'{\bomega}_{j}/m_{ij}$ and $m_{ij}= \sigma_i\omega_{j}$. From Theorem \ref{thm:t} we have 
	\begin{align*}
		&\Pro(\tT_{ij}-\sqrt{T}\phi_{ij}/\hat{m}_{ij}\geq x, \tT_{k\ell}-\sqrt{T}\phi_{k\ell}/\hat{m}_{k\ell}\geq x) \\
		&\qquad \leq \Pro(\tilde{z}_{ij}\geq x-\bar{v}_i, \tilde{z}_{k\ell}\geq x-\bar{v}_k) + O((N\vee T)^{-\nu+2}).
	\end{align*}
	Note that  $(\tilde{z}_{ij},\tilde{z}_{k\ell})=T^{-1/2}\sum_{t=1}^T(\xi_t^{(i,j)},\xi_t^{(k,\ell)})$ is the sum of the square integrable martingale difference sequence with respect to the filtration, $\cF_t=\sigma\{\bu_{t-s},\bx_{t-s+1}:s=0,1,\dots\}$. \cite[Theorem 3.2]{Kifer2013} establishes the strong approximation that without changing its distribution the vector sequence $(\xi_t^{(i,j)},\xi_t^{(k,\ell)})$ can be redefined on a richer probability space where there exists a bivariate standard normal random vector $(\cZ_{ij},\cZ_{k\ell})$ with the covariance (correlation) $\rho_{ijk\ell}=\sigma_{ik}\omega_{j\ell}/(m_{ij}m_{k\ell})$ such that $(\tilde{z}_{ij},\tilde{z}_{k\ell})=(\cZ_{ij},\cZ_{k\ell})+O(T^{-\kappa_2})$ holds a.s.\ for some $\kappa_2>0$, provided that the sufficient conditions (3.4) and (3.5) in his paper are true. Therefore, the proof of $(b)$ completes if these two conditions are verified. 
	
	\textit{Check Kifer's condition (3.4).} For any $(i,j)$ and $(k,\ell)$, we obtain
	\begin{align*}
		&\left|T^{-1}\sum_{t=1}^T\left(\E [\xi_t^{(i,j)}\xi_t^{(k,\ell)} \mid \cF_{t-1}]
		-\frac{\sigma_{ik}{\bomega}_{j}'\bSigma_x\bomega_\ell}{m_{ij}m_{k\ell}}\right)\right|\\
		&= \frac{|\sigma_{ik}|}{m_{ij}m_{k\ell}}\left|{\bomega}_{j}'T^{-1}\sum_{t=1}^T(\bx_t\bx_t'-\bSigma_x)\bomega_\ell\right| \\
		&\leq \frac{\|{\bomega}_{j}\|_1\|\bomega_\ell\|_1}{\gamma^2}\left\|T^{-1}\bX\bX'-T^{-1}\E\bX\bX'\right\|_{\max}.
	\end{align*}
	Thus Lemma \ref{lem:tail_yy} with Conditions \ref{ass:invest} and \ref{ass:mineig2} entails that the upper bound of this inequality becomes $O(M_\omega^2 T^{-1/2}\log^{3/2}(N\vee T))$ with probability at least $1-O((N\vee T)^{-\nu})$. By the assumed condition, we have 
	\begin{align*}
		M_\omega^2 T^{-1/2}\log^{3/2}(N\vee T) \lesssim \max_{i\in[N]}\bar{v}_i = O(T^{-\kappa_1})
	\end{align*}
	for some constant $\kappa_1\in(0,1/2)$. This verifies condition (3.4). 
	
	\textit{Check Kifer's condition (3.5).} 
	By Condition \ref{ass:mineig2}, we observe that 
	\begin{align*}
		\E \left[\|(\xi_t^{(i,j)},\xi_t^{(k,\ell)})\|_2^4\right]
		\lesssim \max_{i,j}\E \left[ |\xi_t^{(i,j)}|^4\right] 
		\leq \max_{i,j}\frac{\E[u_{it}^4]}{\sigma_i^4\omega_j^4}\max_j\E \left[(\bx_t'{\bomega}_{j})^4\right], 
	\end{align*}
	where $\max_i\max_t\E[u_{it}^4]<\infty$ by Condition \ref{ass:subG} and $\min_{i,j}\sigma_i^4\omega_j^4\geq \gamma^4$. 
	Without loss of generality, suppose $K=1$. By the VMA($\infty$) representation in \eqref{cond:sum0}, we obtain
	\begin{align*}
		\E \left[(\bx_t'{\bomega}_{j})^4\right]
		\lesssim \sum_{\ell=0}^\infty \E \left[ (\bomega_j'\bB_\ell\bu_{t-\ell})^4 \right]
		+ \left( \sum_{\ell=0}^\infty \E \left[ (\bomega_j'\bB_\ell\bu_{t-\ell})^2 \right] \right)^2
		=:(I) + (II).
	\end{align*}
	Recall that $\bSigma_x=\E[\bx_t\bx_t']=\sum_{\ell=0}^\infty \bB_\ell\bSigma_u\bB_\ell'$. Thus we have
	\begin{align*}
		(II) 
		= \left( \sum_{\ell=0}^\infty \bomega_j'\bB_\ell\bSigma_u\bB_\ell'\bomega_j \right)^2
		= (\bomega_j'\bSigma_x\bomega_j)^2=\omega_j^4 \leq 1/\gamma^2.
	\end{align*}
	Khintchine's inequality for a weighted sum of i.i.d.\ subG random variables \citep[Exercise 2.6.5]{Vershynin2018} yields
	\begin{align*}
		(I) &= \sum_{\ell=0}^\infty \E \left[ (\bomega_j'\bB_\ell\bu_{t-\ell})^4 \right]
		\lesssim \sum_{\ell=0}^\infty \left( \bomega_j'\bB_\ell\bSigma_u\bB_\ell'\bomega_j \right)^2 \\
		&\leq  \left( \sum_{\ell=0}^\infty \bomega_j'\bB_\ell\bSigma_u\bB_\ell'\bomega_j \right)^2
		= \omega_j^4\leq 1/\gamma^2.
	\end{align*}
	Therefore for given constant $\kappa_1\in(0,1/2)$, we have
	\begin{align*}
		&\sum_{t=1}^\infty t^{\kappa_1-1} \E \left[\|(\xi_t^{(i,j)},\xi_t^{(k,\ell)})\|^21\left\{\|(\xi_t^{(i,j)},\xi_t^{(k,\ell)})\|^2\geq 1/t^{\kappa_1-1}\right\}\right] \\
		&\leq \sum_{t=1}^\infty t^{2\kappa_1-2} \E \left[\|(\xi_t^{(i,j)},\xi_t^{(k,\ell)})\|^4\right] 
		\lesssim \frac{\E[u_{it}^4]}{\gamma^8} \sum_{t=1}^\infty t^{2\kappa_1-2} 
		< \infty,
	\end{align*}
	which verifies condition (3.5). This completes the proof of $(b)$. 
	
	Prove $(a)$. Note that 
	\begin{align*}
		\Pro\left(\tT_{ij}\geq x\right)
		&\geq \Pro\left(\tilde{z}_{ij} \geq x + \bar{v}_i\right)
		- \Pro\left(\max_{i,j}|v_{ij}| > \max_{i\in[N]}\bar{v}_i\right) \\
		&= \Pro\left(\tilde{z}_{ij} \geq x + \max_{i\in[N]}\bar{v}_i\right) - O((N\vee T)^{-\nu+2}),
	\end{align*}
	where the last equality follows from Theorem \ref{thm:t}. The remainder of the proof of $(a)$ is the same as that of $(b)$. This completes all the proofs. 
\end{proof}

\begin{lem}\label{lem:probratio} For any bivariate standard normal random vector $(\cZ_{1},\cZ_{2})$ with correlation $\rho\in(-1,1)$, it holds that
	\begin{align*}
		\sup_{z\in[0,\infty)}\frac{\Pro(\cZ_{1}\geq z, \cZ_{2}\geq z)}{Q(z)^2}
		\leq \frac{1}{\sqrt{1-\rho^2}}.
	\end{align*}
\end{lem}
\begin{proof}[Proof of Lemma \ref{lem:probratio}]
	Let $(x,y)\mapsto \phi(x,y;\rho)$ denote the density function of the bivariate standard normal random vector with correlation $\rho \in(-1,1)$. A simple calculation yields for any $z\geq 0$,
	\begin{align*}
		&\Pro\left(\cZ_{ij}\geq z, \cZ_{k\ell}\geq z\right) 
		= \int_z^\infty\!\!\!\int_z^\infty \phi(x,y;\rho) \diff x\diff y \\
		&= \frac{1}{\sqrt{1-\rho^2}}\int_z^\infty\!\!\!\int_z^\infty \phi(x,y;0)\exp\left\{-\frac{2-\rho^2}{2(1-\rho^2)}\left(x^2+y^2\right)+\frac{2\rho}{2(1-\rho^2)}xy\right\} \diff x\diff y \\
		&\leq \frac{1}{\sqrt{1-\rho^2}}\int_z^\infty\!\!\!\int_z^\infty \phi(x,y;0) \diff x\diff y \max_{0\leq x,y<\infty}\exp\left\{-\frac{2-\rho^2}{2(1-\rho^2)}\left(x^2+y^2\right)+\frac{2\rho}{2(1-\rho^2)}xy\right\}. 
	\end{align*}
	Note that $\int_z^\infty\int_z^\infty\phi(x,y;0)\diff x\diff y=Q(z)^2$ and 
	\begin{align*}
		\max_{0\leq x,y<\infty}\left\{-\frac{2-\rho^2}{2(1-\rho^2)}\left(x^2+y^2\right)+\frac{2\rho}{2(1-\rho^2)}xy\right\}
		\leq \max_{0\leq x,y<\infty}\frac{-2+\rho+\rho^2}{2(1-\rho^2)}(x^2+y^2)=0
	\end{align*}
	uniformly in $\rho\in(-1,1)$. Combining the results gives the desired uniform bound. 
\end{proof}

\begin{lem}\label{lem:max_ym}
	If Conditions \ref{ass:subG} and \ref{ass:stab} are true, then we have
	\begin{align*}
		\Pro\left( \max_{i\in[N]}\max_{t\in[T]}\max_{k\in[K]} |y_{i,t-k}|^m \gtrsim \sqrt{\log^{m} NT} \right) = O((N\vee T)^{-\nu}).
	\end{align*}
\end{lem}
\begin{proof}[Proof of Lemma \ref{lem:max_ym}]
	For $\by_t=\sum_{\ell=0}^\infty \bB_{\ell}\bu_{t-\ell}$, let  $\bb_{\ell,i}$ denote the $i$th column vector of $\bB_{\ell}'$ and $b=\sum_{\ell=T}^{\infty} \|\bB_{\ell}\|_\infty$ as before. By the truncation argument as in the proof of Lemma \ref{lem:tail_yu}, we have for any $x>0$,
	\begin{align*}
		&\Pro\left( \max_{i\in[N]}\max_{t\in[T]}\max_{k\in[K]} |y_{i,t-k}| > x \right) \\
		&\leq \Pro \left( \max_{t\in[T]}\max_{k\in[K]} \sum_{\ell=0}^\infty \|\bB_{\ell}\|_\infty\|\bu_{t-k-\ell}\|_\infty  > x \right) \\
		&\leq \Pro \left( \max_{t\in[T]}\max_{k\in[K]} \sum_{\ell=0}^{T-1} \|\bB_{\ell}\|_\infty\|\bu_{t-k-\ell}\|_\infty  > x/2 \right) \\
		&\qquad\qquad\qquad + \Pro \left( \max_{t\in[T]}\max_{k\in[K]} \sum_{\ell=T}^\infty \|\bB_{\ell}\|_\infty\|\bu_{t-k-\ell}\|_\infty  > x/2 \right) \\
		&\leq \Pro \left( \max_{t\in[T]}\max_{\ell\in[T]}\|\bu_{t-2\ell}\|_\infty > \frac{x}{2b} \right) 
		+ \frac{2}{x} \sum_{\ell=T}^{\infty} \|\bB_{\ell}\|_\infty\E\left[\max_{t\in[T]} \|\bu_{t-\ell}\|_\infty \right] \\
		&\leq 3NT\exp\left(-\frac{x^2}{4c_ub^2} \right) 
		+ \frac{2c_1\delta^T}{x(1-\delta)} \log^{1/2}NT,
	\end{align*}
	where $c_1>0$, $\delta\in(0,1)$, and $c_u>0$ have been given in Lemmas \ref{lem:subE}(b), \ref{lem:VMA}, and Condition \ref{ass:subG}, respectively.
	Thus for any finite $m\in\bbN$, we have
	\begin{align*}
		&\Pro\left( \max_{i\in[N]}\max_{t\in[T]}\max_{k\in[K]} |y_{i,t-k}|^m > x \right) 
		= \Pro\left( \max_{i\in[N]}\max_{t\in[T]}\max_{k\in[K]} |y_{i,t-k}| > x^{1/m} \right) \\
		&\lesssim NT\exp\left(-\frac{x^{2/m}}{4c_ub^2} \right) 
		+ \frac{\delta^T}{x^{1/m}(1-\delta)} \log^{1/2}NT.
	\end{align*}
	Setting $x = \{4c_ub^2(\nu+2)\log (N\vee T)\}^{m/2}$ for any $\nu>0$ leads to the result. 
\end{proof}

\begin{lem}\label{lem:XU*}
	Define $\lambda^* \asymp \sqrt{T^{-1}\log^3(N\vee T)}$. 
	If Conditions \ref{ass:subG}--\ref{ass:mineig} and  \ref{ass:wildboot} are true and $\bar{s}\lambda=o(1)$ holds, then the event,
	\begin{align*}
		\Pro^* \left( \|T^{-1}\bU^*\bX'\|_{\max} \gtrsim \lambda^* \right) = O((N\vee T)^{-\nu}),
	\end{align*}
	occurs with probability at least $1-O((N\vee T)^{-\nu})$. 
\end{lem}
\begin{proof}[Proof of Lemma \ref{lem:XU*}]
	Recall the notation that $\bU^*=(\hat{\bu}_1\zeta_1,\dots,\hat{\bu}_T\zeta_T)$ and $\hat{\bu}_t=\by_t-\hat{\bPhi}_{\normalfont{\textsf{L}}}\bx_t
	=\bu_t+(\bPhi-\hat{\bPhi}_{\normalfont{\textsf{L}}})\bx_t$. We have
	\begin{align}
		&\Pro^*\left(\left\|T^{-1}\sum_{t=1}^T \bx_{t}\bu_t^{*\prime}\right\|_{\max} > x \right) \notag\\
		&\leq \Pro^*\left(\left\|T^{-1}\sum_{t=1}^T \bx_{t}\bu_t'\zeta_t\right\|_{\max} > x/2 \right) 
		+ \Pro^*\left( \left\|T^{-1}\sum_{t=1}^T \bx_{t}(\hat{\bu}_t-\bu_t)'\zeta_t\right\|_{\max} > x/2 \right).\label{boot:ineq1}
	\end{align}
	Bound the first term of \eqref{boot:ineq1}. Since $\zeta_t$ is a sequence of i.i.d.\ subG random variables by Condition \ref{ass:wildboot}, applying the union bound and Hoeffding's inequality under $\Pro^*$ yields for any $x>0$,
	\begin{align*}
		&\Pro^*\left( \left\|T^{-1}\sum_{t=1}^T \bx_{t}\bu_t'\zeta_t \right\|_{\max} > x/2 \right) \\
		&\leq 2 KN^2\max_{i,j\in[N]}\max_{k\in[K]}\exp \left( -\frac{x^2T}{4c_\zeta T^{-1}\sum_{t=1}^Ty_{i,t-k}^2u_{jt}^2} \right) \\
		&\leq 2 KN^2\exp \left( -\frac{x^2T}{4c_\zeta \max_{i,j\in[N]}\max_{t\in[T]}\max_{k\in[K]}y_{i,t-k}^2u_{jt}^2} \right) \\
		&\lesssim 2KN^2\exp \left( -\frac{x^2T}{4c_\zeta \log^2 NT} \right),
	\end{align*}
	where the last inequality holds with probability at least $1-O((N\vee T)^{-\nu})$ by Lemma \ref{lem:max_ym}.

	By Proposition \ref{thm:errbound}(b) and Lemma \ref{lem:max_ym}, the second term of \eqref{boot:ineq1} is bounded as
	\begin{align*}
		&\Pro^*\left( \left\|T^{-1}\sum_{t=1}^T \bx_{t}(\hat{\bu}_t-\bu_t)'\zeta_t\right\|_{\max} > x/2 \right)\\
		&=\Pro^*\left(\left\|T^{-1}\sum_{t=1}^T \bx_{t}\bx_t'\zeta_t(\bPhi-\hat{\bPhi}_{\normalfont{\textsf{L}}})'\right\|_{\max} > x/2\right)\\
		&\leq \Pro^*\left( \left\|T^{-1}\sum_{t=1}^T \bx_{t}\bx_t'\zeta_t\right\|_{\max}\max_{i\in[N]}\left\| \hat{\bphi}_{i\cdot}^{\normalfont{\textsf{L}}}-\bphi_{i\cdot}\right\|_{1} > x/2 \right)\\
		&\leq \Pro^*\left( \left\|T^{-1}\sum_{t=1}^T \bx_{t}\bx_t'\zeta_t\right\|_{\max}\bar{s}\lambda \gtrsim x/2 \right)\\
		&\leq 2K^2N^2\max_{k,\ell\in[K]}\max_{i,j\in[N]} \exp \left( -\frac{x^2T}{4c_\zeta \bar{s}^2\lambda^2 T^{-1}\sum_{t=1}^Ty_{i,t-k}^2y_{j,t-\ell}^2} \right) \\
		&\lesssim 2K^2N^2 \exp \left( -\frac{x^2T}{4c_\zeta \bar{s}^2\lambda^2\log^2 NT} \right),
	\end{align*}
	with probability at least $1-O((N\vee T)^{-\nu})$. 
	Combining the obtained results with taking 
	\begin{align*}
		x= 2\sqrt{c_\zeta(\nu+2) T^{-1}\log^{3}NT}~(\,\asymp \lambda^*)
	\end{align*}
	gives the desired statement since $\bar{s}\lambda=o(1)$. This completes the proof. 
\end{proof}

\begin{lem}\label{lem:lassoest*}
	Recall $\lambda^* \asymp \sqrt{T^{-1}\log^3(N\vee T)}$, defined in Lemma \ref{lem:XU*}. 
	If Conditions \ref{ass:subG}--\ref{ass:mineig} and  \ref{ass:wildboot} are true, then the event,
	\begin{align*}
		\Pro^*\left( \max_{i\in[N]}\left\|\hat{\bphi}_{i\cdot}^{\normalfont{\textsf{L}}*}-\hat{\bphi}_{i\cdot}^{\normalfont{\textsf{L}}}\right\|_1\gtrsim \bar{s}\lambda^*\right) = O((N\vee T)^{-\nu+1})
	\end{align*}
	occurs with high probability. 
\end{lem}
\begin{proof}[Proof of Lemma \ref{lem:lassoest*}]
	Define $\bdelta_{i\cdot}^* = \hat{\bphi}_{i\cdot}^{\normalfont{\textsf{L}}*}-\hat{\bphi}_{i\cdot}^{\normalfont{\textsf{L}}}$. 
	Derive the non-asymptotic error bound for the Lasso estimator. 
	For $\lambda^*$ such that $b\lambda/2 \leq \lambda^*\asymp \sqrt{T^{-1}\log^3 (N\vee T)}$, define two events,
	\begin{align*}
		\mathcal{E}_1^*=\left\{\left\| T^{-1}\bX{\bU^*}' \right\|_{\max} \leq \lambda^*/2\right\},~~~ 
		\mathcal{E}_2^* = \left\{ \left\|T^{-1}\bX\bX'-\E[T^{-1}\bX\bX']\right\|_{\max} \leq \lambda^* \right\}.
	\end{align*}
	Lemmas \ref{lem:XU*} and \ref{lem:tail_yy} guarantee that the event, 
	\begin{align*}
		\left\{\Pro^*\left( \mathcal{E}_1^{*c} \right)=O((N\vee T)^{-\nu})\right\}\cap \cE_2^*, 
	\end{align*}
	occurs with probability at least $1-O((N\vee T)^{-\nu})$. Then, following the proof of Proposition \ref{thm:errbound}, we obtain the desired error bound, $\|\bdelta_{i\cdot}^*\|_1\lesssim \hat{s}_i\lambda^*$, on event $\cE_1^*\cap \cE_2^*$. Finally an application of Theorem 1(b) of \cite{ZhuLiu2020} gives $\hat{s}_i\lesssim s_i\leq \bar{s}$ with high probability. 
	The assertion is obtained by the union bound. 
	This completes the proof. 
\end{proof}

\begin{lem}\label{lem:sigma*}
	Recall $\lambda^* \asymp \sqrt{T^{-1}\log^3(N\vee T)}$, defined in Lemma \ref{lem:XU*}. 
	If Conditions \ref{ass:subG}--\ref{ass:mineig} and  \ref{ass:wildboot} are true and $\bar{s}\lambda=o(1)$ holds, then for $\bar{\tau}_1=\lambda^* + \bar{s}^2\lambda^{*2} $ the event,
	\begin{align*}
		\Pro^*\left(\max_{i\in[N]}\left|\hat{\sigma}_i^{*2}-\hat{\sigma}_i^2\right| \gtrsim \bar{\tau}_1 \right) = O((N\vee T)^{-\nu+1}),
	\end{align*}
	occurs with probability at least $1-O((N\vee T)^{-\nu+1})$. 
\end{lem}
\begin{proof}[Proof of Lemma \ref{lem:sigma*}]
	Recall that $u_{it}^*=\hat{u}_{it}\zeta_t$, $\sigma_i^{*2}=T^{-1}\sum_{t=1}^T{u}_{it}^{*2}$,  and $\hat{\sigma}_i^2=T^{-1}\sum_{t=1}^T\hat{u}_{it}^2$. By the construction, we have $\sigma_i^{*2}-\hat{\sigma}_i=T^{-1}\sum_{t=1}^T\hat{u}_{it}^2(\zeta_t^2-1)$, where $\{\hat{u}_{it}^2(\zeta_t^2-1)\}_t$ (for each $i\in[N]$) is a sequence of independent centered  sub-exponential random variables under $\Pro^*$. Thus for some constant $c_\zeta>0$, Bernstein's inequality with the union bound gives for any $x\in(0,1]$,
	\begin{align*}
		&\Pro^*\left( \max_{i\in[N]}\left|\sigma_i^{*2}-\hat{\sigma}_i^2 \right| > x \right)=\Pro^* \left( \max_{i\in[N]}\left| T^{-1}\sum_{t=1}^T\hat{u}_{it}^2(\zeta_t^2-1) \right| >x \right) \\
		&\leq 2N\max_i \exp \left(-\frac{x^2T}{c_\zeta T^{-1}\sum_{t=1}^T\hat{u}_{it}^4} \right)
		\leq 2N\exp \left(-\frac{x^2T}{c_\zeta\max_i\max_t\hat{u}_{it}^4} \right).
	\end{align*}
	Denote by ${\bdelta}_i\in\bbR^{KN}$ the $i$th column vector of $\bDelta'=(\hat{\bPhi}^L-\bPhi^0)'$. 
	Then we have 
	\begin{align*}
		\hat{u}_{it}^4=(u_{it}-\bx_t'\bdelta_i)^4
		\lesssim u_{it}^4+(\bx_t'\bdelta_i)^4
		\lesssim u_{it}^4+ \|\bx_t\|_{\infty}^4\max_{i\in[N]}\|\bdelta_i\|_1^4. 
	\end{align*}
	Therefore, by Proposition \ref{thm:errbound}(b) and Lemma \ref{lem:max_ym}, it holds that $\max_i\max_t{u}_{it}^4\lesssim \log^2(NT)$ and $\max_t\|\bx_t\|_{\infty}^4\max_{i\in[N]}\|\bdelta_i\|_1^4\lesssim \bar{s}\lambda\log^2(NT) $ occur with probability at least $1-O((N\vee T)^{-\nu+1})$. Therefore, taking
	\begin{align*}
		x\asymp \sqrt{(\nu+1)T^{-1}\log^3(NT)}\asymp \lambda^*
	\end{align*}
	leads to the upper bound, 
	\begin{align*}
		&\Pro^*\left( \max_i\left|\sigma_i^{*2}-\hat{\sigma}_i^2 \right| \gtrsim \lambda^* \right)\lesssim (N\vee T)^{-\nu},
	\end{align*}
	which holds with probability at least $1-O((N\vee T)^{-\nu+1})$. 
	
	Recall $\bdelta_{i\cdot}^* = \hat{\bphi}_{i\cdot}^{\normalfont{\textsf{L}}*}-\hat{\bphi}_{i\cdot}^{\normalfont{\textsf{L}}}$. 
	Then we have
	\begin{align*}
		&\Pro^*\left(\max_i \left| \hat{\sigma}_i^{*2}-\hat{\sigma}_i^{2} \right| > x \right) \\
		&= \Pro^*\left(\max_i \left| \sigma_{i}^{*2}-\hat{\sigma}_i^{2} -2 T^{-1}\sum_{t=1}^T u_{it}^*\bx_t'\bdelta_i^{*} + {\bdelta_i^{*}}'\hat{\bSigma}_x\bdelta_i^{*} \right| > x \right) \notag\\
		&\leq \Pro^*\left(\max_i \left| \sigma_i^{*2}-\hat{\sigma}_i^2\right| +2\left\|T^{-1}{\bU^*}'\bX\right\|_{\max}\max_i\left\|\bdelta_i^{*}\right\|_1 + \max_i\left\|\bdelta_i^{*}\right\|_1^2\left\|\hat{\bSigma}_x\right\|_{\max} >x \right),
	\end{align*}
	where $\|\hat{\bSigma}_x\|_{\max}=O(1)$ with probability at least $1-O((N\vee T)^{-\nu})$ by Lemma \ref{lem:tail_yy} and Condition \ref{ass:mineig2}. 
	By Lemmas \ref{lem:XU*} and \ref{lem:lassoest*}, setting 
	\begin{align*}
		x \asymp \lambda^* + \bar{s}\lambda^{*2}+ \bar{s}^2\lambda^{*2}
		\asymp \lambda^* + \bar{s}^2\lambda^{*2} = \bar{\tau}_1
	\end{align*}
	with the union bound yields the desired result. This completes the proof. 
\end{proof}


\begin{lem}\label{lem:u*2-hhat2}
	If Conditions \ref{ass:subG}--\ref{ass:mineig} and  \ref{ass:wildboot} are true, then for $\bar{\tau}_2=M_\omega^2\lambda$ the event, 
	\begin{align*}
		\Pro^*\left(\left|T^{-1}\sum_{t=1}^T ({u}_{it}^{*2}-\hat{u}_{it}^2)\hat{\bomega}_{j}'\bx_t\bx_t'\hat{\bomega}_{j}\right|\gtrsim \bar{\tau}_2 \right) = O((NT)^{-\nu}),
	\end{align*}
	occurs with probability at least $1-O((N\vee T)^{-\nu})$.
\end{lem}
\begin{proof}[Proof of Lemma \ref{lem:u*2-hhat2}]
	Because $\E^*{u}_{it}^{*2}=\hat{u}_{it}^2$, $({u}_{it}^{*2}-\hat{u}_{it}^2)\hat{\bomega}_{j}'\bx_t\bx_t'\hat{\bomega}_{j}$ is a sequence of i.i.d.\ sub-exponential random variables under $\Pro^*$, Bernstein's inequality is applied. For any $x\in(0,1]$, we obtain
	\begin{align*}
		&\Pro^*\left( \left|T^{-1}\sum_{t=1}^T ({u}_{it}^{*2}-\hat{u}_{it}^2)\hat{\bomega}_{j}'\bx_t\bx_t'\hat{\bomega}_{j}\right| > x \right)  
		\lesssim \exp\left(-\frac{x^2T}{T^{-1}\sum_{t=1}^T(\hat{\bomega}_{j}'\bx_t\bx_t'\hat{\bomega}_{j})^2} \right) \\
		&\lesssim \exp\left(-\frac{x^2T}{\max_{j}\max_t\|\hat{\bomega}_{j}\|_1^4\|\bx_t\|_\infty^4} \right) \\
		&\lesssim \exp\left(-\frac{x^2T}{M_\omega^4\log^2 (NT)} \right),
	\end{align*}
	where the last inequality holds with probability at least $1-O((N\vee T)^{-\nu})$. 
	Thus taking 
	\begin{align*}
		x=\sqrt{\nu M_\omega^4T^{-1}\log^3 (NT)}
		\asymp M_\omega^2\lambda =\bar{\tau}_2
	\end{align*}
	leads to the upper bound to be $O((NT)^{-\nu})$. This completes the proof. 
\end{proof}

\begin{lem}\label{lem:uhat2-sig2}
	If Conditions \ref{ass:subG}--\ref{ass:mineig} are true and $\bar{s}\lambda+M_\omega^2\lambda=o(1)$ holds, then for 
	\begin{align*}
		\bar{\tau}_3 := \max\left\{M_\omega^{3-2r}\lambda^{1-r}s_\omega \log^2(N\vee T), ~
		\bar{s}\lambda \log(N\vee T),~
		M_\omega^2\lambda\log(N\vee T)\right\},
	\end{align*}
	we have 
	\begin{align*}
		\Pro\left( \left|T^{-1}\sum_{t=1}^T \left(\hat{u}_{it}^{2}-\hat{\sigma}_i^2 \right)\hat{\bomega}_{j}'\bx_t\bx_t'\hat{\bomega}_{j}\right| \gtrsim \bar{\tau}_3 \right) = O((N\vee T)^{-\nu}).
	\end{align*}
\end{lem}
\begin{proof}[Proof of Lemma \ref{lem:uhat2-sig2}]
	We have
	\begin{align*}
		&\Pro\left( \left|T^{-1}\sum_{t=1}^T \left(\hat{u}_{it}^{2}-\hat{\sigma}_i^2 \right)\hat{\bomega}_{j}'\bx_t\bx_t'\hat{\bomega}_{j}\right| > x \right) \\
		&\leq \Pro\left( \left|T^{-1}\sum_{t=1}^T \left(u_{it}^2-\sigma_i^2 \right){\bomega}_{j}'\bx_t\bx_t'{\bomega}_{j}\right| > x/4 \right) \\
		&\quad + \Pro\left(\max_i\max_t\left|u_{it}^2-\sigma_i^2\right|\left|\bx_t'(\hat{\bomega}_{j}\hat{\bomega}_{j}'- {\bomega}_{j}{\bomega}_{j}')\bx_t\right| > x/4 \right) \\ &\quad + \Pro\left(\max_{i}\max_{t}\left|\hat{u}_{it}^{2}-u_{it}^2\right|\hat{\bomega}_{j}'\hat{\bSigma}_x\hat{\bomega}_{j} > x/4\right) +\Pro\left(\max_{i}\left|\sigma_i^2-\hat{\sigma}_i^2\right| \hat{\bomega}_{j}'\hat{\bSigma}_x\hat{\bomega}_{j} > x/4 \right)
	\end{align*}
	
	We first evaluate the second to fourth probabilities. By (the proof of) Lemma \ref{lem:omegahat1} with Condition \ref{ass:mineig2}, the inequalities
	\begin{align*}
		&\max_i\max_j\max_t\left|u_{it}^2-\sigma_i^2\right| \bx_t'(\hat{\bomega}_{j}\hat{\bomega}_{j}'- {\bomega}_{j}{\bomega}_{j}')\bx_t \\
		&\quad\leq \max_i\max_j\max_t|u_{it}^2+\sigma_i^2|\|\bx_t\|_\infty^2\|\hat{\bomega}_{j}(\hat{\bomega}_{j}- {\bomega}_{j})' + (\hat{\bomega}_{j}- {\bomega}_{j}){\bomega}_{j}'\|_1 \\
		&\quad\lesssim \|\hat{\bomega}_{j}- {\bomega}_{j}\|_1\|{\bomega}_{j}\|_1\log^2(N\vee T) \lesssim M_\omega^{3-2r}\lambda^{1-r}s_\omega \log^2(N\vee T), \\
		&\max_i\max_j\max_t\left|\hat{u}_{it}^{2}-u_{it}^2\right|\hat{\bomega}_{j}'\hat{\bSigma}_x\hat{\bomega}_{j} \\
		&\quad\lesssim 2\max_i\max_t|u_{it}|\|\bx_t\|_\infty\|\bdelta_i\|_1 + \max_t\|\bx_t\bx_t'\|_\infty\max_i\|\bdelta_i\|_1^2 \\
		&\quad\lesssim \bar{s}\lambda \log(N\vee T), \\
		&\max_i\max_j\left|\hat{\sigma}_{i}^{2}-\sigma_{i}^2\right|\hat{\bomega}_{j}'\hat{\bSigma}_x\hat{\bomega}_{j}
		\lesssim \bar{s}\lambda^2 
	\end{align*}
	simultaneously hold with probability at least $1-O((N\vee T)^{-\nu})$. 
	
	Next we evaluate the first probability. For $\bar{c}_3=M_\omega^2\log^2(N\vee T)$, the Azuma-Hoeffding inequality with the union bound and Lemma \ref{lem:max_ym} yield
	\begin{align*}
		&\Pro\left( \max_{i,j}\left|T^{-1}\sum_{t=1}^T ({u}_{it}^{2}-\sigma_i^2){\bomega}_{j}'\bx_t\bx_t'{\bomega}_{j}\right| > x \right) \\
		&\leq \Pro\left( \max_{i,j}\left|T^{-1}\sum_{t=1}^T ({u}_{it}^{2}-\sigma_i^2){\bomega}_{j}'\bx_t\bx_t'{\bomega}_{j}\right| > x \mid \max_{i,j,t}\left|{u}_{it}^{2}-\sigma_i^2\right|{\bomega}_{j}'\bx_t\bx_t'{\bomega}_{j}\lesssim \bar{c}_3 \right) \\
		&\qquad + \Pro\left(\max_{i,j,t}\left|{u}_{it}^{2}-\sigma_i^2\right|{\bomega}_{j}'\bx_t\bx_t'{\bomega}_{j} \gtrsim \bar{c}_3 \right) \\
		&\leq 2KN^2\exp\left(-\frac{x^2T}{2\bar{c}_3^2}\right)+O((N\vee T)^{-\nu}),
	\end{align*}
	where the upper bound becomes $O((N\vee T)^{-\nu})$ by setting 
	\begin{align*}
		x = 2(\nu+2)\bar{c}_3T^{-1/2}\log^{1/2}(N\vee T)\asymp M_\omega^2\lambda\log(N\vee T).
	\end{align*}

	If we set 
	\begin{align*}
		x&\asymp \max \left\{M_\omega^{3-2r}\lambda^{1-r}s_\omega \log^2(N\vee T), ~
		\bar{s}\lambda \log(N\vee T),~
		\bar{s}\lambda^2,~
		M_\omega^2\lambda\log(N\vee T)\right\} \\
		&=\max \left\{M_\omega^{3-2r}\lambda^{1-r}s_\omega \log^2(N\vee T), ~
		\bar{s}\lambda \log(N\vee T),~
		M_\omega^2\lambda\log(N\vee T)\right\} = \bar{\tau}_3,
	\end{align*}
	the result follows. This completes the proof. 
\end{proof}


\begin{lem}\label{lem:boot}
	If Conditions \ref{ass:subG}--\ref{ass:mineig2} and \ref{ass:wildboot} are true and $\bar{s}\lambda+M_\omega^2\lambda=o(1)$ holds, then for 
	\begin{align*}
		\bar{\mu}_1 := \max\left\{\bar{\tau}_1,\bar{\tau}_2,\bar{\tau}_3\right\},
	\end{align*}
	the event, 
	\begin{align*}
		\Pro^*\left( \max_{i\in[N]}\max_{j\in[KN]}\left|\frac{\tilde{m}_{ij}^{*}}{\hat{m}_{ij}^*} -1 \right| \gtrsim \bar{\mu}_1 \right) 
		= O((N\vee T)^{-\nu}),
	\end{align*}
	occurs with probability at least $1-O((N\vee T)^{-\nu})$. 
\end{lem}
\begin{proof}[Proof of Lemma \ref{lem:boot}]
	For any $x>0$, a simple calculus yields
	\begin{align*}
		&\Pro^*\left( \max_{i\in[N]}\max_{j\in[KN]}\left|\frac{\tilde{m}_{ij}^{*}}{\hat{m}_{ij}^*} -1 \right| > x \right) \\
		&\leq \Pro^*\left( \max_{i}\max_{j} \left|\tilde{m}_{ij}^{*2} - \hat{m}_{ij}^{*2} \right| > x\hat{m}_{ij}^*\left(\tilde{m}_{ij}^{*}+ \hat{m}_{ij}^* \right) \right)\\
		&\leq \Pro^*\left( \max_{i}\max_{j} \left|\tilde{m}_{ij}^{*2} - \hat{m}_{ij}^{*2} \right| > x \min_i\min_j\hat{m}_{ij}^{*2} \right) \\
		&\leq \Pro^*\left( \max_{i}\max_{j} \left|\tilde{m}_{ij}^{*2} - \hat{m}_{ij}^{*2} \right| \gtrsim \gamma^3 x \right) 
		+ \Pro^*\left( \min_i\min_j\hat{m}_{ij}^{*2} \lesssim \gamma^3 \right)
	\end{align*}
	We see that $\tilde{m}_{ij}^{*2}\geq 0$ a.s.\ and, by Lemmas \ref{lem:omegahat1} and \ref{lem:sigma*} with Condition \ref{ass:mineig2}, we obtain
	\begin{align*}
		\hat{m}_{ij}^{*2}=\hat{\sigma}_{i}^{*2}\hat{\bomega}_{j}'\hat{\bSigma}_x\hat{\bomega}_{j}
		&\geq \hat\sigma_i^2\hat{\bomega}_{j}'\hat\bSigma_x\hat{\bomega}_{j} - |\hat\sigma_i^{*2}-\hat\sigma_i^2|\hat{\bomega}_{j}'\hat\bSigma_x\hat{\bomega}_{j} \\
		&\gtrsim \sigma_i^2\omega_j^2 - M_\omega^2\lambda -\bar{\tau}_1(\omega_j^2 - M_\omega^2\lambda) \geq \gamma^2 - O(M_\omega^2\lambda).
	\end{align*}
	uniformly in $(i,j)\in[N]\times[KN]$ with probability at least $1-O((N\vee T)^{-\nu+1})$. 
	Thus it suffices to evaluate the upper bound of
	\begin{align}
		&\Pro^*\left(\max_{i}\max_{j} \left|\tilde{m}_{ij}^{*2}-\hat{m}_{ij}^{*2}\right| > x \right) \notag\\
		&\leq \Pro^*\left(\max_{i}\max_{j} \left|\tilde{m}_{ij}^{*2}-\hat{m}_{ij}^{2}\right| > x/2 \right)
		+ \Pro^*\left(\max_{i}\max_{j} \left|\hat{m}_{ij}^{*2}-\hat{m}_{ij}^{2}\right| > x/2 \right).  \label{m2m2}
	\end{align}
	The second probability of \eqref{m2m2} is bounded by the same argument just above; we obtain
	\begin{align}
		\Pro^*\left(\max_{i}\max_{j}\left|\hat{m}_{ij}^{*2} - \hat{m}_{ij}^{2}\right| \gtrsim \bar{\tau}_1 \right) 
		&=O((N\vee T)^{-\nu+1}), \label{mhat*}
	\end{align}
	which occurs with probability at lease $1-O((N\vee T)^{-\nu+1})$. Next bound the first probability of \eqref{m2m2}. Set
	\begin{align*}
		x_2 := c\bar{\tau}_2 + {\tau}_{ij,3}~~\text{with}~~ 
		\bar{\tau}_2=M_\omega^2\lambda,~~~
		{\tau}_{ij,3}=\left|T^{-1}\sum_{t=1}^T \left( \hat{u}_{it}^2- \hat{\sigma}_i^2 \right) \hat{\bomega}_{j}'\bx_t\bx_t'\hat{\bomega}_{j}\right|
	\end{align*}
	for some constant $c>0$. Then Lemma \ref{lem:u*2-hhat2} establishes that the event,
	\begin{align}
		&\Pro^*\left( \left|\tilde{m}_{ij}^{*2} - \hat{m}_{ij}^2\right| >x_2 \right) \notag\\
		&\leq \Pro^* \left( \left|
		T^{-1}\sum_{t=1}^T \left( {u}_{it}^{*2} - \hat{u}_{it}^2\right) \hat{\bomega}_{j}'\bx_t\bx_t'\hat{\bomega}_{j}\right| + \left|T^{-1}\sum_{t=1}^T \left( \hat{u}_{it}^2- \hat{\sigma}_i^2 \right) \hat{\bomega}_{j}'\bx_t\bx_t'\hat{\bomega}_{j}\right| > x_2 \right) \notag \\
		&\leq \Pro^* \left( \left|
		T^{-1}\sum_{t=1}^T \left( {u}_{it}^{*2} - \hat{u}_{it}^2\right) \hat{\bomega}_{j}'\bx_t\bx_t'\hat{\bomega}_{j}\right| \gtrsim \bar{\tau}_2 \right)
		=O((N\vee T)^{-\nu}), \label{mtilde*}
	\end{align}
	holds with probability at least $1-O((N\vee T)^{-\nu})$. By Lemma \ref{lem:uhat2-sig2}, ${\tau}_{ij,3}$ is bounded by 
	\begin{align*}
		{\tau}_{ij,3}\lesssim \bar{\tau}_3 = \max\left\{M_\omega^{3-2r}\lambda^{1-r}s_\omega \log^2(N\vee T), ~
		\bar{s}\lambda \log(N\vee T),~
		M_\omega^2\lambda\log(N\vee T)\right\}
	\end{align*}
	with probability at least $1-O((N\vee T)^{-\nu})$.

	Combine the two probabilities. Since $\lambda\asymp \lambda^*$,  taking
	\begin{align*}
		x &= \max\left\{\bar{\tau}_1,x_2 \right\} 
		\asymp \max\left\{ \bar{\tau}_1,\bar{\tau}_2,\bar{\tau}_3 \right\}=\bar{\mu}_1 
	\end{align*}
	in \eqref{m2m2} establishes the desired inequality up to some positive constant factor. This completes the proof.
\end{proof}

\begin{lem}\label{lem:boot2}
	If Conditions \ref{ass:subG}--\ref{ass:mineig2} and \ref{ass:wildboot} are true, then for 
	\begin{align*}
		\bar{\mu}_2 := \max\left\{ \lambda^{2-r} M_\omega^{2-2r}s_\omega,\ \bar{s}\lambda M_\omega \log^{3/2}(N\vee T) \right\},
	\end{align*}
	the event,
	\begin{align*}
		\Pro^*\left( \max_{i\in[N]}\max_{j\in[KN]}\left|\frac{R_{ij}^*}{\hat{m}_{ij}^*} \right| \gtrsim \bar{\mu}_2 \right) 
		= O((N\vee T)^{-\nu}),
	\end{align*}
	occurs with probability at least $1-O((N\vee T)^{-\nu})$. 
\end{lem}
\begin{proof}[Proof of Lemma \ref{lem:boot2}]
	Since we have established that $\hat{m}_{ij}^*$ is uniformly lower bounded by a positive constant by Lemmas \ref{lem:omegahat1} and \ref{lem:sigma*} with Conditions \ref{ass:mineig2} and $M_\omega^2\lambda=o(1)$ in the proof of Lemma \ref{lem:boot}. Thus it is sufficient to derive the upper bound of $\max_i\max_j|R_{ij}^*|$. Observe that 
	\begin{align*}
		|R_{ij}^*| \leq \|T^{-1/2}\bu_i^*\bX'\|_\infty\|\hat{\bomega}_{j}-{\bomega}_{j}\|_1 + \|\sqrt{T}(\hat{\bphi}_i^L-\hat{\bphi}_i^{L*})\|_1\|\hat{\bSigma}_x\hat{\bomega}_{j}-\be_j\|_\infty.
	\end{align*}
	By the proof of Proposition \ref{thm:asynormal}, the event, 
	\begin{align*}
		\left\{\max_{j\in[N]}\|\hat{\bomega}_{j}-{\bomega}_{j}\|_1 \lesssim (M_\omega^2\lambda)^{1-r}s_\omega \right\}\cap \left\{\max_{j\in[N]}\|\hat\bSigma_x\hat{\bomega}_{j}-\be_j\|_\infty\lesssim M_\omega\lambda \right\},
	\end{align*}
	occurs with probability at lease $1-O((N\vee T)^{-\nu})$. Conditional on this event with setting 
	\begin{align*}
		x &= (\lambda^*\max_j\|\hat{\bomega}_{j}-{\bomega}_{j}\|_1)\vee (\sqrt{T}\bar{s}\lambda^*\max_j\|\hat\bSigma_x\hat{\bomega}_{j}-\be_j\|_\infty),
	\end{align*}
	we have by Lemmas \ref{lem:XU*} and \ref{lem:lassoest*},
	\begin{align*}
		&\Pro^*\left( \max_{i,j}|R_{ij}^*|>x\right) \\
		&\qquad \leq \Pro^*\left( \max_i\|T^{-1/2}\bu_i^*\bX'\|_\infty\max_j\|\hat{\bomega}_{j}-{\bomega}_{j}\|_1>x/2 \right) \\
		&\qquad\qquad + \Pro^*\left(\max_i\|\sqrt{T}(\hat{\bphi}_i^L-\hat{\bphi}_i^{L*})\|_1\max_j\|\hat\bSigma_x\hat{\bomega}_{j}-\be_j\|_\infty > x/2 \right) \\
		&\qquad \leq \Pro^*\left( \max_i\|T^{-1/2}\bu_i^*\bX'\|_\infty \gtrsim \lambda^* \right) + \Pro^*\left(\max_i\|\hat{\bphi}_i^L-\hat{\bphi}_i^{L*}\|_1 \gtrsim \bar{s}\lambda^* \right) \\
		&\qquad = O((N\vee T)^{-\nu+1})
	\end{align*}
	with high probability. Finally we see that by the proof of Proposition \ref{thm:asynormal}
	\begin{align*}
		x &= (\lambda^*\max_j\|\hat{\bomega}_{j}-{\bomega}_{j}\|_1)\vee (\sqrt{T}\bar{s}\lambda^*\max_j\|\hat\bSigma_x\hat{\bomega}_{j}-\be_j\|_\infty)\\
		&\lesssim \left( \lambda^{2-r} M_\omega^{2-2r}s_\omega \right) \vee \left( \bar{s}\lambda M_\omega \log^{3/2}(N\vee T) \right) = \bar{\mu}_2,
	\end{align*}
	with high probability. This completes the proof.
\end{proof}

\begin{lem}\label{Jiang}
	If Conditions \ref{ass:subG}--\ref{ass:mineig2} and \ref{ass:wildboot}, and $M_\omega \leq T^{1/2}/\log^3(N\vee T)$ are true, then the events, 
	\begin{align*}
		\Pro^*\left( \frac{S_{ij}^*}{\tilde{m}_{ij}^{*}}  > \tt \right)
		= Q(\tt)\left\{ 1+O\left( \frac{M_\omega\log NT}{T^{1/2}} \right)\left(1+\tt\right)^{3} \right\}
	\end{align*}
	and 
	\begin{align*}
		\Pro^*\left( \frac{S_{ij}^*}{\tilde{m}_{ij}^{*}}  < -\tt \right)
		= Q(\tt)\left\{ 1+O\left( \frac{M_\omega\log NT}{T^{1/2}} \right)\left(1+\tt\right)^{3} \right\}
	\end{align*}
	for $\tt\in[0,\bar{\tt}]$ uniformly in $(i,j)\in\cH$, occur with probability at least $1-O((N\vee T)^{-\nu})$.
\end{lem}
\begin{proof}[Proof of Lemma \ref{Jiang}]
	Apply \cite[Theorem 2.3]{JingShaoWang2003} to the self-normalized sum, $S_{ij}^*/\tilde{m}_{ij}^*$ under $\Pro^*$. 
	By Lemmas \ref{lem:uhat2-sig2} and \ref{lem:omegahat1} with Conditions \ref{ass:mineig2} and \ref{ass:wildboot}, we have
	\begin{align*}
		B_{ij}^2&:=T^{-1}\sum_{t=1}^T\E^*{u}_{it}^{*2}(\bx_t'\hat{\bomega}_{j})^{2}
		= T^{-1}\sum_{t=1}^T\hat{u}_{it}^{2}(\bx_t'\hat{\bomega}_{j})^{2} \\
		&\geq \hat{\sigma}_{i}^{2}\hat{\bomega}_{j}'\hat{\bSigma}_x\hat{\bomega}_{j} - O(\bar{\tau}_3) \geq \gamma^3(1-o(1))
	\end{align*}
	with probability at least $1-O((N\vee T)^{-\nu})$. Furthermore, since $\E^*|\zeta|^3$ is bounded by a finite constant under Condition \ref{ass:wildboot}, we have
	\begin{align*}
		L_{ij} &:= T^{-3/2}\sum_{t=1}^T\E^*|{u}_{it}^*|^{3}|\bx_t'\hat{\bomega}_{j}|^{3}
		\asymp T^{-3/2}\sum_{t=1}^T|\hat{u}_{it}|^{3}|\bx_t'\hat{\bomega}_{j}|^{3} \\
		&\leq T^{-3/2}\sum_{t=1}^T|\hat{u}_{it}|^{2}|\bx_t'\hat{\bomega}_{j}|^{2}\max_{t}|\hat{u}_{it}||\bx_t'\hat{\bomega}_{j}|
		=T^{-1/2}B_{ij}^2\max_{t}|\hat{u}_{it}||\bx_t'\hat{\bomega}_{j}|, 
	\end{align*}
	where Proposition \ref{thm:errbound} and (the proofs of) Lemmas \ref{lem:omegahat1} and \ref{lem:max_ym} give 
	\begin{align*}
		&\max_{t}|\hat{u}_{it}||\bx_t'\hat{\bomega}_{j}|
		\lesssim \max_{t}\left(|u_{it}|+\|\bx_t\|_\infty\|\bdelta_i\|_1\right)\|\bx_t\|_\infty\|\hat{\bomega}_{j}\|_1\\
		&\leq \left(\max_{i}\max_{t}|u_{it}|+\max_{t}\|\bx_t\|_\infty\|\bdelta_i\|_1\right)\max_{t}\|\bx_t\|_\infty\max_{j}\|\hat{\bomega}_{j}\|_1
		\lesssim M_\omega\log(NT) 
	\end{align*}
	with probability at least $1-O((N\vee T)^{-\nu})$.
	Thus 
	\begin{align*}
		d_{ij} := \frac{B_{ij}}{L_{ij}^{1/3}}
		\gtrsim \frac{T^{1/6}B_{ij}^{1/3}}{M_\omega^{1/3}\log^{1/3}(NT) }
		\geq \frac{T^{1/6}\gamma^{1/2}(1+o(1))}{M_\omega^{1/3}\log^{1/3}(NT) }
		\geq \log^{2/3}(N\vee T),
	\end{align*}
	where the lower bound holds uniformly in $(i,j)$ since $M_\omega \leq T^{1/2}/\log^3(N\vee T)$ is implied by the condition $\bar{\mu}=o(1)$. Therefore, by \cite[Theorem 2.3]{JingShaoWang2003}, we conclude that 
	\begin{align*}
		\Pro^*\left( \frac{S_{ij}^*}{\tilde{m}_{ij}^{*}}  > \tt \right)
		= Q(\tt)\left\{ 1+O\left( \frac{M_\omega\log NT}{T^{1/2}} \right)\left(1+\tt\right)^{3} \right\}
	\end{align*}
	holds with high probability for $\tt\in[0,\bar{\tt}]$, where $\bar{\tt}\asymp \log^{1/2}(N\vee T)=o(d_{ij})$, uniformly in $(i,j)\in\cH$. This completes the proof. 
\end{proof}

\begin{lem}\label{lem:screen}
Let $\nu>4$. If Conditions \ref{ass:subG}--\ref{ass:mineig2} and \ref{ass:wildboot} are true, then the event, $\cS\subset \what\cS_{\textsf{L}}$, occurs with probability at least $1-O((N\vee T)^{-\nu+1})$.
\end{lem}
\begin{proof}[Proof of Lemma \ref{lem:screen}]
	For each $i\in[N]$, the Karush-Kuhn-Tucker (KKT) condition of the minimization problem in \eqref{syslasso} is given by 
	\begin{align*}
		\bX\left(\by_{i\cdot}-\hat{\bphi}_{i\cdot}\bX\right)'/T = \bpsi_i\lambda,
	\end{align*}
	where $\bpsi_i=(\psi_{i1},\dots,\psi_{KN})'$ is defined by $\psi_{ij}=\sgn{\hat{\phi}_{ij}}$ for $\hat{\phi}_{ij}\not=0$ and $\psi_{ij}\in[-1,1]$ for $\hat{\phi}_{ij}=0$. By the definition of $\by_{i\cdot}$ and some algebra, the KKT condition is equivalently denoted as
	\begin{align*}
		\bOmega(T^{-1}\bX\bX'-\bSigma_x)(\bphi_{i\cdot}-\hat{\bphi}_{i\cdot})'+(\bphi_{i\cdot}-\hat{\bphi}_{i\cdot})'+T^{-1}\bOmega\bX\bu_{i\cdot}' = \bOmega\bpsi_i\lambda.
	\end{align*}
	We Take the $\ell_\infty$-norm. Then by the triangle and H\"older's inequality, we obtain
	\begin{align*}
		\|\bphi_{i\cdot}-\hat{\bphi}_{i\cdot}\|_\infty
		&\leq \|T^{-1}\bOmega\bX\bu_{i\cdot}'\|_\infty + \|\bOmega(T^{-1}\bX\bX'-\bSigma_x)(\bphi_{i\cdot}-\hat{\bphi}_{i\cdot})'\|_\infty + \|\bOmega\bpsi_i\|_\infty\lambda \\
		&\leq \|\bOmega\|_\infty \left(\|T^{-1}\bX\bu_{i\cdot}'\|_\infty + \|T^{-1}\bX\bX'-\bSigma_x\|_{\max}\|\bphi_{i\cdot}-\hat{\bphi}_{i\cdot}\|_1 + \lambda\right). 
	\end{align*}
	By Proposition \ref{thm:errbound}(b), Condition \ref{ass:invest}, and Lemmas \ref{lem:tail_yu} and \ref{lem:tail_yy}, we have 
	\begin{align*}
		\|\bphi_{i\cdot}-\hat{\bphi}_{i\cdot}\|_\infty
		\lesssim M_\omega \sqrt{T^{-1}\log^3(N\vee T)}=:u_n
	\end{align*}
	with probability at least $1- O((N\vee T)^{-\nu})$. 
	
	The event, $\{\cS_i\subset\what\cS_i^\textsf{L}\}$, is equivalent to the event that $\hat{\phi}_{ij}^\textsf{L}\not=0$ holds whenever $j\in\cS_i$. This together with the triangle inequality implies that
	\begin{align*}
		\Pro\left( \cS_i\subset\what\cS_i^\textsf{L} \right) 
		&= \Pro\left( \min_{j\in\cS_i}|\hat{\phi}_{ij}^\textsf{L}| >0 \right) \\
		&\geq \Pro\left( \min_{j\in\cS_i}\left\{|\phi_{ij}|-|\hat{\phi}_{ij}^\textsf{L}-\phi_{ij} |\right\} >0 \right) \\
		&\geq \Pro\left( \max_{j\in\cS_i}|\hat{\phi}_{ij}^\textsf{L}-\phi_{ij} | <\min_{j\in\cS_i}|\phi_{ij} | \right)\\
		&\geq \Pro\left( \|\what{\bphi}_{i\cdot}^\textsf{L}-\bphi_{i\cdot} \|_\infty < \min_{j\in\cS_i}|\phi_{ij}| \right).
	\end{align*}
	Therefore, if $u_n 
	\ll \min_{j\in\cS_i}|\phi_{ij}| $, then $\cS_i\subset\what\cS_i^\textsf{L}$ occurs with probability at least $1- O((N\vee T)^{-\nu})$. This result holds for all $i\in[N]$. Then we can conclude
	\begin{align*}
		\Pro\left(\cS\subset \what\cS_{\normalfont{\textsf{L}}}\right)
		\geq 1-\sum_{i\in[N]}\Pro\left(\cS_i\not\subset \what\cS_i^{\normalfont{\textsf{L}}}\right)
		\geq 1-NO((N\vee T)^{-\nu}). 
	\end{align*}
	the lower bound converges to one since $\nu>4$. This establishes the result. 
\end{proof}

\section{Precision Matrix Estimation}\label{ssec:precision}

The construction of the debiased lasso estimator requires a precision matrix estimator. In a low dimensional setting, it is obtained as the inverse of the sample covariance matrix of $\bX$, but it becomes less accurate or even infeasible as the dimensionality tends to large. To this end, we need some regularized estimator. In this paper we employ the CLIME method of \cite{CaiEtAl2011} for the estimation of $\bOmega$.

Define $\hat\bTheta=(\hat{\theta}_{ij}) \in \bbR^{KN\times KN}$ as the solution to the following minimization problem: 
\begin{align*}
	\min\|\bTheta\|_1~~~\text{subject to}~~~\|\hat{\bSigma}_{x,\ep}\bTheta-\bI_{KN}\|_{\max} \leq \lambda_1,
\end{align*}
where $\lambda_1$ is a positive regularization parameter and  $\hat{\bSigma}_{x,\ep}=\hat{\bSigma}_x+\ep\bI_{KN}$ for $\ep\geq 0$. Note that $\hat{\bSigma}_{x,0}=\hat{\bSigma}_{x}$. The CLIME estimator of $\bOmega$ is then obtained by symmetrization with
\begin{align*}
	\hat{\omega}_{ij} = \hat{\omega}_{ji}
	= \hat{\theta}_{ij}1\{|\hat{\theta}_{ij}|\leq |\hat{\theta}_{ji}|\}
	+ \hat{\theta}_{ji}1\{|\hat{\theta}_{ij}|> |\hat{\theta}_{ji}|\}.
\end{align*}
The non-asymptotic error bound is derived in Lemma \ref{lem:omegahat1}, which is of independent interest. For related theoretical results of stationary time series, see \cite{ShuNan2019}. In our numerical analysis, we use \texttt{fastclime} by \cite{PangEtAl2014}, which is available at \url{https://github.com/cran/fastclime}.

When the lasso is debiased (desparsified) in a linear regression context, the \textit{nodewise regression} \citep{van2014asymptotically} is frequently used as an alternative. However, it might not suitable for the VAR models since it repeats the regression of each variable in $\bX$ on the others, which includes regressions of a ``past'' variable on ``future'' ones. 

\section{Additional Results on Robustification}\label{ssec:eval}

In this section, we show that $|\tT_{ij}|^p$ and $\exp(c|\tT_{ij}|^\alpha)$ are UI for some fixed constants $p,c,\alpha>0$ for arbitrarily fixed $(i,j)\in\cS^c$. In pursuit of the goal, we prove 
\begin{align}\label{ui_tail}
\sup_{T,N,K}\Pro\left(|\tT_{ij}|>x\right) \lesssim \exp\left(-\delta x^\alpha + \log x \right)
\end{align}
for any (large) $x>0$ and some $\delta,\alpha>0$, which is sufficient for $\exp(c|\tT_{ij}|^\alpha)$ to be UI for any $c\in[0,\delta)$ by Lemma \ref{lem:ui_suffice_exp} below. Furthermore, this result implies the UI of $|\tT_{ij}|^p$ for any constant $p>0$ by Lemma \ref{lem:ui_exp_poly}.

Before starting the proof of \eqref{ui_tail}, we assume the following: 
\begin{enumerate}
\item[A.] Conditions 2--5 in Section \ref{sec:theory}.
\item[B.] The error term, $\{\bu_t\}$, is a sequence of i.i.d.\ Gaussian random vectors with mean zero and covariance matrix $\bSigma_u$. 
\item[C.] The parameter space of $m_{ij}$ is lower bounded by some positive constant. 
\item[D.] $T\geq \{(s_\omega^{4/(1-2r)}M_\omega^{8(1-r)/(1-2r)})\vee s_i^4\}\log^{10}(N\vee T)$ and $N$ is at most a polynomial order of $T$. 
\end{enumerate}

Condition B is used for establishing Lemma \ref{lem:ui_z/m}, where a concentration inequality for a Gaussian martingale is applied. Condition C prevents $\hat{m}_{ij}$ from being unboundedly close to zero in probability. These two conditions help streamline the proof.

\begin{proof}[Proof of \eqref{ui_tail}]
We have 
\begin{align}
\Pro\left(|\tT_{ij}|>x\right)
&=\Pro\left( \left|\frac{z_{ij}+r_{ij}}{m_{ij}}\right|\left|\frac{m_{ij}}{\hat{m}_{ij}}\right|>x\right) \notag \\ 
&\leq \Pro\left( \left|z_{ij}/m_{ij}\right|>x/2\right) 
+ \Pro\left( \left|r_{ij}\right|\gtrsim x/2\right). \label{two_P}
\end{align}
We bound each probability as an exponential function of $x$. By Lemmas \ref{lem:ui_z/m} and \ref{lem:ui_2ineq}(b), the first probability of \eqref{two_P} is evaluated as 
\begin{align}
\Pro\left( \left|z_{ij}/m_{ij}\right| > x\right) 
&\lesssim \exp\left( -\frac{x^2}{2(x+1)} \right) + 
\Pro\left( \left\|\hat{\bSigma}_x-\bSigma_x\right\|_{\max} \gtrsim x/M_\omega^2\right) \notag \\
&\lesssim \exp\left( -\frac{x}{2(1+1/x)} \right) + 
K^2N^2T^3x^2\exp\left(-\frac{x^{2/3}T^{1/3}}{M_\omega^{4/3}}\right). \label{ui_z/m}
\end{align}
The second probability of \eqref{two_P} is 
\begin{align*}
&\Pro\left( \left|r_{ij}\right| \gtrsim x\right) \\
&\leq \Pro\left( \left\|\frac{1}{T^{3/4}}\bu_{i\cdot}\bX^{\prime}\right\|_\infty \left\|T^{1/4}(\hat{\bomega}_{j}-{\bomega}_{j})\right\|_1
+ \left\|T^{1/4}(\hat\bphi_{i\cdot}^{\normalfont{\textsf{L}}} - \bphi_{i\cdot}) \right\|_1 \left\|T^{1/4}(\hat{\bSigma}_{x}\hat{\bomega}_{j}-\be_{j})\right\|_\infty \gtrsim x\right) \\
&\leq \Pro\left( \left\|\frac{1}{T}\bu_{i\cdot}\bX^{\prime}\right\|_\infty \gtrsim \frac{x^{1/2}}{T^{1/4}}\right) 
+ \Pro\left(\left\|\hat{\bomega}_{j}-{\bomega}_{j}\right\|_1\gtrsim \frac{x^{1/2}}{T^{1/4}}\right) \\
&\qquad +\Pro\left(\left\|\hat\bphi_{i\cdot}^{\normalfont{\textsf{L}}} - \bphi_{i\cdot} \right\|_1\gtrsim \frac{x^{1/2}}{T^{1/4}}\right) 
+\Pro\left(  \left\|\hat{\bSigma}_{x}\hat{\bomega}_{j}-\be_{j}\right\|_\infty \gtrsim \frac{x^{1/2}}{T^{1/4}}\right).
\end{align*}
By Lemma \ref{lem:ui_3terms}, the last three therms are evaluated as
\begin{align*}
&\Pro\left(\left\|\hat{\bomega}_{j}-{\bomega}_{j}\right\|_1\gtrsim \frac{x^{1/2}}{T^{1/4}}\right) 
+\Pro\left(\left\|\hat\bphi_{i\cdot}^{\normalfont{\textsf{L}}} - \bphi_{i\cdot} \right\|_1\gtrsim \frac{x^{1/2}}{T^{1/4}}\right) 
+\Pro\left(  \left\|\hat{\bSigma}_{x}\hat{\bomega}_{j}-\be_{j}\right\|_\infty \gtrsim \frac{x^{1/2}}{T^{1/4}} \right) \\
&\lesssim \Pro\left( \left\|\frac{1}{T}\bu_{i\cdot}\bX^{\prime}\right\|_\infty \gtrsim \frac{x^{1/(2-2r)}}{T^{1/(4-4r)}s_\omega^{1/(1-r)}M_\omega^2}  \right) 
+ \Pro\left( \left\|\hat{\bSigma}_x-\bSigma_x\right\|_{\max} > \frac{x^{1/(2-2r)}}{T^{1/(4-4r)}s_\omega^{1/(1-r)}M_\omega^2} \right) \\
&\qquad + \Pro\left( \left\|\frac{1}{T}\bu_{i\cdot}\bX^{\prime}\right\|_\infty \gtrsim \frac{x^{1/2}}{T^{1/4}(s_i\vee M_\omega)} \right) 
+ \Pro\left( \left\|\hat{\bSigma}_x-\bSigma_x\right\|_{\max} \gtrsim \frac{x^{1/2}}{T^{1/4}(s_i\vee M_\omega)} \right).
\end{align*}
Thus by Lemma \ref{lem:ui_2ineq}, we obtain
\begin{align}
&\Pro\left( \left|r_{ij}\right| \gtrsim x\right) \notag \\
&\lesssim K^2N^2T^3x \left\{ \exp\left(-\frac{x^{1/(3-3r)}T^{(1-2r)/(6-6r)}}{s_\omega^{2/(3-3r)} M_\omega^{4/3}}\right) + \exp\left(-\frac{x^{1/3}T^{1/6}}{(s_i\vee M_\omega )^{2/3}}\right) \right\}. \label{ui_r}
\end{align}

Combining \eqref{two_P} with \eqref{ui_z/m} and \eqref{ui_r} gives
\begin{align*}
&\Pro\left(|\tT_{ij}|\gtrsim x\right) \lesssim \exp\left( -\frac{x}{2(1+1/x)} \right) + K^2N^2T^3x^2 \exp\left(-\frac{x^{2/3}T^{1/3}}{M_\omega^{4/3}}\right)\\
&\qquad + K^2N^2T^3x \left\{ \exp\left(-\frac{x^{1/(3-3r)}T^{(1-2r)/(6-6r)}}{s_\omega^{2/(3-3r)} M_\omega^{4/3}}\right) + \exp\left(-\frac{x^{1/3}T^{1/6}}{(s_i\vee M_\omega )^{2/3}}\right) \right\} \\
&\lesssim K^2N^2T^3x \left\{ \exp\left(-\frac{x^{1/(3-3r)}T^{(1-2r)/(6-6r)}}{s_\omega^{2/(3-3r)} M_\omega^{4/3}}\right) + \exp\left(-\frac{x^{1/3}T^{1/6}}{(s_i\vee M_\omega )^{2/3}}\right) \right\}. 
\end{align*}
Under the condition, $T\geq \{(s_\omega^{4/(1-2r)}M_\omega^{8(1-r)/(1-2r)})\vee s_i^4\}\log^{10}(N\vee T)$, the upper bound reduces to 
\begin{align*}
\Pro\left(|\tT_{ij}|\gtrsim x\right) 
&\lesssim \exp\left(- x^{1/3}\log^{1+\ep}(N\vee T) +\log x + 7\log(N\vee T)\right) 
\end{align*}
for some $\ep>0$. Taking the supremum over $(N,T,K)$ with the assumed condition yields 
\begin{align*}
\sup_{N,T,K}\Pro\left(|\tT_{ij}|\gtrsim x\right) 
&\lesssim \exp\left(- x^{1/3} +\log x \right).
\end{align*}
This completes the proof. 
\end{proof}

\subsection{Lemmas and their proofs}

\begin{lem}\label{lem:ui}
Any positive random sequence $\{X_n\}$ that satisfies 
\begin{align*}
\sup_n\Pro(X_n>x) \lesssim x^{-r}\log^s x 
\end{align*}
for some $r>1$ and $s>0$ is UI. 
\end{lem}
\begin{proof}[Proof of Lemma \ref{lem:ui}]
For any $\epsilon\in(1,r)$, we have 
\begin{align*}
\sup_n \E [X_n^{\epsilon}] 
= \sup_n\int_0^\infty \Pro \left(X_n>x^{1/\epsilon} \right) \diff x
\lesssim 1 + \epsilon^{-s} \int_1^\infty  x^{-r/\epsilon}\log^s x \diff x.
\end{align*}
Since $r/\epsilon>1$, the integral is finite. 
This condition implies that $\{X_n\}$ is UI; see \citep[pp.\ 31--32]{Billingsley1999}. This completes the proof.  
\end{proof}

\begin{lem}\label{lem:ui_suffice_exp}
For any positive random sequence $\{X_n\}$, if 
\begin{align*}
\Pro(X_n>x)\lesssim \exp(-\delta x^\alpha + \log x)
\end{align*}
holds for any $x>0$ and some $\alpha,\delta>0$, then $\exp(c X_n^\alpha)$ is UI for any $c\in[0,\delta)$. 
\end{lem}
\begin{proof}[Proof of Lemma \ref{lem:ui_suffice_exp}]
We have
\begin{align*}
\Pro\left(\exp(c X_n^\alpha) >x \right) 
&= \Pro\left( X_n > (\log x^{1/c})^{1/\alpha} \right) \\
&\lesssim \exp\left(-\delta \log x^{1/c} +\log(\log x^{1/c})^{1/\alpha} \right) 
\asymp x^{-\delta/c}\log^{1/\alpha} x. 
\end{align*}
Since $c\in[0,\delta)$ is assumed, we have $\delta/c>1$. Thus Lemma \ref{lem:ui} establishes the UI of $\exp(c X_n^\alpha)$. This completes the proof. 
\end{proof}

\begin{lem}\label{lem:ui_exp_poly}
For any positive random sequence $\{X_n\}$, if $\exp(cX_n^\alpha)$ is UI for given constants $c,\alpha>0$, then $X_n^p$ is also UI for any constant $p>0$. 
\end{lem}
\begin{proof}[Proof of Lemma \ref{lem:ui_exp_poly}]
For any $x>0$, we have $x^p\leq (2p/c)^{p/\alpha}\exp(cx^\alpha)$. Thus, for any $M>0$, we obtain
\begin{align*}
&\sup_n \E \left[ X_n^p 1\{X_n^p\geq M\} \right] \\
&\leq \sup_n \E \left[ (2p/c)^{p/\alpha} \exp(cX_n^\alpha)1\{(2p/c)^{p/\alpha} \exp(cX_n^\alpha)\geq M\}\right].
\end{align*}
Since $(2p/c)^{p/\alpha}$ is a constant, and hence $(2p/c)^{p/\alpha} \exp(cX_n^\alpha)$ is UI by the assumption, the upper bound can be made arbitrarily small by taking large $M$. This establishes the UI of $X_n^p$.  
\end{proof}

\begin{lem}\label{lem:ui_z/m}
For any $x>0$, we have 
\begin{align*}
\Pro\left( \left|z_{ij}/m_{ij}\right| > x\right) 
\lesssim \exp\left( -\frac{x^2}{2(x+1)} \right) + 
\Pro\left( \left\|\hat{\bSigma}_x-\bSigma_x\right\|_{\max} > \omega_jx/M_\omega^2\right).
\end{align*}
\end{lem}
\begin{proof}[Proof of Lemma \ref{lem:ui_z/m}]
Define 
\begin{align*}
M_T &= z_{ij}/m_{ij}=\frac{1}{\sqrt{T}}\sum_{t=1}^T\frac{u_{it}\bx_t'\bomega_j}{m_{ij}}, \\
\langle M \rangle_T &= \frac{1}{T} \sum_{t=1}^T\E\left[\frac{u_{it}^2(\bx_t'\bomega_j)^2}{m_{ij}^2}\mid \mathcal{F}_{t-1}\right] 
= \frac{\bomega_j'\hat{\bSigma}_x\bomega_j}{\bomega_j'{\bSigma}_x\bomega_j}.
\end{align*}
By \citep[Theorem 3.25]{Bercu2015} under the assumed Gaussianity of $u_{it}$, we have 
\begin{align*}
\Pro\left( M_T > x\right) 
&\leq \Pro\left( M_T > x,~ \langle M \rangle_T\leq x+1\right) 
+ \Pro\left( \langle M \rangle_T > x+1 \right) \\
&\leq \exp\left( -\frac{x^2}{2(x+1)} \right) 
+ \Pro\left( \frac{\bomega_j'\hat{\bSigma}_x\bomega_j}{\bomega_j'{\bSigma}_x\bomega_j} > x+1\right).
\end{align*}
By H\"older's inequality and $\|\bomega_j\|_1\leq M_\omega$ in Condition \ref{ass:invest}, the second probability of this upper bound is
\begin{align*}
\Pro\left( \frac{\bomega_j'\hat{\bSigma}_x\bomega_j}{\bomega_j'{\bSigma}_x\bomega_j} > x+1\right)
&= \Pro\left( \bomega_j'\left(\hat{\bSigma}_x-\bSigma_x\right)\bomega_j > \omega_jx\right) \\
&\leq \Pro\left( \left\|\hat{\bSigma}_x-\bSigma_x\right\|_{\max} > \omega_jx/M_\omega^2\right)
\end{align*}
This completes the proof. 
\end{proof}

\begin{lem}\label{lem:ui_3terms}
For any $\lambda>0$, we have the following:
\begin{align*}
(a)~~&\Pro\left(\left\|\hat{\bomega}_{j}-{\bomega}_{j}\right\|_1 > \lambda \right) \\
&\lesssim \Pro\left( \left\|\frac{1}{T}\bu_{i\cdot}\bX^{\prime}\right\|_\infty \gtrsim \frac{\lambda^{1/(1-r)}}{s_\omega^{1/(1-r)}M_\omega^2}  \right) + \Pro\left( \left\|\hat{\bSigma}_x-\bSigma_x\right\|_{\max} > \frac{\lambda^{1/(1-r)}}{s_\omega^{1/(1-r)}M_\omega^2} \right), \\
(b)~~&\Pro\left(\left\|\hat\bphi_{i\cdot}^{\normalfont{\textsf{L}}} - \bphi_{i\cdot} \right\|_1 > \lambda \right) \\
&\lesssim \Pro\left( \left\|\frac{1}{T}\bu_{i\cdot}\bX^{\prime}\right\|_\infty \gtrsim \frac{\lambda}{s_i} \right) 
+ \Pro\left( \left\|\hat{\bSigma}_x-\bSigma_x\right\|_{\max} \gtrsim \frac{\lambda}{s_i} \right), \\
(c)~~&\Pro\left(  \left\|\hat{\bSigma}_{x}\hat{\bomega}_{j}-\be_{j}\right\|_\infty > \lambda \right) \\
&\lesssim \Pro\left( \left\|\frac{1}{T}\bu_{i\cdot}\bX^{\prime}\right\|_\infty \gtrsim \frac{\lambda}{M_\omega} \right) + \Pro\left( \left\|\hat{\bSigma}_x-\bSigma_x\right\|_{\max} \gtrsim \frac{\lambda}{M_\omega} \right).
\end{align*}
\end{lem}
\begin{proof}[Proof of Lemma \ref{lem:ui_3terms}]
The proofs of (a) and (c) complete by the proof of Proposition \ref{thm:asynormal}. Similarly, the proof of Proposition \ref{thm:errbound} gives the proof of (b). 
\end{proof}

\begin{lem}\label{lem:ui_2ineq}
For any large $y>0$, we have the following:
\begin{align*}
(a)~~&\Pro\left( \left\|\frac{1}{T}\bu_{i\cdot}\bX^{\prime}\right\|_\infty > y \right) \lesssim KN^2T^2y\exp\left(-y^{2/3}T^{1/3}\right), \\
(b)~~&\Pro\left( \left\|\hat{\bSigma}_x-\bSigma_x\right\|_{\max} > y \right) \lesssim K^2N^2T^3y^2\exp\left(-y^{2/3}T^{1/3}\right).
\end{align*}
\end{lem}
\begin{proof}[Proof of Lemma \ref{lem:ui_2ineq}]
In Lemma \ref{lem:tail_yu}, setting $\bar{c}_1=y^{2/3}T^{1/3}$ and $r_T=yT$ yields 
\begin{align*}
&\Pro\left( \left\|\frac{1}{T}\bu_{i\cdot}\bX^{\prime}\right\|_\infty > y \right) \\
&\lesssim KN^2\exp\left(-y^{2/3}T^{1/3}\right) + KN^2T^2y\exp\left(-y^{2/3}T^{1/3}\right) + KT^{2/3}\frac{\delta^{yT}\log N}{y^{2/3}} \\
&\lesssim KN^2T^2y\exp\left(-y^{2/3}T^{1/3}\right),
\end{align*}
where the last inequality holds for any $\delta\in(0,1)$ as long as $y>0$ is sufficiently large. 
Similarly in Lemma \ref{lem:tail_yy}, setting $\bar{c}_2=y^{2/3}T^{1/3}$ and $r_T=yT$ yields 
\begin{align*}
&\Pro\left( \left\|\hat{\bSigma}_x-\bSigma_x\right\|_{\max} > y \right) \\
&\lesssim K^2N^2T^2y^2\exp\left(-y^{2/3}T^{1/3}\right) + K^2N^2T^3y^2\exp\left(-y^{2/3}T^{1/3}\right) + K^2 \frac{\delta^{yT}\log N}{y^{2/3}T^{1/3}} \\
&\lesssim K^2N^2T^3y^2\exp\left(-y^{2/3}T^{1/3}\right).
\end{align*}
This completes the proof. 
\end{proof}

\setcounter{table}{0}
\renewcommand{\thetable}{\thesection\arabic{table}}
\setcounter{figure}{0}
\renewcommand{\thefigure}{\thesection\arabic{figure}}
\section{Additional Experimental Results}


First, we report the results of a multiple test at $q=0.1$ in Figure \ref{figure:prelim2} to the results at $q=0.2$ in Figure \ref{figure:prelim}(d). Using the larger value of $q$, the overall red and blue color in Figure \ref{figure:prelim2} becomes weaker than in \ref{figure:prelim}(d), suggesting that all the elements tend to be selected less frequently.

\graphicspath{ {./images/} }
\begin{figure}[h!]
	\centering
	\begin{minipage}{0.24\hsize}
		\centering
		\includegraphics[width=1.00\linewidth]{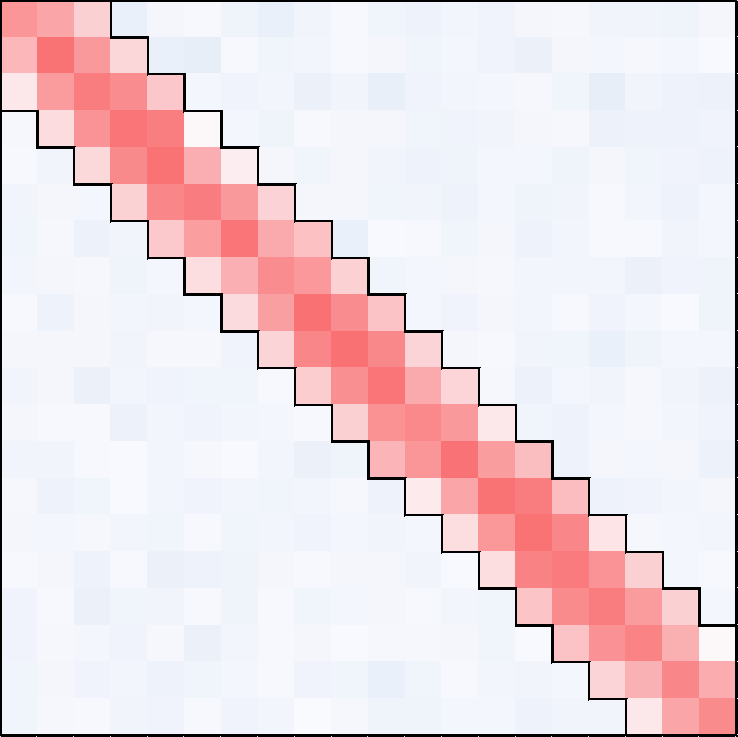}
		\subcaption{\footnotesize{multiple test,$q=0.1$}}
	\end{minipage} 
 \\
    \centering
 	\includegraphics[width=0.2\linewidth]{legend.red.blue.pdf}
	\caption{Heatmap of selection frequencies in  $\bPhi_1$}
	\label{figure:prelim2}
\end{figure}

Figure \ref{figure:phi1} shows a heatmap of $\bPhi$ used for Table \ref{table:fdr_het1} over three different degrees of sparsity (controlled by $d$), given $N$. Figure \ref{figure:phi2} shows the $\bPhi$ in the additional design. As can be seen from the cross shape, the new design differs from the first design in the following aspects: (i) more intensive clustering with fewer nodes; (ii) higher sparsity. We believe that considering these two different networks sheds light on the robustness properties of our inference method for different network structures. 

\graphicspath{ {./images/} }
\begin{figure}[h!]
	\centering
	\begin{minipage}{0.24\hsize}
		\centering
		\includegraphics[width=1.00\linewidth]{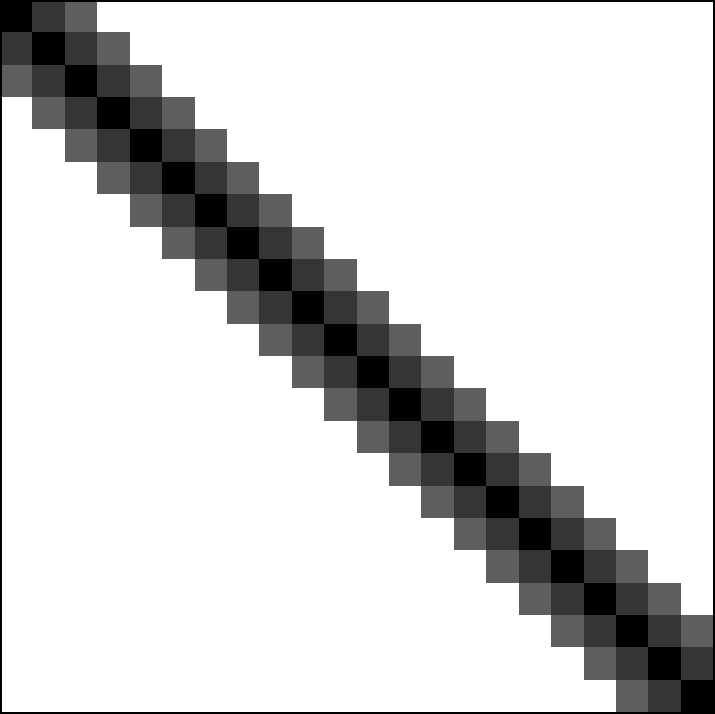}
		\subcaption{\footnotesize{$d=2$ $(m=5)$}}
	\end{minipage}
	\begin{minipage}{0.24\hsize}
		\centering
		\includegraphics[width=1.00\linewidth]{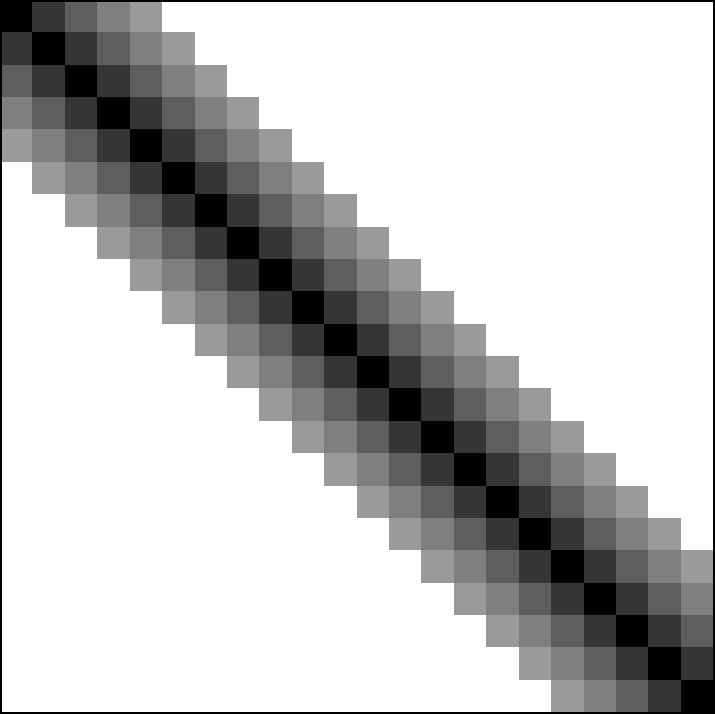}
		\subcaption{\footnotesize{$d=4$ $(m=9)$}}
	\end{minipage}
	\begin{minipage}{0.24\hsize}
		\centering
		\includegraphics[width=1.00\linewidth]{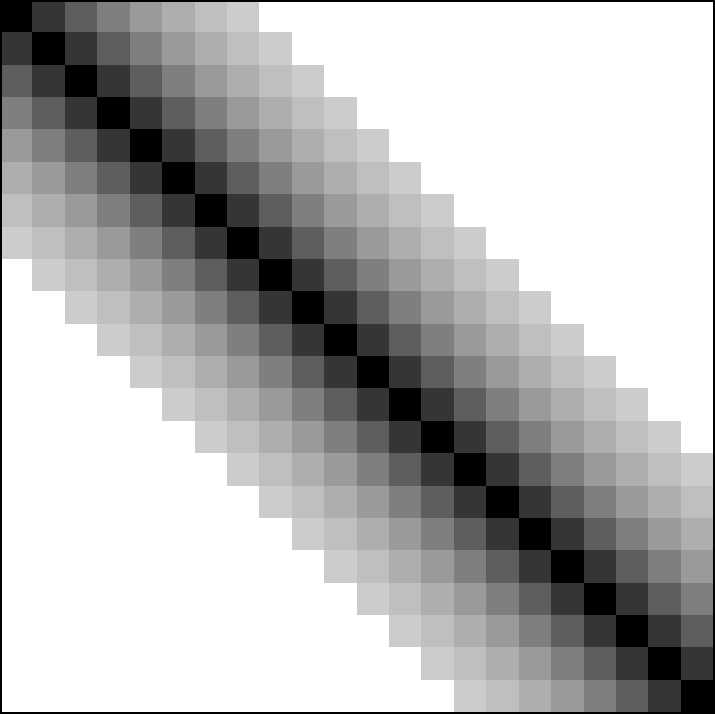}
		\subcaption{\footnotesize{$d=7$ $(m=15)$}}
	\end{minipage} 
	
	\centering
	\includegraphics[width=0.1\linewidth]{legend.black.pdf}
	\caption{Heatmap of absolute value of $\bPhi_1$, Design 1}
	\label{figure:phi1}
\end{figure}

\graphicspath{ {./images/} }
\begin{figure}[h!]
	\centering
	\begin{minipage}{0.24\hsize}
		\centering
		\includegraphics[width=1.00\linewidth]{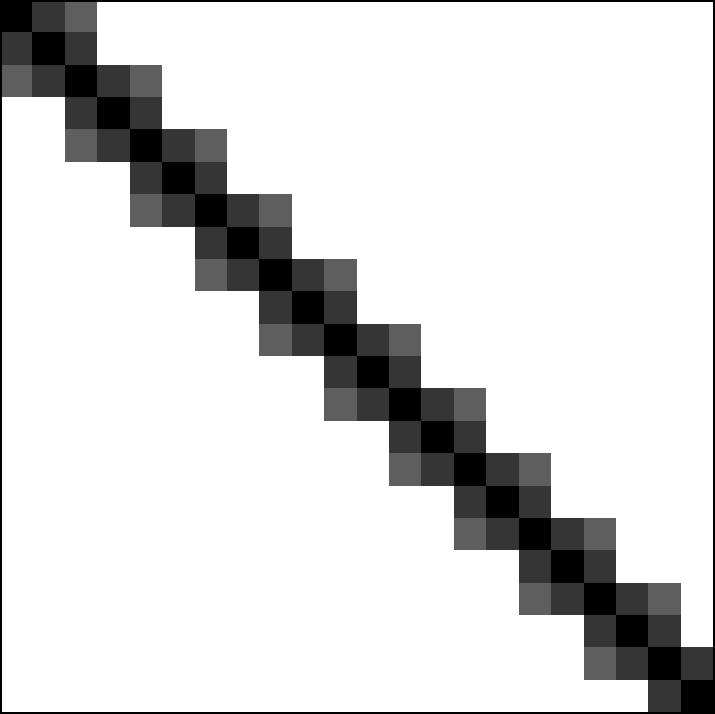}
		\subcaption{\footnotesize{$d=2$ $(m=5)$}}
	\end{minipage}
	\begin{minipage}{0.24\hsize}
		\centering
		\includegraphics[width=1.00\linewidth]{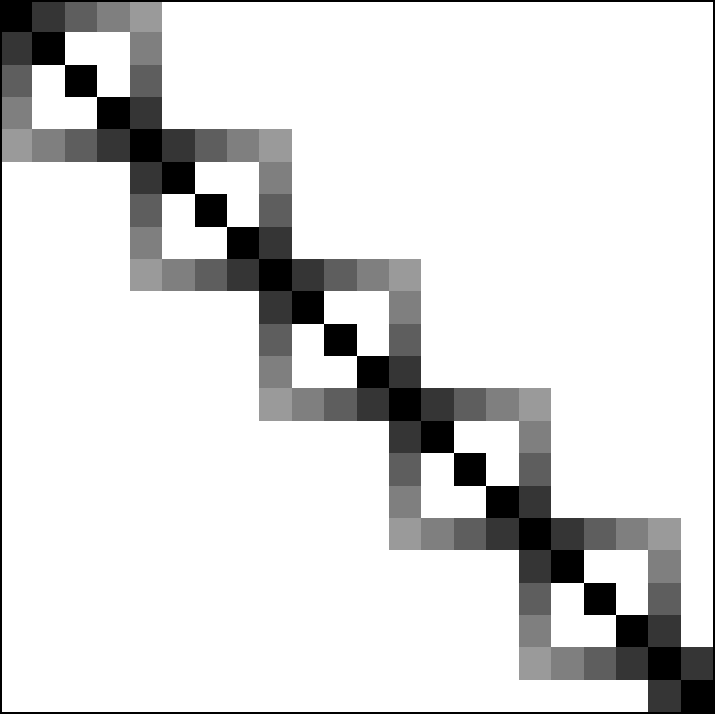}
		\subcaption{\footnotesize{$d=4$ $(m=9)$}}
	\end{minipage}
	\begin{minipage}{0.24\hsize}
		\centering
		\includegraphics[width=1.00\linewidth]{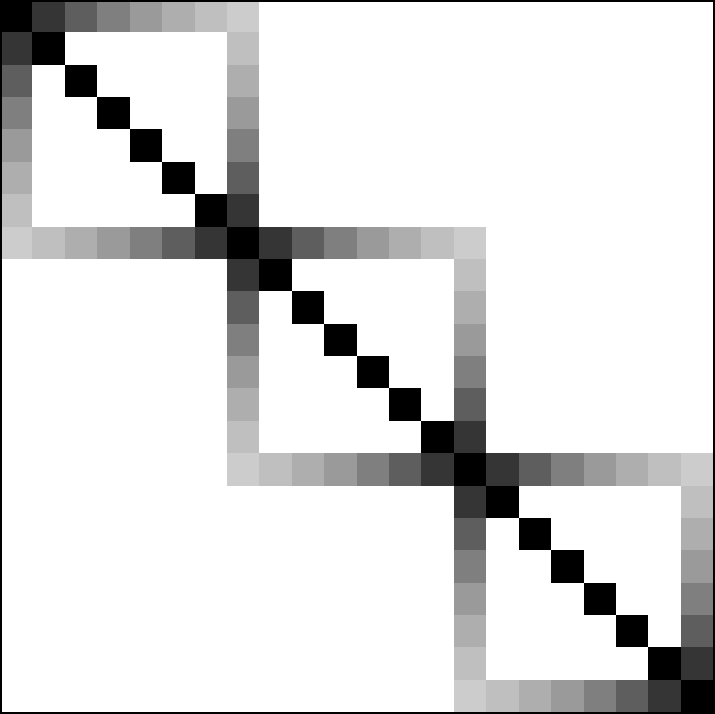}
		\subcaption{\footnotesize{$d=7$ $(m=15)$}}
	\end{minipage} 
	
	\centering
	\includegraphics[width=0.1\linewidth]{legend.black.pdf}
	\caption{Heatmap of absolute value of $\bPhi_1$, Design 2}
	\label{figure:phi2}
\end{figure}

Now we report the results with $\bPhi_1$ design 2 corresponding to Tables \ref{table:fdr_het1} and \ref{table:fdr_het1_ev}. The DGPs are identical to the ones used in Tables \ref{table:fdr_het1} and \ref{table:fdr_het1_ev}, except for $\bPhi_1$. The results are reported in Tables \ref{table:fdr_het1_F4} and \ref{table:fdr_het1_ev2} below. As can be seen, the results are qualitatively very similar to those in Tables \ref{table:fdr_het1} and \ref{table:fdr_het1_ev}, in the sense that the asymptotic threshold often fails to control the FDR for mixture normal errors, but the bootstrap and e-BH thresholds control much better. The e-BH has a slight loss of power compared to the bootstrap threshold due to the conservative type I error.
\renewcommand{\arraystretch}{0.55}
	\begin{table}[H]
		\caption{Directional FDR and power using asymptotic and bootstrap thresholds for $q=0.1$, with cross-sectionally uncorrelated errors, $\bPhi_1$ Design 2}	
		\label{table:fdr_het1_F4}
		\centering
\begin{tabular}
[c]{rlrrrrrrrrrrrr}\hline
\multicolumn{14}{c}{$m=2$ ($\max s_{i}=5$)}\\\hline
&  &  & \multicolumn{5}{c}{$T=200$} &  & \multicolumn{5}{c}{$T=300$%
}\\\cline{4-8}\cline{10-14}
&  &  & \multicolumn{2}{c}{asymptotic} &  & \multicolumn{2}{c}{bootstrap} &  &
\multicolumn{2}{c}{asymptotic} &  & \multicolumn{2}{c}{bootstrap}%
\\\cline{4-5}\cline{7-8}\cline{10-11}\cline{13-14}
&  &  & \multicolumn{1}{c}{dFDR} & \multicolumn{1}{c}{PWR} &  &
\multicolumn{1}{c}{dFDR} & \multicolumn{1}{c}{PWR} &  &
\multicolumn{1}{c}{dFDR} & \multicolumn{1}{c}{PWR} &  &
\multicolumn{1}{c}{dFDR} & \multicolumn{1}{c}{PWR}\\
\multicolumn{14}{l}{Standard Normal Error}\\
& $N=50$ &  & 10.3 & 98.3 &  & 7.6 & 97.8 &  & 9.5 & 99.9 &  & 7.7 & 99.8\\
& $N=100$ &  & 11.5 & 96.2 &  & 7.6 & 95.2 &  & 10.4 & 99.6 &  & 7.9 & 99.5\\
& $N=200$ &  & 11.8 & 93.9 &  & 7.3 & 92.2 &  & 11.2 & 99.3 &  & 8.1 & 99.1\\
& $N=300$ &  & 11.0 & 92.0 &  & 6.3 & 89.2 &  & 10.6 & 99.0 &  & 8.2 & 98.8\\
\multicolumn{14}{l}{Mixture Normal Error}\\
& $N=50$ &  & 13.5 & 95.5 &  & 7.4 & 93.5 &  & 11.7 & 99.4 &  & 7.7 & 99.1\\
& $N=100$ &  & 15.5 & 92.1 &  & 6.8 & 88.1 &  & 13.5 & 98.4 &  & 7.7 & 97.7\\
& $N=200$ &  & 14.3 & 89.4 &  & 6.4 & 84.8 &  & 13.7 & 97.7 &  & 7.5 & 96.7\\
& $N=300$ &  & 16.2 & 87.0 &  & 5.3 & 79.2 &  & 14.1 & 97.1 &  & 7.7 &
96.0\\\hline
\multicolumn{14}{c}{$m=4$ ($\max s_{i}=9$)}\\\hline
&  &  & \multicolumn{5}{c}{$T=200$} &  & \multicolumn{5}{c}{$T=300$%
}\\\cline{4-8}\cline{10-14}
&  &  & \multicolumn{2}{c}{asymptotic} &  & \multicolumn{2}{c}{bootstrap} &  &
\multicolumn{2}{c}{asymptotic} &  & \multicolumn{2}{c}{bootstrap}%
\\\cline{4-5}\cline{7-8}\cline{10-11}\cline{13-14}
&  &  & \multicolumn{1}{c}{dFDR} & \multicolumn{1}{c}{PWR} &  &
\multicolumn{1}{c}{dFDR} & \multicolumn{1}{c}{PWR} &  &
\multicolumn{1}{c}{dFDR} & \multicolumn{1}{c}{PWR} &  &
\multicolumn{1}{c}{dFDR} & \multicolumn{1}{c}{PWR}\\
\multicolumn{14}{l}{Standard Normal Error}\\
& $N=50$ &  & 10.5 & 88.2 &  & 7.2 & 86.0 &  & 9.7 & 96.2 &  & 7.6 & 95.6\\
& $N=100$ &  & 11.1 & 83.0 &  & 7.4 & 80.5 &  & 10.2 & 93.7 &  & 7.8 & 92.9\\
& $N=200$ &  & 10.8 & 77.6 &  & 7.2 & 74.9 &  & 10.3 & 91.1 &  & 8.2 & 90.2\\
& $N=300$ &  & 11.5 & 73.0 &  & 5.8 & 67.7 &  & 10.4 & 88.6 &  & 8.1 & 87.6\\
\multicolumn{14}{l}{Mixture Normal Error}\\
& $N=50$ &  & 14.7 & 83.6 &  & 7.4 & 77.7 &  & 12.2 & 93.9 &  & 7.8 & 92.1\\
& $N=100$ &  & 16.2 & 78.7 &  & 7.2 & 72.3 &  & 13.9 & 90.5 &  & 8.2 & 87.9\\
& $N=200$ &  & 15.8 & 73.8 &  & 6.5 & 66.1 &  & 13.7 & 87.8 &  & 7.9 & 84.8\\
& $N=300$ &  & 18.6 & 69.5 &  & 4.9 & 56.3 &  & 15.7 & 85.1 &  & 7.9 &
81.0\\\hline
\multicolumn{14}{c}{$m=7$ ($\max s_{i}=15$)}\\\hline
&  &  & \multicolumn{5}{c}{$T=200$} &  & \multicolumn{5}{c}{$T=300$%
}\\\cline{4-8}\cline{10-14}
&  &  & \multicolumn{2}{c}{asymptotic} &  & \multicolumn{2}{c}{bootstrap} &  &
\multicolumn{2}{c}{asymptotic} &  & \multicolumn{2}{c}{bootstrap}%
\\\cline{4-5}\cline{7-8}\cline{10-11}\cline{13-14}
&  &  & \multicolumn{1}{c}{dFDR} & \multicolumn{1}{c}{PWR} &  &
\multicolumn{1}{c}{dFDR} & \multicolumn{1}{c}{PWR} &  &
\multicolumn{1}{c}{dFDR} & \multicolumn{1}{c}{PWR} &  &
\multicolumn{1}{c}{dFDR} & \multicolumn{1}{c}{PWR}\\
\multicolumn{14}{l}{Standard Normal Error}\\
& $N=50$ &  & 11.5 & 64.6 &  & 7.0 & 61.0 &  & 10.4 & 75.5 &  & 7.5 & 73.3\\
& $N=100$ &  & 11.9 & 61.7 &  & 7.6 & 58.8 &  & 10.8 & 72.8 &  & 8.2 & 71.1\\
& $N=200$ &  & 12.2 & 54.6 &  & 6.2 & 50.6 &  & 11.7 & 67.5 &  & 7.5 & 65.1\\
& $N=300$ &  & 12.3 & 51.0 &  & 4.7 & 44.9 &  & 11.0 & 64.0 &  & 7.6 & 62.0\\
\multicolumn{14}{l}{Mixture Normal Error}\\
& $N=50$ &  & 19.5 & 62.3 &  & 2.0 & 44.2 &  & 15.3 & 72.9 &  & 8.3 & 68.0\\
& $N=100$ &  & 20.1 & 60.0 &  & 6.5 & 49.7 &  & 16.2 & 70.2 &  & 9.1 & 65.9\\
& $N=200$ &  & 18.7 & 51.9 &  & 4.2 & 39.0 &  & 16.3 & 64.1 &  & 7.2 & 58.6\\
& $N=300$ &  & 22.7 & 49.6 &  & 1.4 & 29.5 &  & 18.7 & 61.6 &  & 7.6 &
55.6\\\hline
\end{tabular}

	\end{table}

\setlength{\tabcolsep}{3.3pt}%
\renewcommand{\arraystretch}{0.95}
\begin{table}[htb!]
	\caption{FDR and power using e-BH thresholds for $q=0.1$, with cross-sectionally uncorrelated errors, $\bPhi_1$ Design 2}	
	\label{table:fdr_het1_ev2}
	\centering
\begin{tabular}
[c]{rlrrrrrrrrrrrr}\hline
\multicolumn{14}{c}{$m=2$ ($\max s_{i}=5$)}\\\hline
&  &  & \multicolumn{5}{c}{$T=200$} &  & \multicolumn{5}{c}{$T=300$%
}\\\cline{4-8}\cline{10-14}%
\multicolumn{1}{c}{} & \multicolumn{1}{c}{$f(x)$} & \multicolumn{1}{c}{} &
\multicolumn{2}{c}{$|x|^{p}$} &  & \multicolumn{2}{c}{$\exp(c|x|^{\alpha})$} &
\multicolumn{1}{c}{} & \multicolumn{2}{c}{$|x|^{p}$} &  &
\multicolumn{2}{c}{$\exp(c|x|^{\alpha})$}\\\cline{4-5}\cline{7-8}%
\cline{10-11}\cline{13-14}
&  &  & \multicolumn{1}{c}{FDR} & \multicolumn{1}{c}{PWR} &  &
\multicolumn{1}{c}{FDR} & \multicolumn{1}{c}{PWR} &  & \multicolumn{1}{c}{FDR}
& \multicolumn{1}{c}{PWR} &  & \multicolumn{1}{c}{FDR} &
\multicolumn{1}{c}{PWR}\\
\multicolumn{14}{l}{Standard Normal Error}\\
& $N=50$ &  & 1.9 & 95.2 &  & 1.3 & 94.1 &  & 1.4 & 99.4 &  & 0.9 & 99.1\\
& $N=100$ &  & 2.2 & 91.4 &  & 1.5 & 90.0 &  & 1.7 & 98.7 &  & 1.1 & 98.3\\
& $N=200$ &  & 2.0 & 86.6 &  & 1.4 & 85.0 &  & 1.6 & 97.8 &  & 1.1 & 97.4\\
& $N=300$ &  & 1.3 & 82.1 &  & 1.1 & 80.4 &  & 1.2 & 96.9 &  & 0.9 & 96.4\\
\multicolumn{14}{l}{Mixture Normal Error}\\
& $N=50$ &  & 3.1 & 90.1 &  & 2.3 & 88.5 &  & 2.3 & 97.9 &  & 1.6 & 97.4\\
& $N=100$ &  & 3.8 & 84.7 &  & 2.7 & 82.8 &  & 2.8 & 96.0 &  & 2.0 & 95.3\\
& $N=200$ &  & 2.6 & 78.8 &  & 2.0 & 76.8 &  & 2.4 & 94.2 &  & 1.8 & 93.4\\
& $N=300$ &  & 2.6 & 74.6 &  & 2.1 & 72.9 &  & 2.1 & 92.8 &  & 1.6 &
91.9\\\hline
\multicolumn{14}{c}{$m=4$ ($\max s_{i}=9$)}\\\hline
&  &  & \multicolumn{5}{c}{$T=200$} &  & \multicolumn{5}{c}{$T=300$%
}\\\cline{4-8}\cline{10-14}%
\multicolumn{1}{c}{} & \multicolumn{1}{c}{$f(x)$} & \multicolumn{1}{c}{} &
\multicolumn{2}{c}{$|x|^{p}$} &  & \multicolumn{2}{c}{$\exp(c|x|^{\alpha})$} &
\multicolumn{1}{c}{} & \multicolumn{2}{c}{$|x|^{p}$} &  &
\multicolumn{2}{c}{$\exp(c|x|^{\alpha})$}\\\cline{4-5}\cline{7-8}%
\cline{10-11}\cline{13-14}
&  &  & \multicolumn{1}{c}{FDR} & \multicolumn{1}{c}{PWR} &  &
\multicolumn{1}{c}{FDR} & \multicolumn{1}{c}{PWR} &  & \multicolumn{1}{c}{FDR}
& \multicolumn{1}{c}{PWR} &  & \multicolumn{1}{c}{FDR} &
\multicolumn{1}{c}{PWR}\\
\multicolumn{14}{l}{Standard Normal Error}\\
& $N=50$ &  & 2.1 & 78.3 &  & 1.4 & 76.0 &  & 1.6 & 91.2 &  & 1.1 & 89.9\\
& $N=100$ &  & 2.0 & 72.2 &  & 1.4 & 70.1 &  & 1.6 & 87.5 &  & 1.1 & 86.0\\
& $N=200$ &  & 1.4 & 64.4 &  & 1.0 & 62.4 &  & 1.3 & 83.3 &  & 0.9 & 81.6\\
& $N=300$ &  & 1.3 & 58.3 &  & 1.0 & 56.7 &  & 1.1 & 79.0 &  & 0.8 & 77.5\\
\multicolumn{14}{l}{Mixture Normal Error}\\
& $N=50$ &  & 3.8 & 72.6 &  & 2.8 & 70.0 &  & 2.8 & 87.4 &  & 2.0 & 85.6\\
& $N=100$ &  & 4.0 & 67.4 &  & 2.9 & 64.9 &  & 3.1 & 82.9 &  & 2.1 & 81.0\\
& $N=200$ &  & 2.9 & 59.5 &  & 2.2 & 57.3 &  & 2.4 & 78.3 &  & 1.7 & 76.4\\
& $N=300$ &  & 3.0 & 53.7 &  & 2.5 & 52.2 &  & 2.5 & 74.0 &  & 1.9 &
72.6\\\hline
\multicolumn{14}{c}{$m=7$ ($\max s_{i}=15$)}\\\hline
&  &  & \multicolumn{5}{c}{$T=200$} &  & \multicolumn{5}{c}{$T=300$%
}\\\cline{4-8}\cline{10-14}%
\multicolumn{1}{c}{} & \multicolumn{1}{c}{$f(x)$} & \multicolumn{1}{c}{} &
\multicolumn{2}{c}{$|x|^{p}$} &  & \multicolumn{2}{c}{$\exp(c|x|^{\alpha})$} &
\multicolumn{1}{c}{} & \multicolumn{2}{c}{$|x|^{p}$} &  &
\multicolumn{2}{c}{$\exp(c|x|^{\alpha})$}\\\cline{4-5}\cline{7-8}%
\cline{10-11}\cline{13-14}
&  &  & \multicolumn{1}{c}{FDR} & \multicolumn{1}{c}{PWR} &  &
\multicolumn{1}{c}{FDR} & \multicolumn{1}{c}{PWR} &  & \multicolumn{1}{c}{FDR}
& \multicolumn{1}{c}{PWR} &  & \multicolumn{1}{c}{FDR} &
\multicolumn{1}{c}{PWR}\\
\multicolumn{14}{l}{Standard Normal Error}\\
& $N=50$ &  & 2.1 & 54.0 &  & 1.5 & 51.9 &  & 1.7 & 65.7 &  & 1.1 & 63.4\\
& $N=100$ &  & 1.9 & 51.5 &  & 1.3 & 49.6 &  & 1.7 & 63.1 &  & 1.1 & 61.2\\
& $N=200$ &  & 1.5 & 43.4 &  & 1.1 & 42.2 &  & 1.5 & 57.6 &  & 1.1 & 56.2\\
& $N=300$ &  & 1.1 & 39.1 &  & 0.9 & 38.5 &  & 0.9 & 53.2 &  & 0.7 & 52.4\\
\multicolumn{14}{l}{Mixture Normal Error}\\
& $N=50$ &  & 5.4 & 51.4 &  & 4.0 & 49.4 &  & 3.7 & 62.6 &  & 2.6 & 60.5\\
& $N=100$ &  & 5.1 & 49.2 &  & 3.9 & 47.4 &  & 3.7 & 60.2 &  & 2.7 & 58.2\\
& $N=200$ &  & 3.3 & 39.5 &  & 2.7 & 38.3 &  & 2.8 & 53.2 &  & 2.1 & 51.7\\
& $N=300$ &  & 3.6 & 36.6 &  & 3.2 & 36.0 &  & 2.8 & 50.2 &  & 2.4 &
49.4\\\hline
\end{tabular}

\end{table}
 
	\subsection{Results with $\hat{m}_{ij}=\hat{\sigma}_i \hat{\omega}_{j}$}
 
The results based on $\hat{m}_{ij}=\hat{\sigma}_i\hat{\omega}_j$ are reported in Table \ref{table:fdr_het1_b}.
The DGP is identical to the one used in Table \ref{table:fdr_het1}. The results with this $\hat{m}_{ij}$ are shown in Table \ref{table:fdr_het1_b}. The e-BH result omitted. 
Bootstrap dFDR is virtually identical to bootstrap dFDR based on $\hat{m}_{ij}=\hat{\sigma}_i\sqrt{\hat{\bomega}'_{j} \hat{\bSigma}_x\hat{\bomega}_{j}}$, whilst asymptotic dFDR often becomes more conservative than bootstrap dFDR. As a result, the power is often lower than the bootstrap one. These results suggest that Procedure \ref{proc:boot} is expected to provide stable performance regardless of the choice of $\hat{m}_{ij}$.

	\setlength{\tabcolsep}{3.3pt}%
\renewcommand{\arraystretch}{0.55}
	\begin{table}[H]
		\caption{Directional FDR and power using asymptotic and bootstrap thresholds for $q=0.1$, with cross-sectionally uncorrelated errors, $\hat{m}_{ij}=\hat{\sigma}_i \hat{\omega}_{j}$}	
		\label{table:fdr_het1_b}
		\centering
		
		\begin{tabular}
			[c]{rlrrrrrrrrrrrr}\hline
			\multicolumn{14}{c}{$m=2$ ($\max s_{i}=5$)}\\\hline
			&  &  & \multicolumn{5}{c}{$T=200$} &  & \multicolumn{5}{c}{$T=300$}%
			\\\cline{4-8}\cline{10-14}
			&  &  & \multicolumn{2}{c}{asymptotic} &  & \multicolumn{2}{c}{bootstrap} &  &
			\multicolumn{2}{c}{asymptotic} &  & \multicolumn{2}{c}{bootstrap}%
			\\\cline{4-5}\cline{7-8}\cline{10-11}\cline{13-14}
			&  &  & \multicolumn{1}{c}{dFDR} & \multicolumn{1}{c}{PWR} &  &
			\multicolumn{1}{c}{dFDR} & \multicolumn{1}{c}{PWR} &  &
			\multicolumn{1}{c}{dFDR} & \multicolumn{1}{c}{PWR} &  &
			\multicolumn{1}{c}{dFDR} & \multicolumn{1}{c}{PWR}\\
			\multicolumn{14}{l}{Standard Normal Error}\\
			& $N=50$ &  & 6.0 & 96.6 &  & 6.7 & 96.9 &  & 6.0 & 99.7 &  & 6.9 & 99.7\\
			& $N=100$ &  & 6.3 & 92.8 &  & 6.9 & 93.2 &  & 6.1 & 99.1 &  & 7.3 & 99.2\\
			& $N=200$ &  & 6.0 & 89.6 &  & 7.5 & 90.7 &  & 6.0 & 98.5 &  & 8.2 & 98.8\\
			& $N=300$ &  & 5.7 & 86.7 &  & 6.7 & 87.4 &  & 5.5 & 98.0 &  & 8.5 & 98.4\\
			\multicolumn{14}{l}{Mixture Normal Error}\\
			& $N=50$ &  & 6.8 & 92.7 &  & 5.8 & 92.0 &  & 6.6 & 98.7 &  & 6.2 & 98.7\\
			& $N=100$ &  & 7.7 & 86.8 &  & 5.6 & 84.7 &  & 7.3 & 97.2 &  & 6.5 & 96.9\\
			& $N=200$ &  & 8.5 & 84.2 &  & 6.7 & 82.5 &  & 7.7 & 96.3 &  & 7.7 & 96.3\\
			& $N=300$ &  & 8.4 & 80.8 &  & 5.5 & 77.3 &  & 7.7 & 95.1 &  & 7.7 &
			95.1\\\hline
			\multicolumn{14}{c}{$m=4$ ($\max s_{i}=9$)}\\\hline
			&  &  & \multicolumn{5}{c}{$T=200$} &  & \multicolumn{5}{c}{$T=300$}%
			\\\cline{4-8}\cline{10-14}
			&  &  & \multicolumn{2}{c}{asymptotic} &  & \multicolumn{2}{c}{bootstrap} &  &
			\multicolumn{2}{c}{asymptotic} &  & \multicolumn{2}{c}{bootstrap}%
			\\\cline{4-5}\cline{7-8}\cline{10-11}\cline{13-14}
			&  &  & \multicolumn{1}{c}{dFDR} & \multicolumn{1}{c}{PWR} &  &
			\multicolumn{1}{c}{dFDR} & \multicolumn{1}{c}{PWR} &  &
			\multicolumn{1}{c}{dFDR} & \multicolumn{1}{c}{PWR} &  &
			\multicolumn{1}{c}{dFDR} & \multicolumn{1}{c}{PWR}\\
			\multicolumn{14}{l}{Standard Normal Error}\\
			& $N=50$ &  & 5.8 & 87.1 &  & 6.7 & 87.9 &  & 5.8 & 96.0 &  & 6.8 & 96.4\\
			& $N=100$ &  & 5.8 & 81.2 &  & 8.3 & 83.6 &  & 5.9 & 93.7 &  & 8.4 & 94.8\\
			& $N=200$ &  & 5.9 & 75.3 &  & 9.3 & 78.3 &  & 5.6 & 90.5 &  & 9.4 & 92.2\\
			& $N=300$ &  & 6.1 & 69.3 &  & 8.5 & 71.7 &  & 5.6 & 87.3 &  & 9.7 & 89.5\\
			\multicolumn{14}{l}{Mixture Normal Error}\\
			& $N=50$ &  & 6.0 & 82.9 &  & 5.8 & 82.6 &  & 5.9 & 93.8 &  & 6.1 & 94.0\\
			& $N=100$ &  & 6.4 & 77.3 &  & 7.1 & 78.1 &  & 6.3 & 91.1 &  & 7.6 & 91.8\\
			& $N=200$ &  & 7.2 & 71.6 &  & 8.2 & 72.6 &  & 6.8 & 87.2 &  & 9.0 & 88.5\\
			& $N=300$ &  & 7.4 & 65.8 &  & 6.8 & 65.1 &  & 6.8 & 83.7 &  & 8.9 &
			85.0\\\hline
			\multicolumn{14}{c}{$m=7$ ($\max s_{i}=15$)}\\\hline
			&  &  & \multicolumn{5}{c}{$T=200$} &  & \multicolumn{5}{c}{$T=300$}%
			\\\cline{4-8}\cline{10-14}
			&  &  & \multicolumn{2}{c}{asymptotic} &  & \multicolumn{2}{c}{bootstrap} &  &
			\multicolumn{2}{c}{asymptotic} &  & \multicolumn{2}{c}{bootstrap}%
			\\\cline{4-5}\cline{7-8}\cline{10-11}\cline{13-14}
			&  &  & \multicolumn{1}{c}{dFDR} & \multicolumn{1}{c}{PWR} &  &
			\multicolumn{1}{c}{dFDR} & \multicolumn{1}{c}{PWR} &  &
			\multicolumn{1}{c}{dFDR} & \multicolumn{1}{c}{PWR} &  &
			\multicolumn{1}{c}{dFDR} & \multicolumn{1}{c}{PWR}\\
			\multicolumn{14}{l}{Standard Normal Error}\\
			& $N=50$ &  & 5.1 & 66.4 &  & 6.0 & 67.8 &  & 5.2 & 78.8 &  & 6.1 & 79.9\\
			& $N=100$ &  & 5.6 & 58.1 &  & 8.0 & 61.0 &  & 5.6 & 72.3 &  & 8.0 & 74.7\\
			& $N=200$ &  & 5.7 & 50.2 &  & 8.8 & 53.3 &  & 5.4 & 65.4 &  & 9.1 & 68.7\\
			& $N=300$ &  & 5.8 & 47.1 &  & 8.8 & 49.8 &  & 5.2 & 62.6 &  & 9.8 & 66.2\\
			\multicolumn{14}{l}{Mixture Normal Error}\\
			& $N=50$ &  & 5.0 & 63.1 &  & 4.9 & 62.9 &  & 4.9 & 76.2 &  & 5.1 & 76.5\\
			& $N=100$ &  & 5.9 & 55.5 &  & 6.7 & 56.6 &  & 5.9 & 69.6 &  & 7.2 & 71.0\\
			& $N=200$ &  & 7.5 & 48.2 &  & 7.5 & 48.0 &  & 6.6 & 62.8 &  & 8.1 & 64.3\\
			& $N=300$ &  & 6.8 & 45.7 &  & 7.3 & 46.2 &  & 6.2 & 60.4 &  & 9.0 &
			62.8\\\hline
		\end{tabular}
	\end{table}
	
	\subsection{Results with cross-sectionally correlated errors}
	We have chosen the subset of the design in Table \ref{table:fdr_het1}, \{$m=2$, $T=200$, mixture normal errors\}, which gives the worst FDR control by the asymptotic threshold therein. The DGP is identical to that used in Tables \ref{table:fdr_het1} and \ref{table:fdr_het1_ev}, except that the errors are cross-sectionally correlated. 

 We consider three designs. In the first design, for a given cross section unit $i$, the maximum number of cross-correlated errors is an integer part of $2N^{0.6}-1$. 
Specifically, the cross-correlated errors are generated as follows: $\bSigma_{u}=(\sigma_{ij})$ with $\sigma_{ij}= k~1\{| i-j+1| \leq \varphi\}$ with $k=i/j$ for $i \leq j$ and $k=j/i$ for $i \geq j$, and $\varphi$ is an integer part of $N^{\eta}$. We set $\eta=0.6$. 
 The results are reported in Table \ref{table:fdr_het2}. The performance of the asymptotic and bootstrap thresholds, as well as that of the e-BH thresholds, is very similar to the case with cross-sectionally uncorrelated errors reported in Table \ref{table:fdr_het1}. For $q=0.05$, the comparative analysis of the results from different thresholds is very similar to that for $q=0.10$.

 The second and the third designs are more challenging with respect to Condition 6, in which all cross section units are correlated. In the second design, $\sigma_{ij}=\rho_u^{|i-j|}$ with $\rho_u=0.8$. In the third design, we consider the correlation implied by an error factor model, $
u_{ti}=(f_{t}b_{i}+v_{ti})/\sqrt{\sigma_{f}^{2}+1}$, where $f_{t}\sim IIDN(0,\sigma_{f}^{2})$ and $v_{ti}\sim IIDN(0,1)$ are drawn
for each replication but $b_{i}\sim IIDU(-\sqrt{3},\sqrt{3})$ is fixed over
the replications. This structure gives pervasive cross-correlations of
$u_{ti}$, as $\sigma_{ij}=E[u_{ti}u_{tj}]=\varphi_{\sigma_{f}^{2}}b_{i}b_{j}$
for $i\neq j$ and $\varphi_{\sigma_{f}^{2}}(b_{i}^{2}+1)$ for $i=j$, given
$\mathbf{b}=(b_{1},...,b_{N})^{\prime}$ and $\sigma_{f}^{2}$, where
$\varphi_{\sigma_{f}^{2}}=\sigma_{f}^{2}/(\sigma_{f}^{2}+1)$. Rather than
generating $u_{ti}$ by the factor model, we generate $\mathbf{u}_{t}%
=\boldsymbol{\Sigma}_{u}^{1/2}\boldsymbol{\varepsilon}_{t}$, where
$\boldsymbol{\varepsilon}_{t}\sim IIDN(\mathbf{0},\mathbf{I}_{N})$,
constructing $\boldsymbol{\Sigma}_{u}$ using $\mathbf{b}$ and $\sigma_{f}%
^{2}=1/4$. The corresponding results are summarized in Tables \ref{table:fdr_het7} and \ref{table:fdr_het6}, respectively. As can be seen, for $T=200$, both the asymptotic and the bootstrap dFDR exceed the level $q=0.1$ and deviate further as $N$ increases, while both versions of the e-BH thresholds control the FDR much better, keeping it well below 10\%.

 \setlength{\tabcolsep}{3.3pt}%
\renewcommand{\arraystretch}{1}
	\begin{table}[H]
		\caption{FDR and power using asymptotic, bootstrap and e-BH thresholds
			for $q=0.10,0.05$, for a given $i$ the maximum number of cross-correlated errors is $2N^{0.6}-1$}
		\label{table:fdr_het2}
		\centering
\begin{tabular}
[c]{rlrccrccrccrcc}\hline
\multicolumn{14}{c}{$m=2$ ($\max s_{i}=5$), $T=200$, Mixture normal,
cross-correlated errors}\\\hline
&  &  & \multicolumn{2}{c}{asymptotic} &  & \multicolumn{2}{c}{bootstrap} &  &
\multicolumn{2}{c}{$|x|^{p}$} &  & \multicolumn{2}{c}{$\exp(c|x|^{\alpha})$%
}\\\cline{4-5}\cline{7-8}\cline{10-11}\cline{13-14}
&  &  & dFDR & PWR &  & dFDR & PWR &  & FDR & PWR &  & FDR & PWR\\
\multicolumn{5}{l}{$q=0.10$} &  & \multicolumn{1}{r}{} & \multicolumn{1}{r}{}
&  & \multicolumn{1}{r}{} & \multicolumn{1}{r}{} &  & \multicolumn{1}{r}{} &
\multicolumn{1}{r}{}\\
& $N=50$ &  & \multicolumn{1}{r}{9.6} & \multicolumn{1}{r}{90.1} &  &
\multicolumn{1}{r}{6.5} & \multicolumn{1}{r}{88.0} &  &
\multicolumn{1}{r}{1.9} & \multicolumn{1}{r}{80.0} &  &
\multicolumn{1}{r}{1.4} & \multicolumn{1}{r}{77.3}\\
& $N=100$ &  & \multicolumn{1}{r}{10.2} & \multicolumn{1}{r}{83.9} &  &
\multicolumn{1}{r}{6.5} & \multicolumn{1}{r}{80.6} &  &
\multicolumn{1}{r}{1.9} & \multicolumn{1}{r}{70.6} &  &
\multicolumn{1}{r}{1.3} & \multicolumn{1}{r}{67.2}\\
& $N=200$ &  & \multicolumn{1}{r}{13.2} & \multicolumn{1}{r}{79.0} &  &
\multicolumn{1}{r}{7.9} & \multicolumn{1}{r}{75.0} &  &
\multicolumn{1}{r}{2.1} & \multicolumn{1}{r}{64.7} &  &
\multicolumn{1}{r}{1.6} & \multicolumn{1}{r}{62.1}\\
& $N=300$ &  & \multicolumn{1}{r}{14.5} & \multicolumn{1}{r}{77.0} &  &
\multicolumn{1}{r}{7.4} & \multicolumn{1}{r}{71.6} &  &
\multicolumn{1}{r}{2.1} & \multicolumn{1}{r}{62.1} &  &
\multicolumn{1}{r}{1.6} & \multicolumn{1}{r}{60.2}\\
\multicolumn{5}{l}{$q=0.05$} &  & \multicolumn{1}{r}{} & \multicolumn{1}{r}{}
&  & \multicolumn{1}{r}{} & \multicolumn{1}{r}{} &  & \multicolumn{1}{r}{} &
\multicolumn{1}{r}{}\\
& $N=50$ &  & \multicolumn{1}{r}{5.2} & \multicolumn{1}{r}{86.5} &  &
\multicolumn{1}{r}{2.1} & \multicolumn{1}{r}{76.3} &  &
\multicolumn{1}{r}{1.0} & \multicolumn{1}{r}{74.7} &  &
\multicolumn{1}{r}{0.7} & \multicolumn{1}{r}{71.6}\\
& $N=100$ &  & \multicolumn{1}{r}{5.5} & \multicolumn{1}{r}{79.4} &  &
\multicolumn{1}{r}{2.9} & \multicolumn{1}{r}{74.5} &  &
\multicolumn{1}{r}{0.8} & \multicolumn{1}{r}{63.5} &  &
\multicolumn{1}{r}{0.6} & \multicolumn{1}{r}{60.5}\\
& $N=200$ &  & \multicolumn{1}{r}{7.5} & \multicolumn{1}{r}{74.7} &  &
\multicolumn{1}{r}{3.6} & \multicolumn{1}{r}{68.8} &  &
\multicolumn{1}{r}{0.9} & \multicolumn{1}{r}{57.6} &  &
\multicolumn{1}{r}{0.8} & \multicolumn{1}{r}{56.2}\\
& $N=300$ &  & \multicolumn{1}{r}{8.3} & \multicolumn{1}{r}{72.8} &  &
\multicolumn{1}{r}{3.2} & \multicolumn{1}{r}{64.0} &  &
\multicolumn{1}{r}{0.8} & \multicolumn{1}{r}{54.9} &  &
\multicolumn{1}{r}{0.8} & \multicolumn{1}{r}{54.5}\\\hline
\end{tabular}

	\end{table}

 \setlength{\tabcolsep}{3.3pt}%
\renewcommand{\arraystretch}{1}
	\begin{table}[H]
		\caption{FDR and power using asymptotic, bootstrap and e-BH thresholds
			for $q=0.10$, where $\sigma_{ij}=\rho_u^{|i-j|}$ with $\rho_u=0.8$}
		\label{table:fdr_het7}
		\centering
\begin{tabular}{cccccccccccccc}
\hline
\multicolumn{14}{c}{$m=2$ ($\max s_{i}=5$), $T=200$, Mixture normal,
cross-correlated errors}\\\hline
&  &  & \multicolumn{2}{c}{asymptotic} &  & \multicolumn{2}{c}{bootstrap} &  &
\multicolumn{2}{c}{$|x|^{p}$} &  & \multicolumn{2}{c}{$\exp(c|x|^{\alpha})$%
}\\\cline{4-5}\cline{7-8}\cline{10-11}\cline{13-14}
&  &  & dFDR & PWR &  & dFDR & PWR &  & FDR & PWR &  & FDR & PWR\\
\multicolumn{5}{c}{$q=0.10$} &  &  &  &  &  &  &  &  &  \\ 
& $N=50$ &  & 10.9 & 71.8 &  & 10.3 & 71.2 &  & 2.1 & 53.5 &  & 1.4 & 49.0
\\ 
& $N=100$ &  & 13.0 & 66.3 &  & 11.6 & 65.2 &  & 2.4 & 48.5 &  & 1.7 & 44.7
\\ 
& $N=200$ &  & 16.3 & 61.7 &  & 15.1 & 60.9 &  & 2.8 & 45.7 &  & 2.3 & 43.6
\\ 
& $N=300$ &  & 17.6 & 63.1 &  & 15.4 & 61.8 &  & 3.0 & 47.7 &  & 2.5 & 46.4
\\ \hline
\end{tabular}

	\end{table}

  \setlength{\tabcolsep}{3.3pt}%
\renewcommand{\arraystretch}{1}
	\begin{table}[H]
		\caption{FDR and power using asymptotic, bootstrap and e-BH thresholds
			for $q=0.10$, with a factor model, $
u_{ti}=(f_{t}b_{i}+v_{ti})/\sqrt{\sigma_{f}^{2}+1}$}
		\label{table:fdr_het6}
		\centering
\begin{tabular}{cccccccccccccc}
\hline
\multicolumn{14}{c}{$m=2$ ($\max s_{i}=5$), $T=200$, Mixture normal,
cross-correlated errors}\\\hline
&  &  & \multicolumn{2}{c}{asymptotic} &  & \multicolumn{2}{c}{bootstrap} &  &
\multicolumn{2}{c}{$|x|^{p}$} &  & \multicolumn{2}{c}{$\exp(c|x|^{\alpha})$%
}\\\cline{4-5}\cline{7-8}\cline{10-11}\cline{13-14}
&  &  & dFDR & PWR &  & dFDR & PWR &  & FDR & PWR &  & FDR & PWR\\
\multicolumn{5}{c}{$q=0.10$} &  &  &  &  &  &  &  &  &  \\ 
& $N=50$ &  & 10.6 & 93.3 &  & 7.5 & 91.8 &  & 2.2 & 85.2 &  & 1.6 & 82.9 \\ 
& $N=100$ &  & 14.5 & 87.7 &  & 10.0 & 85.2 &  & 3.1 & 76.6 &  & 2.1 & 73.7
\\ 
& $N=200$ &  & 23.3 & 83.6 &  & 15.6 & 80.3 &  & 4.8 & 70.7 &  & 3.6 & 68.2
\\ 
& $N=300$ &  & 30.3 & 81.1 &  & 17.1 & 75.4 &  & 5.8 & 66.4 &  & 4.6 & 64.3
\\ \hline
\end{tabular}

	\end{table}

	\setcounter{table}{0}
	\renewcommand{\thetable}{\thesection\arabic{table}}
	\setcounter{figure}{0}
	\renewcommand{\thefigure}{\thesection\arabic{figure}}
	\section{List of Variable Names in Empirical Applications}
	\setlength{\tabcolsep}{3.3pt}%
		\renewcommand\arraystretch{0.55}
		\begin{table}[H]
			\caption{Full list of FRED-MD macroeconomic variables.}	
			\label{table:fredmdv}
			\centering
			\scriptsize
			%
			\begin{tabular}
				[c]{|r|l|rr}\hline
				\multicolumn{4}{|c|}{\textbf{FREDMD ID and Variable Descriptions}}\\\hline
				\multicolumn{2}{|c|}{Group 1: Output and Income} & \multicolumn{2}{c|}{Group
					5: Interest and Exchange Rate}\\\hline
				1 & Real Personal Income & \multicolumn{1}{r|}{68} &
				\multicolumn{1}{l|}{Federal Funds Rate}\\
				2 & Real personal income ex transfer receipts & \multicolumn{1}{r|}{69} &
				\multicolumn{1}{l|}{3-Month AA Financial Commercial Paper Rate}\\
				3 & IP Index & \multicolumn{1}{r|}{70} & \multicolumn{1}{l|}{3-Month Treasury
					Bill}\\
				4 & IP: Final Products and Nonindustrial Supplies & \multicolumn{1}{r|}{71} &
				\multicolumn{1}{l|}{6-Month Treasury Bill}\\
				5 & IP: Final Products (Market Group) & \multicolumn{1}{r|}{72} &
				\multicolumn{1}{l|}{1-Year Treasury Rate}\\
				6 & IP: Consumer Goods & \multicolumn{1}{r|}{73} & \multicolumn{1}{l|}{5-Year
					Treasury Rate}\\
				7 & IP: Durable Consumer Goods & \multicolumn{1}{r|}{74} &
				\multicolumn{1}{l|}{10-Year Treasury Rate}\\
				8 & IP: Nondurable Consumer Goods & \multicolumn{1}{r|}{75} &
				\multicolumn{1}{l|}{Moody's Seasoned Aaa Corporate Bond Yield Aaa bond}\\
				9 & IP: Business Equipment & \multicolumn{1}{r|}{76} &
				\multicolumn{1}{l|}{Moody's Seasoned Baa Corporate Bond Yield Baa bond}\\
				10 & IP: Materials & \multicolumn{1}{r|}{77} & \multicolumn{1}{l|}{3-Month
					Commercial Paper Minus FEDFUNDS CP-FF spread}\\
				11 & IP: Durable Materials & \multicolumn{1}{r|}{78} &
				\multicolumn{1}{l|}{3-Month Treasury C Minus FEDFUNDS 3 mo-FF spread}\\
				12 & IP: Nondurable Materials & \multicolumn{1}{r|}{79} &
				\multicolumn{1}{l|}{6-Month Treasury C Minus FEDFUNDS 6 mo-FF spread}\\
				13 & IP: Manufacturing (SIC) & \multicolumn{1}{r|}{80} &
				\multicolumn{1}{l|}{1-Year Treasury C Minus FEDFUNDS 1 yr-FF spread}\\
				14 & IP: Residential Utilities & \multicolumn{1}{r|}{81} &
				\multicolumn{1}{l|}{5-Year Treasury C Minus FEDFUNDS 5 yr-FF spread}\\
				15 & IP: Fuels & \multicolumn{1}{r|}{82} & \multicolumn{1}{l|}{10-Year
					Treasury C Minus FEDFUNDS 10 yr-FF spread}\\
				16 & Capacity Utilization: Manufacturing & \multicolumn{1}{r|}{83} &
				\multicolumn{1}{l|}{Moody's Aaa Corporate Bond Minus FEDFUNDS}\\\cline{1-2}%
				\multicolumn{2}{|c|}{Group 2: Labor Market} & \multicolumn{1}{r|}{84} &
				\multicolumn{1}{l|}{Moody's Baa Corporate Bond Minus FEDFUNDS}\\\cline{1-2}%
				17 & Help-Wanted Index for United States & \multicolumn{1}{r|}{85} &
				\multicolumn{1}{l|}{Trade Weighted U.S. Dollar Index: Major Currencies Ex
					rate}\\
				18 & Ratio of Help Wanted/No. Unemployed & \multicolumn{1}{r|}{86} &
				\multicolumn{1}{l|}{Switzerland / U.S. Foreign Exchange Rate}\\
				19 & Civilian Labor Force & \multicolumn{1}{r|}{87} &
				\multicolumn{1}{l|}{Japan / U.S. Foreign Exchange Rate}\\
				20 & Civilian Employment & \multicolumn{1}{r|}{88} & \multicolumn{1}{l|}{U.S.
					/ U.K. Foreign Exchange Rate}\\
				21 & Civilian Unemployment Rate & \multicolumn{1}{r|}{89} &
				\multicolumn{1}{l|}{Canada / U.S. Foreign Exchange Rate}\\\cline{3-4}%
				22 & Average Duration of Unemployment (Weeks) & \multicolumn{2}{c|}{Group 6:
					Prices}\\\cline{3-4}%
				23 & Civilians Unemployed -Less Than 5 Weeks & \multicolumn{1}{r|}{90} &
				\multicolumn{1}{l|}{PPI: Finished Goods}\\
				24 & Civilians Unemployed for 5-14 Weeks & \multicolumn{1}{r|}{91} &
				\multicolumn{1}{l|}{PPI: Finished Consumer Goods}\\
				25 & Civilians Unemployed -15 Weeks \& Over & \multicolumn{1}{r|}{92} &
				\multicolumn{1}{l|}{PPI: Intermediate Materials}\\
				26 & Civilians Unemployed for 15-26 Weeks & \multicolumn{1}{r|}{93} &
				\multicolumn{1}{l|}{PPI: Crude Materials}\\
				27 & Civilians Unemployed for 27 Weeks and Over & \multicolumn{1}{r|}{94} &
				\multicolumn{1}{l|}{Crude Oil}\\
				28 & Initial Claims & \multicolumn{1}{r|}{95} & \multicolumn{1}{l|}{PPI:
					Metals and metal products:}\\
				29 & All Employees: Total nonfarm & \multicolumn{1}{r|}{96} &
				\multicolumn{1}{l|}{CPI : All Items}\\
				30 & All Employees: Goods-Producing Industries & \multicolumn{1}{r|}{97} &
				\multicolumn{1}{l|}{CPI : Apparel}\\
				31 & All Employees: Mining and Logging: Mining & \multicolumn{1}{r|}{98} &
				\multicolumn{1}{l|}{CPI : Transportation}\\
				32 & All Employees: Construction & \multicolumn{1}{r|}{99} &
				\multicolumn{1}{l|}{CPI : Medical Care}\\
				33 & All Employees: Manufacturing & \multicolumn{1}{r|}{100} &
				\multicolumn{1}{l|}{CPI : Commodities}\\
				34 & All Employees: Durable goods & \multicolumn{1}{r|}{101} &
				\multicolumn{1}{l|}{CPI : Durables}\\
				35 & All Employees: Nondurable goods & \multicolumn{1}{r|}{102} &
				\multicolumn{1}{l|}{CPI : Services}\\
				36 & All Employees: Service-Providing Industries & \multicolumn{1}{r|}{103} &
				\multicolumn{1}{l|}{CPI : All Items Less Food}\\
				37 & All Employees: Trade & \multicolumn{1}{r|}{104} & \multicolumn{1}{l|}{CPI
					: All items less shelter}\\
				38 & All Employees: Wholesale Trade & \multicolumn{1}{r|}{105} &
				\multicolumn{1}{l|}{CPI : All items less medical care}\\
				39 & All Employees: Retail Trade & \multicolumn{1}{r|}{106} &
				\multicolumn{1}{l|}{Personal Cons. Expend.: Chain Index}\\
				40 & All Employees: Financial Activities & \multicolumn{1}{r|}{107} &
				\multicolumn{1}{l|}{Personal Cons. Exp: Durable goods}\\
				41 & All Employees: Government & \multicolumn{1}{r|}{108} &
				\multicolumn{1}{l|}{Personal Cons. Exp: Nondurable goods}\\
				42 & Avg Weekly Hours : Goods-Producing & \multicolumn{1}{r|}{109} &
				\multicolumn{1}{l|}{Personal Cons. Exp: Services}\\\cline{3-4}%
				43 & Avg Weekly Overtime Hours : Manufacturing & \multicolumn{2}{c|}{Group 7:
					Money and Credit}\\\cline{3-4}%
				44 & Avg Weekly Hours : Manufacturing & \multicolumn{1}{r|}{110} &
				\multicolumn{1}{l|}{M1 Money Stock}\\
				45 & Avg Hourly Earnings : Goods-Producing & \multicolumn{1}{r|}{111} &
				\multicolumn{1}{l|}{M2 Money Stock}\\
				46 & Avg Hourly Earnings : Construction & \multicolumn{1}{r|}{112} &
				\multicolumn{1}{l|}{Real M2 Money Stock}\\
				47 & Avg Hourly Earnings : Manufacturing & \multicolumn{1}{r|}{113} &
				\multicolumn{1}{l|}{St. Louis Adjusted Monetary Base}\\\cline{1-2}%
				\multicolumn{2}{|c|}{Group 3: Consumption, Orders and Inventories} &
				\multicolumn{1}{r|}{114} & \multicolumn{1}{l|}{Total Reserves of Depository
					Institutions}\\\cline{1-2}%
				48 & Real personal consumption expenditures & \multicolumn{1}{r|}{115} &
				\multicolumn{1}{l|}{Reserves Of Depository Institutions}\\
				49 & Real Manu. and Trade Industries Sales & \multicolumn{1}{r|}{116} &
				\multicolumn{1}{l|}{Commercial and Industrial Loans}\\
				50 & Retail and Food Services Sales & \multicolumn{1}{r|}{117} &
				\multicolumn{1}{l|}{Real Estate Loans at All Commercial Banks}\\
				51 & New Orders for Consumer Goods & \multicolumn{1}{r|}{118} &
				\multicolumn{1}{l|}{Total Nonrevolving Credit}\\
				52 & New Orders for Durable Goods & \multicolumn{1}{r|}{119} &
				\multicolumn{1}{l|}{Nonrevolving consumer credit to Personal Income}\\
				53 & New Orders for Nondefense Capital Goods & \multicolumn{1}{r|}{120} &
				\multicolumn{1}{l|}{MZM Money Stock}\\
				54 & Un.lled Orders for Durable Goods & \multicolumn{1}{r|}{121} &
				\multicolumn{1}{l|}{Consumer Motor Vehicle Loans Outstanding}\\
				55 & Total Business Inventories & \multicolumn{1}{r|}{122} &
				\multicolumn{1}{l|}{Total Consumer Loans and Leases Outstanding}\\
				56 & Total Business: Inventories to Sales Ratio & \multicolumn{1}{r|}{123} &
				\multicolumn{1}{l|}{Securities in Bank Credit at All Commercial Banks}%
				\\\cline{3-4}%
				57 & Consumer Sentiment Index & \multicolumn{2}{c|}{Group 8: Stock
					Market}\\\hline
				\multicolumn{2}{|c|}{Group 4: Housing} & \multicolumn{1}{r|}{124} &
				\multicolumn{1}{l|}{S\&P's Common Stock Price Index: Composite}\\\cline{1-2}%
				58 & Housing Starts: Total New Privately Owned & \multicolumn{1}{r|}{125} &
				\multicolumn{1}{l|}{S\&P's Common Stock Price Index: Industrials}\\
				59 & Housing Starts Northeast & \multicolumn{1}{r|}{126} &
				\multicolumn{1}{l|}{S\&P's Composite Common Stock: Dividend Yield}\\
				60 & Housing Starts Midwest & \multicolumn{1}{r|}{127} &
				\multicolumn{1}{l|}{S\&P's Composite Common Stock: Price-Earnings Ratio}\\
				61 & Housing Starts South & \multicolumn{1}{r|}{128} &
				\multicolumn{1}{l|}{VXO}\\\cline{3-4}%
				62 & Housing Starts West &  & \\
				63 & New Private Housing Permits (SAAR) &  & \\
				64 & New Private Housing Permits Northeast &  & \\
				65 & New Private Housing Permits Midwest &  & \\
				66 & New Private Housing Permits South &  & \\
				67 & New Private Housing Permits West &  & \\\cline{1-2}%
			\end{tabular}
		\end{table}
  
		\normalsize
		\renewcommand\arraystretch{.55}
		
		\setlength{\tabcolsep}{3.3pt}%
			\begin{table}[H]
				\caption{Full list of the Local Area Districts (LADs) for the UK regional house price}	
				\label{table:ukhpv}
				\centering
				\begin{tabular}
					[c]{rrrrlrrr}\hline
					\multicolumn{2}{c}{Scotland} &  & \multicolumn{2}{c}{London Area} &  &
					\multicolumn{2}{c}{Wales}\\\cline{1-2}\cline{4-5}\cline{7-8}%
					\multicolumn{1}{c}{ID} & \multicolumn{1}{c}{Region Names} &  &
					\multicolumn{1}{c}{ID} & \multicolumn{1}{c}{Region Names} &  &
					\multicolumn{1}{c}{ID} & \multicolumn{1}{c}{Region Names}\\\hline
					1 & \multicolumn{1}{l}{Clackmannanshire} &  & 32 & City.of.London &  & 65 &
					\multicolumn{1}{l}{Isle.of.Anglesey}\\
					2 & \multicolumn{1}{l}{Dumfries.and.Galloway} &  & 33 & Barking.and.Dagenham &
					& 66 & \multicolumn{1}{l}{Gwynedd}\\
					3 & \multicolumn{1}{l}{East.Ayrshire} &  & 34 & Barnet &  & 67 &
					\multicolumn{1}{l}{Conwy}\\
					4 & \multicolumn{1}{l}{East.Lothian} &  & 35 & Bexley &  & 68 &
					\multicolumn{1}{l}{Denbighshire}\\
					5 & \multicolumn{1}{l}{Na.h-Eileanan.Siar} &  & 36 & Brent &  & 69 &
					\multicolumn{1}{l}{Flintshire}\\
					6 & \multicolumn{1}{l}{Falkirk} &  & 37 & Bromley &  & 70 &
					\multicolumn{1}{l}{Wrexham}\\
					7 & \multicolumn{1}{l}{Highland} &  & 38 & Camden &  & 71 &
					\multicolumn{1}{l}{Ceredigion}\\
					8 & \multicolumn{1}{l}{Inverclyde} &  & 39 & Croydon &  & 72 &
					\multicolumn{1}{l}{Pembrokeshire}\\
					9 & \multicolumn{1}{l}{Midlothian} &  & 40 & Ealing &  & 73 &
					\multicolumn{1}{l}{Carmarthenshire}\\
					10 & \multicolumn{1}{l}{Moray} &  & 41 & Enfield &  & 74 &
					\multicolumn{1}{l}{Swansea}\\
					11 & \multicolumn{1}{l}{North.Ayrshire} &  & 42 & Greenwich &  & 75 &
					\multicolumn{1}{l}{Neath.Port.Talbot}\\
					12 & \multicolumn{1}{l}{Orkney.Islands} &  & 43 & Hackney &  & 76 &
					\multicolumn{1}{l}{Bridgend}\\
					13 & \multicolumn{1}{l}{Scottish.Borders} &  & 44 & Hammersmith.and.Fulham &
					& 77 & \multicolumn{1}{l}{Vale.of.Glamorgan}\\
					14 & \multicolumn{1}{l}{Shetland.Islands} &  & 45 & Haringey &  & 78 &
					\multicolumn{1}{l}{Cardiff}\\
					15 & \multicolumn{1}{l}{South.Ayrshire} &  & 46 & Harrow &  & 79 &
					\multicolumn{1}{l}{Rhondda.Cynon.Taf}\\
					16 & \multicolumn{1}{l}{South.Lanarkshire} &  & 47 & Havering &  & 80 &
					\multicolumn{1}{l}{Caerphilly}\\
					17 & \multicolumn{1}{l}{Stirling} &  & 48 & Hillingdon &  & 81 &
					\multicolumn{1}{l}{Blaenau.Gwent}\\
					18 & \multicolumn{1}{l}{Aberdeen.City} &  & 49 & Hounslow &  & 82 &
					\multicolumn{1}{l}{Torfaen}\\
					19 & \multicolumn{1}{l}{Aberdeenshire} &  & 50 & Islington &  & 83 &
					\multicolumn{1}{l}{Monmouthshire}\\
					20 & \multicolumn{1}{l}{Argyll.and.Bute} &  & 51 & Kensington.and.Chelsea &  &
					84 & \multicolumn{1}{l}{Newport}\\
					21 & \multicolumn{1}{l}{City.of.Edinburgh} &  & 52 & Kingston.upon.Thames &  &
					85 & \multicolumn{1}{l}{Powys}\\
					22 & \multicolumn{1}{l}{Renfrewshire} &  & 53 & Lambeth &  & 86 &
					\multicolumn{1}{l}{Merthyr.Tydfil}\\
					23 & \multicolumn{1}{l}{West.Dunbartonshire} &  & 54 & Lewisham &  &  & \\
					24 & \multicolumn{1}{l}{West.Lothian} &  & 55 & Merton &  &  & \\
					25 & \multicolumn{1}{l}{Angus} &  & 56 & Newham &  &  & \\
					26 & \multicolumn{1}{l}{Dundee.City} &  & 57 & Redbridge &  &  & \\
					27 & \multicolumn{1}{l}{East.Dunbartonshire} &  & 58 & Richmond.upon.Thames &
					&  & \\
					28 & \multicolumn{1}{l}{Fife} &  & 59 & Southwark &  &  & \\
					29 & \multicolumn{1}{l}{Perth.and.Kinross} &  & 60 & Sutton &  &  & \\
					30 & \multicolumn{1}{l}{Glasgow.City} &  & 61 & Tower.Hamlets &  &  & \\
					31 & \multicolumn{1}{l}{North.Lanarkshire} &  & 62 & Waltham.Forest &  &  & \\
					&  &  & 63 & Wandsworth &  &  & \\
					&  &  & 64 & Westminster &  &  & \\\hline
				\end{tabular}
			\end{table}

			\setcounter{table}{0}
			\renewcommand{\thetable}{H\arabic{table}}
			\setcounter{figure}{0}
			\renewcommand{\thefigure}{H\arabic{figure}}
			
			\newpage
			\section{Additional Results of Empirical Applications}
			\noindent
			\begin{figure}[H]
				\begin{minipage}{1\hsize}
				\centering
				\includegraphics[width=120mm]
				{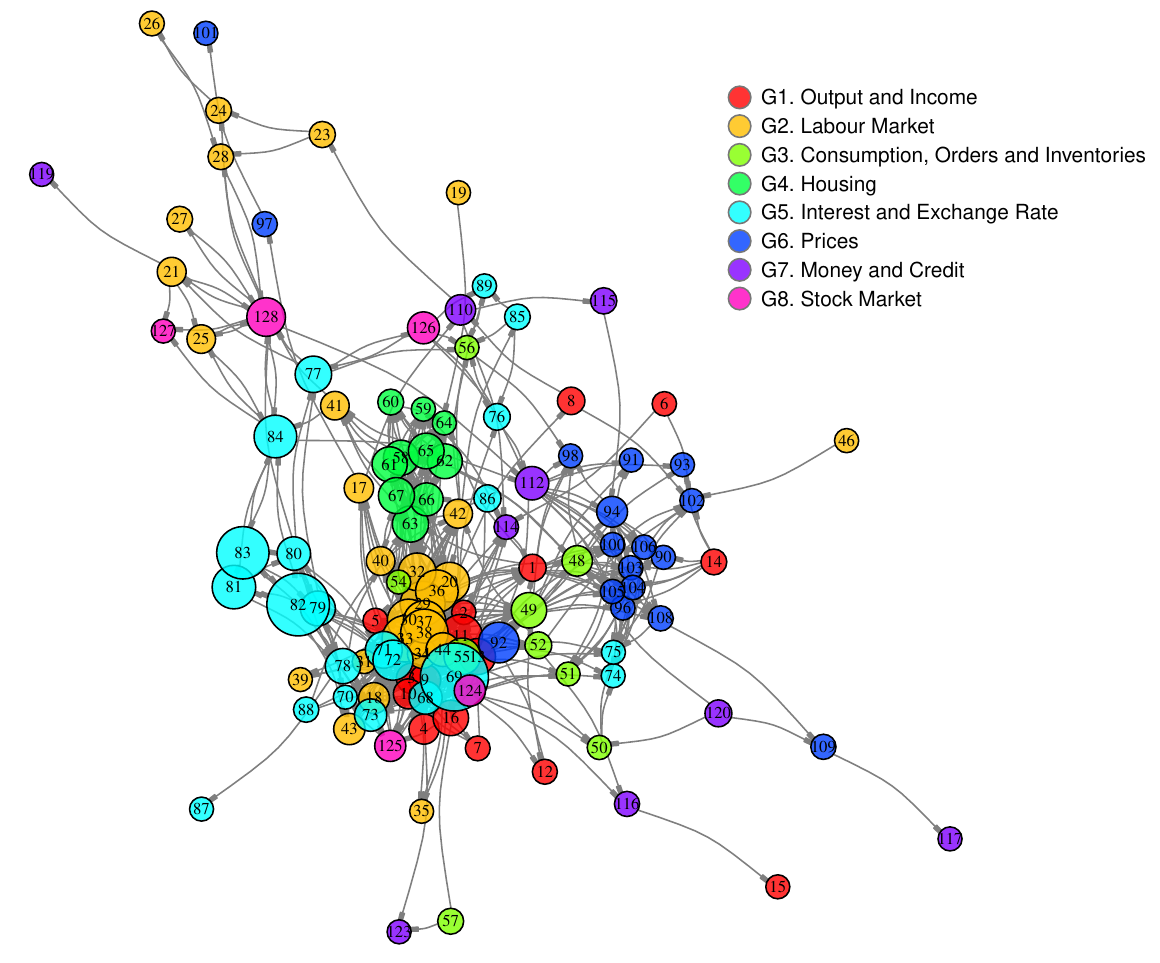}
				\vspace{0mm}
				\caption{\footnotesize{Dynamic network visualization: macroeconomic variables, asymptotic $t_0$, $q=0.05$}}
				\label{figure:nwFREDMDq05asy}
				\vspace{-3mm}
			\end{minipage}
		\end{figure}
		
		
		\noindent
		\begin{figure}[H]
			\begin{minipage}{1\hsize}
			\centering
			\includegraphics[width=120mm]
			{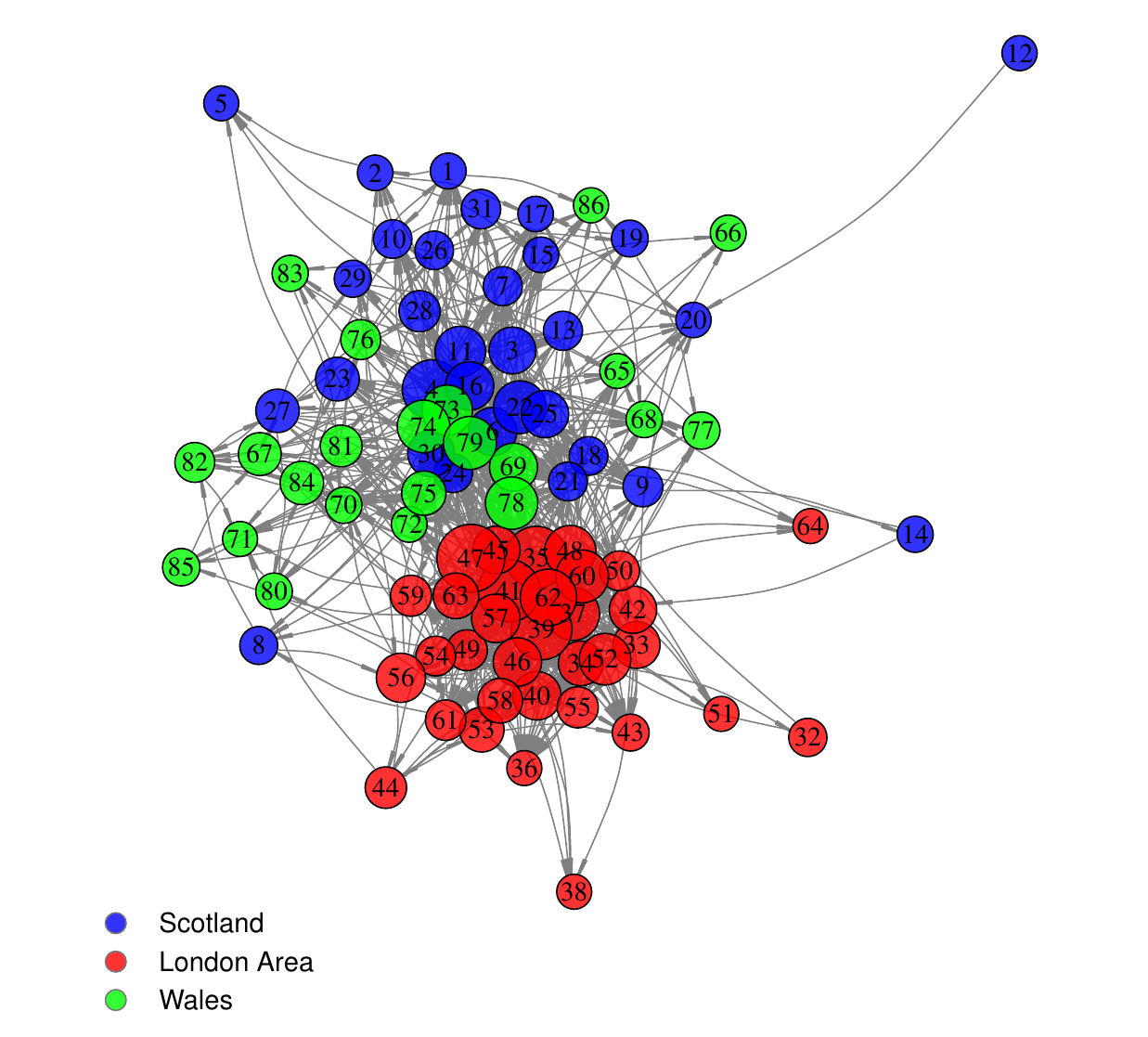}
			\vspace{0mm}
			\caption{\footnotesize{Dynamic network visualization: UK regional house prices, asymptotic $t_0$, $q=0.05$}}
			\label{figure:nwUKHPq05asy}
			\vspace{-3mm}
		\end{minipage}
	\end{figure}


\end{document}